\documentclass[aps,prb,reprint,twocolumn,floats,floatfix,superscriptaddress,nofootinbib]{revtex4-2}

\usepackage{graphicx}
\usepackage{amsmath}
\usepackage{amsfonts}
\usepackage{amssymb}
\usepackage{amstext}
\usepackage[english]{babel}
\usepackage{helvet}
\usepackage{microtype}
\usepackage[pdftex,hidelinks]{hyperref}
\hypersetup{
    colorlinks=true,
    linkcolor=blue,
    filecolor=blue,
    urlcolor=blue,
    citecolor = blue,
}
\usepackage{xcolor}
\usepackage[export]{adjustbox}
\usepackage{bbm}
\usepackage{manfnt}
\usepackage[caption=false]{subfig}
\usepackage{multirow}
\usepackage{dsfont}
\usepackage{amsthm}

\graphicspath{{./images/}}

\newcommand{\ket}[1]{|#1\rangle}
\newcommand{\bra}[1]{\langle#1|}
\newcommand{\braket}[2]{\langle#1|#2\rangle}
\DeclareMathOperator{\tr}{Tr}
\DeclareMathOperator{\diag}{diag}
\DeclareMathOperator{\spec}{Spectrum}
\newcommand{\texteq}[1]{\stackrel{\text{#1}}{=}}
\newcommand{\Mnn}[1]{\mathcal{M}^{(#1)}_{nn}(\mathbb{R})}
\newcommand{\MnnF}[0]{\mathcal{M}_{nn}(\mathbb{R})}
\newcommand\cube[1][.6]{\mathbin{\vcenter{\hbox{\scalebox{#1}{\mbox{\mancube}}}}}}
\newcommand{\Tsp}[2]{\mathfrak{T}_{#1}^{#2}}

\newtheorem{theorem}{Theorem}
\newtheorem{corollary}{Corollary}[theorem]
\theoremstyle{remark}
\newtheorem*{remark}{Remark}

\begin{document}

\title{
    Lattice-reflection symmetry in tensor-network renormalization group with entanglement filtering in two and three dimensions
}


\author{Xinliang Lyu}
\email[]{xlyu@ihes.fr}
%\thanks{}
%\altaffiliation{Present address:
%    Institut des Hautes \'Etudes Scientifiques,
%    91440 Bures-sur-Yvette, France
%}
\affiliation{
    Institute for Solid State Physics, The University of Tokyo,
    Kashiwa, Chiba 277-8581, Japan
}
\affiliation{
    Institut des Hautes \'Etudes Scientifiques,
    91440 Bures-sur-Yvette, France
}
% Last author
\author{Naoki Kawashima}
\email[]{kawashima@issp.u-tokyo.ac.jp} 
\affiliation{
    Institute for Solid State Physics, The University of Tokyo, 
    Kashiwa, Chiba 277-8581, Japan
}
\affiliation{
    Trans-scale Quantum Science Institute, The University of Tokyo
    7-3-1, Hongo, Tokyo 113-0033, Japan
}

%\date{\today}
\date{April 1, 2026}

\begin{abstract}
    Tensor-network renormalization group (TNRG) is an efficient real-space renormalization group method for studying the criticality in both classical and quantum lattice systems.
    Exploiting symmetries of a system in a TNRG algorithm can simplify the implementation of the algorithm and can help produce correct tensor RG flows.
    Although a general framework for considering a global on-site symmetry has been established, it is still unclear how to incorporate a lattice symmetry in TNRG.
    As a first step for lattice symmetries, we propose a method to incorporate the lattice-reflection symmetry in the context of a TNRG with entanglement filtering in both two and three dimensions (2D and 3D).
    To achieve this, we write down a general definition of lattice-reflection symmetry in tensor-network language.
    Then, we introduce a transposition trick for exploiting and imposing the lattice-reflection symmetry in two basic TNRG operations: projective truncations and entanglement filtering.
    Using the transposition trick, the detailed algorithms of the TNRG map in both 2D and 3D are laid out, where the lattice-reflection symmetry is preserved and imposed.
    Finally, we demonstrate how to construct the linearization of the TNRG maps in a given lattice-reflection sector, with the help of which it becomes possible to extract scaling dimensions in each sector separately.
    Our work paves the way for understanding the lattice-rotation symmetry in TNRG.
\end{abstract}

\maketitle

\newpage

\tableofcontents

% Text body starting here
\section{Introduction}
Tensor-network renormalization group (TNRG)~\cite{Levin:Nave:2007,Okunishi:Nishino:2022} is a modern formulation of Kadanoff’s real-space renormalization group (RG) idea~\cite{Kadanoff:1966} in a tensor-network language.
It is a powerful numerical RG method for both classical and quantum lattice systems.
Unlike the real-space RG formulated in the spin representation, whose approximations are uncontrolled~\cite{Kardar:2007}, the TNRG is naturally equipped with a measure of RG errors~\cite{Levin:Nave:2007,Evenbly:2017:algo}.
These errors are controlled by a positive integer $\chi \in \mathbb{Z}^+$ called bond dimension, corresponding to the number of coupling constants kept in an RG map.
In the tensor-network language, Kadanoff’s block-spin transformation becomes coarse graining of tensor blocks, which will be referred to as \emph{block-tensor} transformation in this paper.
The RG approximation errors of a block-tensor transformation depend on the scaling of the entanglement entropy of a system~\cite{Levin:Nave:2007,Lyu:Kawashima:2023}.
When an entanglement filtering (EF)~\cite{Gu:Wen:2009,Evenbly:Vidal:2015} is integrated into a block-tensor transformation, the TNRG is able to exhibit critical fixed points in both two and three dimensions (2D and 3D)~\cite{Evenbly:Vidal:2015,Lyu:Kawashima:2024}.
A systematically improvable real-space RG has been realized in 2D TNRG~\cite{Evenbly:Vidal:2015}, while in 3D, the TNRG can produce reliable estimates of scaling dimensions~\cite{HOTRG:2012,Lyu:Kawashima:2024}.

In order to produce correct RG flows, it is important to incorporate into the TNRG the symmetries of a lattice model that are essential to the nature of its phase transition.
This is because all RG-relevant and marginal operators that are not in the symmetry sector can be eliminated in numerical calculations.
This means that incorporating a proper symmetry in TNRG facilitates the estimation of the critical parameter of a model and, thus, the isolation of a critical fixed-point tensor.
For example, the spin-flip $\mathbb{Z}_2$ symmetry is essential for the second-order phase transition of the Ising model.
Without incorporating this symmetry in TNRG, the low-temperature (low-$T$) fixed point becomes unstable under an RG map and will eventually flow to the high-temperature (high-$T$) fixed point due to the perturbation corresponding to the spin operator~\cite{Gu:Wen:2009}.
Imposing the spin-flip $\mathbb{Z}_2$ symmetry in TNRG can eliminate this perturbation, making the low-$T$ fixed point of the Ising model stable.

Furthermore, exploiting lattice symmetries, like reflection and rotation, can simplify the implementation of the TNRG algorithms. 
For example, it is reasonable to expect that the same piece of tensor in one direction can be used in other directions due to these lattice symmetries.
This reduction of the number of undetermined tensors in the implementation of a TNRG has been shown to be quite helpful in a recent algorithm of the 3D TNRG enhanced by EF~\cite{Lyu:Kawashima:2024}.
However, it is not proven in Ref.~\cite{Lyu:Kawashima:2024} regarding why exploiting the lattice symmetry leads to such simplification.
We will show in~\autoref{subsec:symFM} that the number of the filtering matrix in EF reduces from 24 to 3 by exploiting the lattice-reflection symmetry in 3D.

A general framework for incorporating a global on-site symmetry, like the spin-flip $\mathbb{Z}_2$ symmetry of the Ising model, has been established for tensor-network decompositions and contractions~\cite{Singh:Pfeifer:2010,Singh:Pfeifer:2011,Singh:Vidal:2012}.
As for lattice symmetries, it seems that some practitioners of TNRG know, maybe based on their intuition, how to incorporate lattice-reflection and rotation symmetry on a case-by-case basis.
Without demonstrating why a lattice symmetry can be preserved under the RG, in the study of tensor network renormalization (TNR)~\cite{Evenbly:2017:algo}, Evenbly briefly discussed the lattice-reflection symmetry; 
while in the study of loop-TNR~\cite{Yang:Gu:Wen:2017}, a claim was made about a method for incorporating the lattice-rotation symmetry.
In the context of the lattice Schwinger model, both the lattice-reflection and rotation symmetries are incorporated in a specially designed Grassmann version of tensor network renormalization (TRG) with decorations~\cite{Shimizu:Kuramashi:2018}.
However, there is no systematic and general discussion that can serve as a basis for exploiting lattice symmetries in the development of a new scheme, for example, in 3D.

The purpose of this paper is to offer a general framework for understanding the lattice-reflection symmetry in TNRG with EF.
We focus on real-valued tensors.
The framework works equally well in both the 2D square-lattice and the 3D cubic-lattice tensor networks.
Our method is based on a general truncation scheme in TNRG called projective truncations\footnote{
Most of the existing TNRG algorithms can be reformulated using projective truncations.
}~\cite{Evenbly:2017:algo}.
We explain the origin of a general definition of lattice-reflection symmetry proposed by Evenbly~\cite{Evenbly:2017:algo}, and prove how this symmetry can be preserved and imposed under a TNRG transformation.
A new technique is invented for the proof regarding lattice symmetries in tensor network;
this technique involves dragging the tensors around in a tensor-network diagram and comparing the resultant diagram with the original one.
The notion of a SWAP-gauge matrix, which appears naturally in the definition of the lattice-reflection symmetry in TNRG, can help understand a recent construction of non-orientable surfaces in tensor network, like the crosscap and rainbow boundaries~\cite{Shimizu:2024}.
This framework is employed in the design of a recent EF-enhanced TNRG in 3D~\cite{Lyu:Kawashima:2024}.
Based on this framework, the lattice-rotation symmetry will be studied in a coming paper.

The remaining part of the paper is organized as follows. 
In~\autoref{sec:projtrunc}, we review projective truncations and how to use them to implement a simple block-tensor transformation. 
Then, we introduce a graph-independent EF scheme whose formulation is conducive to incorporating lattice symmetries in~\autoref{sec:GIEF}.
In~\autoref{sec:lattsym}, we demonstrate the key technical ingredients for exploiting the lattice-reflection symmetry in TNRG, including its definition, how to use a transposition trick to impose this symmetry, and the symmetry properties of various kinds of tensors in projective truncations and the EF.
The 2D and 3D algorithms of the EF-enhanced TNRG that preserves the lattice-reflection symmetry are proposed in~\autoref{sec:algo}.
The construction of the linearized RG in different lattice-reflection sectors is expounded in~\autoref{sec:linRG}.
As a numerical demonstration, in~\autoref{sec:numdemo}, we apply the proposed schemes to the Ising model and extract its scaling dimensions from the linearized RG map in different lattice-reflection sectors.

\section{Projective truncations\label{sec:projtrunc}}
In this section, we review a numerical technique called \emph{projective truncations}~\cite{Evenbly:2017:algo} for implementing block-tensor transformations. 
The go-to block-tensor map in 3D, the higher-order tensor renormalization group (HOTRG)~\cite{HOTRG:2012}, can be implemented using projective truncations.
When first proposed, the HOTRG made use of higher-order singular value decomposition. 
However, it is more straightforward to see the implications of lattice-reflection symmetry using projective truncations.

\subsection{2-to-1 isometric tensors}
To implement a block-tensor transformation with rescaling factor $b=2$, a type of 3-leg tensor is used to fuse two incoming legs into one output leg,
\begin{align}
    \label{eq:isom221}
    p: \mathbb{R}^{\chi'} \to 
    \mathbb{R}^{\chi} \otimes \mathbb{R}^{\chi}
    \text{ or pictorially }
    \includegraphics[scale=1.0, valign=c]{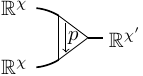}.
\end{align}
The dimensionality $\chi$ of the vector space $\mathbb{R}^\chi$ is known as \emph{bond dimension}.
This object $p$ is a tensor with three indices, $p_{ijk}$, where the first two indices correspond to the vector space $\mathbb{R}^\chi \otimes \mathbb{R}^\chi$ while the third one corresponds to $\mathbb{R}^{\chi'}$.
This tensor $p$ is isometric, satisfying the following condition,
\begin{subequations}
\begin{align}
    \label{eq:isomCondi}
    p^\intercal p = \mathbbm{1}_{\mathbb{R}^{\chi'}}
    \text{ or pictorially }
    \includegraphics[scale=1.0, valign=c]{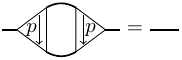},
\end{align}
where we use a single line to represent an identity operator $\mathbbm{1}$.
It is a convention in a tensor-network diagram that if a tensor leg is shared by two tensors, this leg is summed over (we also say the leg is contracted). 
This is similar to Einstein’s summation convention in tensor analysis.
In the usual tensor notation written as components, the isometric condition in Eq.~\eqref{eq:isomCondi} is
\begin{align}
    \label{eq:isomCondi2math}
    \sum_{ij} p_{ijk} p_{ijk'} = \delta_{k k'},
\end{align}
where $\delta_{k k'}$ is the Kronecker delta.
\end{subequations}
The arrow in the pictorial representation of $p$ in Eq.~\eqref{eq:isom221} denotes the order of the incoming indices:
\begin{align}
    \label{eq:arrowNotation}
    \includegraphics[scale=1.0, valign=c]{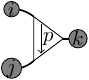}
    = p_{ijk}
    \text{ and }
    \includegraphics[scale=1.0, valign=c]{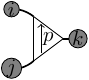}
    = p_{jik}.
\end{align}
Therefore, a reverse of the direction of the arrow of $p$ corresponds to a transpose of its two incoming legs:
\begin{subequations}
    \label{eq:swapOp}
\begin{align}
    \label{eq:pswap}
    \includegraphics[scale=1.0, valign=c]{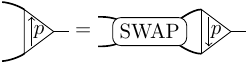},
\end{align}
where we define a SWAP operator as
\begin{align}
    \label{eq:swapDef}
    \includegraphics[scale=1.0, valign=c]{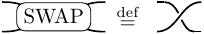}\quad.
\end{align}
\end{subequations}
A pair of two copies of $p$ contracted head-to-head, $p p^\intercal$, forms a projection operator,
\begin{align}
    \label{eq:ppt}
    p p^\intercal =
    \includegraphics[scale=1.0, valign=c]{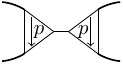},
\end{align}
and it is easy to check that $(p p^\intercal) (p p^\intercal) = p p^\intercal$ using the isometric condition in Eq.~\eqref{eq:isomCondi}.

\subsection{Block-tensor transformation using projective truncations\label{subsec:proj}}
In TNRG, the partition function of a given model is represented by a full contraction of a tensor network.
We will focus on a square-lattice tensor network with periodic boundary condition.
The partition function $Z(A, L_x, L_y)$ is the full contraction of an $L_x \times L_y$ square-lattice tensor network consisting of 4-leg tensor $A$,
\begin{align}
    \label{eq:tn2Z}
    Z(A, L_x, L_y) \texteq{def}
    \includegraphics[scale=1.0, valign=c]{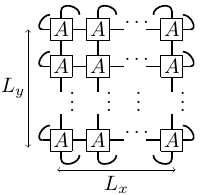}.
\end{align}
By inserting projection operators like $p p^\intercal$ into the square-lattice tensor network, one can implement a block-tensor transformation with $b = 2$ in the following way:
\begin{align}
    \label{eq:insertppt}
    \includegraphics[width=0.8\columnwidth, valign=c]{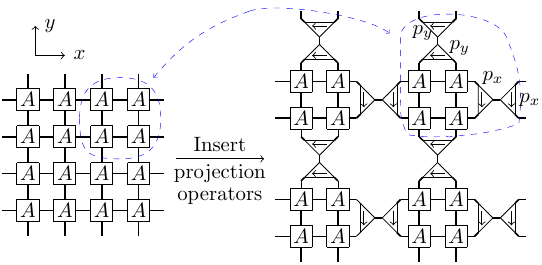}.
\end{align}
Two different isometric tensors $p_x$ and $p_y$ are used for different directions since rotation symmetry is not considered here.
It is now easy to see how to define the coarse-grained tensor $A'$:
\begin{align}
    \label{eq:bkten}
    \includegraphics[scale=1.0, valign=c]{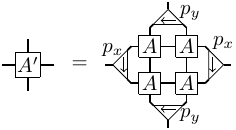}.
\end{align}
After this coarse graining, the partition function can be approximated by a coarser tensor network consisting of $A'$, with half the linear size in both directions: 
\begin{align}
    \label{eq:tncoarser}
    Z(A, L_x, L_y)
    \approx Z\left(A', \frac{L_x}{2}, \frac{L_y}{2}\right).
\end{align}
The coarser tensor network should be a good approximation of the partition function represented by the original tensor network.
Of course, if the output bond dimension $\chi' = \chi^2$, the projection operators in both directions become identity operator and the partition function is invariant under block tensor transformation.
In numerical calculations, however, truncations happen: $\chi' < \chi^2$.
A scheme is needed to choose isometric tensors such that the approximation in Eq.~\eqref{eq:tncoarser} is good.

The idea of projective truncations is choosing a local patch of the tensor network that the projection operator $pp^\intercal$ acts on, and demanding that the patch after the projection is a good approximation.
We call such a patch the environment of a given projective truncation. 
A larger environment usually gives better approximation of the partition function, but has higher computational costs. 
With the application in 3D in our mind, we choose a small environment to make the computational costs manageable. 
Let us focus on the $p_x p_x^\intercal = P_x$, whose approximation is chosen to be
\begin{align}
    \label{eq:projApprox}
    \includegraphics[scale=1.0, valign=c]{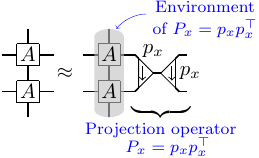}.
\end{align}
A natural way to quantify this RG approximation error is taking the norm of the difference between the two sides,
\begin{align}
    \label{eq:projError}
    \epsilon(p_x) \texteq{def}
    \includegraphics[scale=1.0, valign=c]{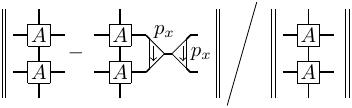},
\end{align}
where the norm $||T||$ is Frobenius norm $|| T || = \sqrt{\sum_{i_1 \ldots i_n} \left(T_{i_1 \ldots i_n}\right)^2}$, and the denominator is put there to normalize the error $0 \leq \epsilon(p_x) \leq 1$.

The isometric tensor should be such that the RG error is the smallest. 
We treat the environment of the projection operator as a matrix by regarding two legs contracted with $p_x p_x^\intercal = P_x$ as one matrix index and the remaining four legs as the other index,
\begin{align}
    \label{eq:envAA}
    \includegraphics[scale=1.0, valign=c]{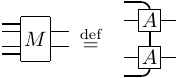}.
\end{align}
The error can be expanded as 
\begin{align}
    \label{eq:errExpand}
    \epsilon^2(p_x) 
    &=
    \frac{||M P_x - M||^2}{||M||^2} \nonumber\\
    &=
    \frac{\tr \left(\left(M P_x - M\right)^{\intercal}\left(M P_x - M\right)\right)}{\tr \left(M^\intercal M\right)} \nonumber\\
    &=
    1 - \frac{\tr\left(p_x^\intercal M^\intercal M p _x\right)}{\tr\left(M^\intercal M\right)}
    =
    1 - \tr\left(p_x^\intercal \rho p _x\right),
\end{align}
where in the last step we define $\rho \texteq{def} M^\intercal M / \tr\left(M^\intercal M\right)$, or pictorially
\begin{align}
    \label{eq:densityM}
    \includegraphics[scale=1.0, valign=c]{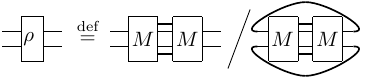}.
\end{align}
Since $\rho$ is positive semi-definite and $\tr(\rho)=1$ by construction, it can be interpreted as a density matrix.
The isometry $p_x$ that maximizes $\tr\left(p_x^\intercal \rho p_x\right)$ is a list of eigenvectors of $\rho$ corresponding to the $\chi'$ largest eigenvalues~\cite{EVP-Gen:2023}. 
Denote the eigenvalues as $\lambda_i, i=1,2,3,\cdots,\chi^2$, with $\lambda_1 \geq \lambda_2 \geq \cdots \geq \lambda_{\chi^2}$, the error of inserting $p_x p_x^\intercal$ is 
\begin{align}
    \label{eq:errorEig}
    \epsilon(p_x) = 
    \sqrt{\sum_{i = \chi'+1}^{\chi^2} \lambda_i}.
\end{align}
Notice that if the eigenvalue spectrum decays exponentially, the error can be well approximated by the $\chi'$-th eigenvalue, $\epsilon(p_x) \sim \sqrt{\lambda_{i=\chi'}}$.

After the isometric tensors $p_x, p_y$ are determined, the tensor network on the right-hand side of the tensor RG equation in Eq.~\eqref{eq:bkten} can be contracted to obtain the coarse-grained tensor.
By applying the tensor RG equation repeatedly, a RG flow in the space of 4-leg tensor is generated,
\begin{align}
    \label{eq:rgflowDef}
    A^{(0)} \mapsto A^{(1)} \mapsto \ldots \mapsto A^{(n)} \mapsto \ldots,
\end{align}
where $A^{(0)}$ denotes the initial tensor and $A^{(n)}$ is the coarse-grained tensor after $n$ RG steps.

\subsection{HOTRG-like block-tensor transformation\label{subsec:hotrg-like}}
Usually, the output bond dimension of $p_x$ and $p_y$ is set to be $\chi' = \chi$. 
The computational costs for the $A'$ contraction in Eq.~\eqref{eq:bkten} are $O(\chi^8)$;
they can be reduced by inserting another projection operator $p_i p_i^\intercal$ in the inner legs of the $2 \times 2$ block,
\begin{align}
    \label{eq:innerppt}
    \includegraphics[scale=1.0, valign=c]{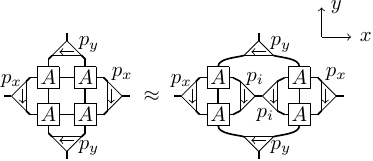}.
\end{align}
We call the transformation on the right-hand side a \emph{HOTRG-like block-tensor transformation}.
In general, the projection truncations acting on the inner leg can be different from those acting on the outer legs.
In the above example, however, we see that the environment of inner $p_i$ can be chosen to be the same as that of $p_x$ due to the translational symmetry in $x$ direction by one lattice constant. 
Therefore, it is legit to set $p_i = p_x$.
Then, the RG map becomes the usual HOTRG~\cite{HOTRG:2012}, which is a composition of two collapses in two directions,
\begin{align}
    \label{eq:hotrg2step}
    \includegraphics[scale=1.0, valign=c]{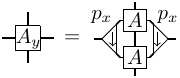},
    \includegraphics[scale=1.0, valign=c]{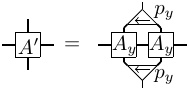}.
\end{align}
Here in 2D, the computation costs reduce from $O(\chi^8)$ to $O(\chi^7)$.
In 3D, the minimal costs of contracting a $2 \times 2 \times 2$ block of tensors are $O(\chi^{18})$.
The HOTRG reduces the costs to  $O(\chi^{11})$.
The order of the HOTRG collapses is arbitrarily chosen to be $y \rightarrow x$ here. 
Due to this arbitrary choice of the direction of collapses, the HOTRG explicitly breaks the lattice-rotation symmetry.
When entanglement filtering is incorporated, the inner and outer projective truncations will have different environments, in which case it is important to have the freedom to choose $p_i$ independently of the outer isometric tensors.

\section{Graph-independent entanglement filtering\label{sec:GIEF}}
We introduce an entanglement filtering (EF) scheme that is helpful for exploiting the lattice symmetries. 
One feature that makes an EF scheme works equally well in both 2D and 3D is that it is graph-independent~\cite{Hauru:2018}.
However, it is not clear how to exploit the lattice symmetries in the graph-independent EF scheme proposed in Ref.~\cite{Hauru:2018}.
The EF scheme proposed in this section combines the idea of being graph-independent emphasized in Ref.~\cite{Hauru:2018} and the optimization strategy developed in Ref.~\cite{Evenbly:2018};
this combination makes it possible to exploit lattice-reflection symmetry in both 2D and 3D.

At this stage, we focus on basic concepts of the EF, without considering any symmetry. 
Whenever a concept works in both 2D and 3D, we choose to expound in 2D for the ease of understanding. 
The generalization to 3D will be presented after these basic concepts are developed in 2D.

\subsection{Formulation of the entanglement filtering\label{subsec:EFformula}}
In 2D, the simplest redundant entanglement locates inside plaquettes as loop correlations, which can be explicitly demonstrated using a toy model called corner-double-line (CDL) tensors~\cite{Gu:Wen:2009}. 
Filtering of these loop correlations can be done by the following approximation,
\begin{align}
    \label{eq:2dEFapprox}
    \includegraphics[scale=0.8, valign=c]{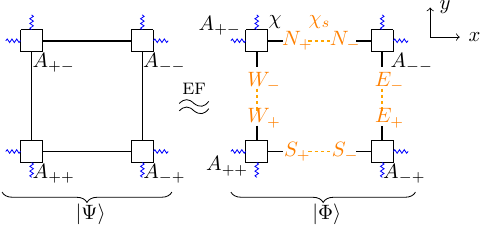}.
\end{align}
The left-hand side is the plaquette where the loop entanglement to be filtered locates. 
On the right-hand side, we insert 4 pairs of \emph{filtering matrices}, $N_-$, $N_+$, $\ldots$, with a squeezed bond dimension $\chi_s < \chi$. 
Here, the letters of the filtering matrices $E,S,W,N$ denotes the direction in the plaquette as East, South, West and North, while the subscript of the bond dimension $\chi_s$ means ``smaller'' or ``squeezed''.
The subscript of the 4-leg tensor, like the $++$ in $A_{++}$, indicates the relative position of the tensor in the plaquette.
This notation in Eq.~\eqref{eq:2dEFapprox}, along with the normal (black) and wavy (blue) lines for the tensor legs, is helpful for understanding how the EF is incorporated into a block-tensor map, which will be explained in~\autoref{subsec:assembly} below.
Moreover, this notation is also conducive to implementing the transposition trick for exploiting and lattice-reflection symmetry, which will be introduced in~\autoref{sec:lattsym}.
For the tensor network in Eq.~\eqref{eq:tn2Z}, all 4-leg tensors in Eq.~\eqref{eq:2dEFapprox} are the same: $A_{++} = A_{+-} = A_{-+} = A_{--} = A$.
The two tensor-network diagrams in Eq.~\eqref{eq:2dEFapprox} can be seen as two ket vectors $\ket{\psi}, \ket{\phi}$, which will be explained in Appendix~\ref{subsec:opts}.

If some loop correlations locate completely inside the plaquette, which means they become a single number after the contraction of the tensor network on the left-hand side, then one can truncate those states contributing to the loop correlations while leave the tensor network on the left-hand side invariant. 
If the localization of the correlations is not completely inside the loop, then one expects to find a good approximation.
Therefore, the loop filtering in Eq.~\eqref{eq:2dEFapprox} can be formulated as the following:

\emph{Given a tensor network forming a plaquette with bond dimension $\chi$, and a squeezed bond dimension $\chi_s < \chi$, determine the filtering matrices such that the filtered plaquette gives a good approximation of the original plaquette.}

More generally, a graph-independent entanglement filtering is the following procedure:

\emph{1. Identify a local patch of a tensor network, where the target entanglement to be cleaned is located. The choice of such patch can be guided by entanglement-entropy area laws~\cite{Levin:Nave:2007,Lyu:Kawashima:2023}.}
\par
\emph{2. By inserting pairs of filtering matrices, squeeze the bond dimension of the bonds where the target entanglement is involved.}
\par
\emph{3. Determine the filtering matrices such that the filtered patch gives a good approximation of the original patch.}

We will introduce a general scheme, which works in both 2D and 3D, for determining the filtering matrices in Appendix~\ref{app:findsmat}.
The squeezed bond dimension $\chi_s$ is expected to have a Goldilocks value. 
If $\chi_s$ is the same as the original bond dimension $\chi$, then no filtering happens. 
However, if $\chi_s$ is too small, the filtered patch may not be able to approximate the original patch. 
The strategy for determining a good $\chi_s$ might depend on a particular scenario where the EF is applied.

% Below are for 3D structures
\begin{figure}[tb]
    \includegraphics[width=0.70\columnwidth,
    valign=c]{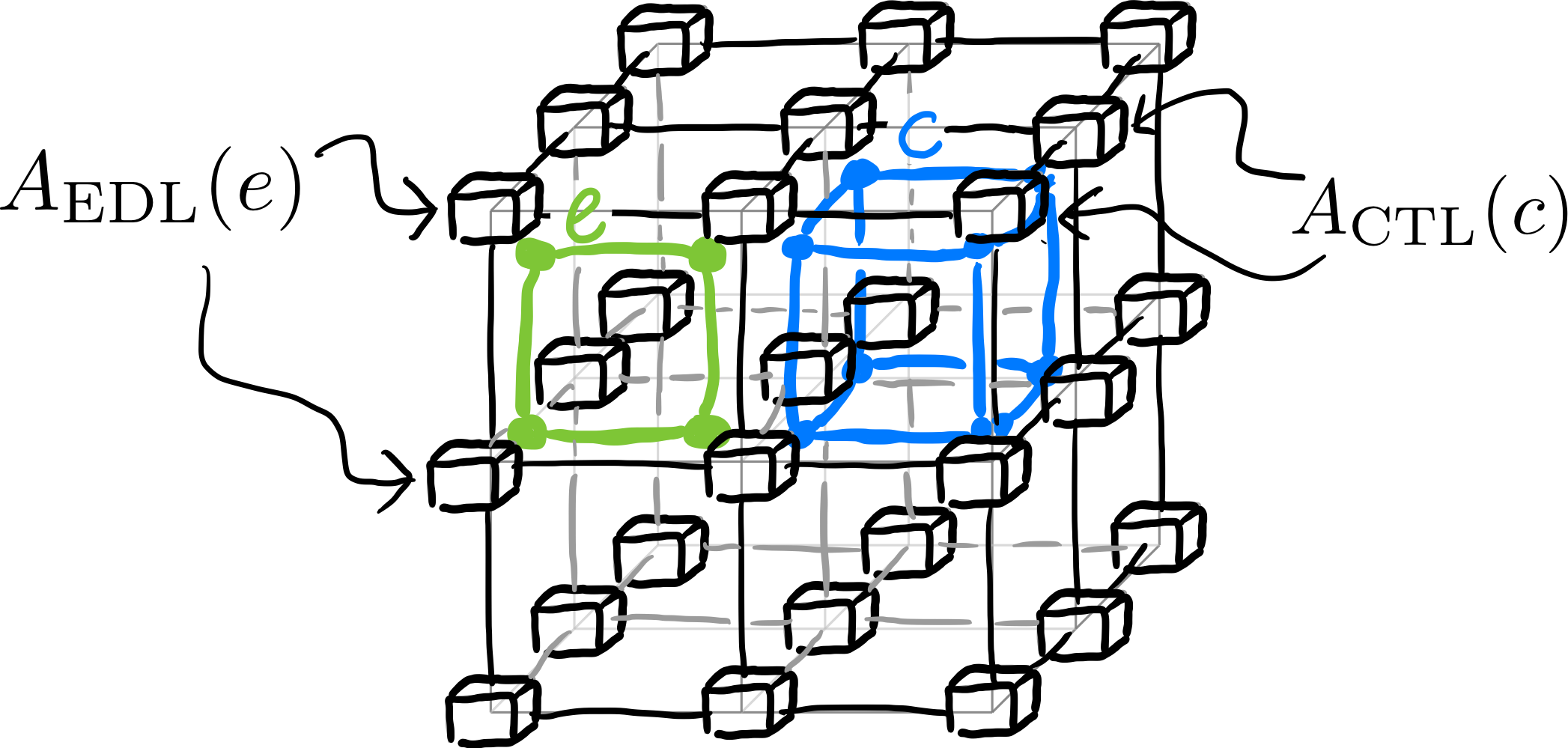}
    \caption{\label{fig:entanglement}
        Redundant entanglement structures in 3D
    }
\end{figure}

In a 3D tensor network that forms a cubic lattice, the redundant entanglement can locate inside 2D plaquettes and 3D cubes. 
In~\autoref{fig:entanglement}, tensors sit on the vertices of the cubic lattice. 
The loop correlations located in the plaquettes can be captured using a 2-leg edge matrix $e$, which is a reminiscence of the CDL toy model in 2D. 
We call the corresponding tensor in 3D an edge-double-line (EDL) tensor~\cite{Lyu:Kawashima:2023},
\begin{align}
    \label{eq:edlDef}
    \includegraphics[scale=1.0, valign=c]{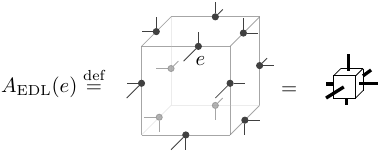}\quad.
\end{align}
The membrane correlations located in the cubes can be captured using a 3-leg corner matrix $c$, which is a new type of correlations in 3D~\cite{Hauru:2018}. 
We call the corresponding tensor a corner-triple-line (CTL) tensor,
\begin{align}
    \label{eq:ctlDef}
    \includegraphics[scale=1.0, valign=c]{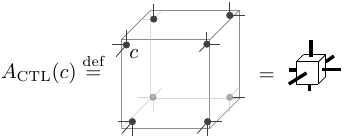}\quad.
\end{align}
These two toy models for the redundant entanglement are not fixed-point tensors of a simple block-tensor map in 3D.
This is different from the 2D case, where the CDL structure is fixed under a simple block-tensor map~\cite{Gu:Wen:2009}.
This difference can be understood as the manifestation of the scaling of the entanglement-entropy area law~\cite{Lyu:Kawashima:2023}.
The removal of these two 3D toy models corresponds to the renormalization of the linear growth term in the area law, which contains non-universal information.
Due to this linear growth, the removal of the redundant entanglement is more urgent in 3D than 2D.

Both the EDL and the CTL tensors can be filtered by choosing the following approximation for a $2 \times 2 \times 2$ cube of tensors,
\begin{align}
    \label{eq:cubeEFapprox}
    \includegraphics[width=0.95\columnwidth, valign=c]{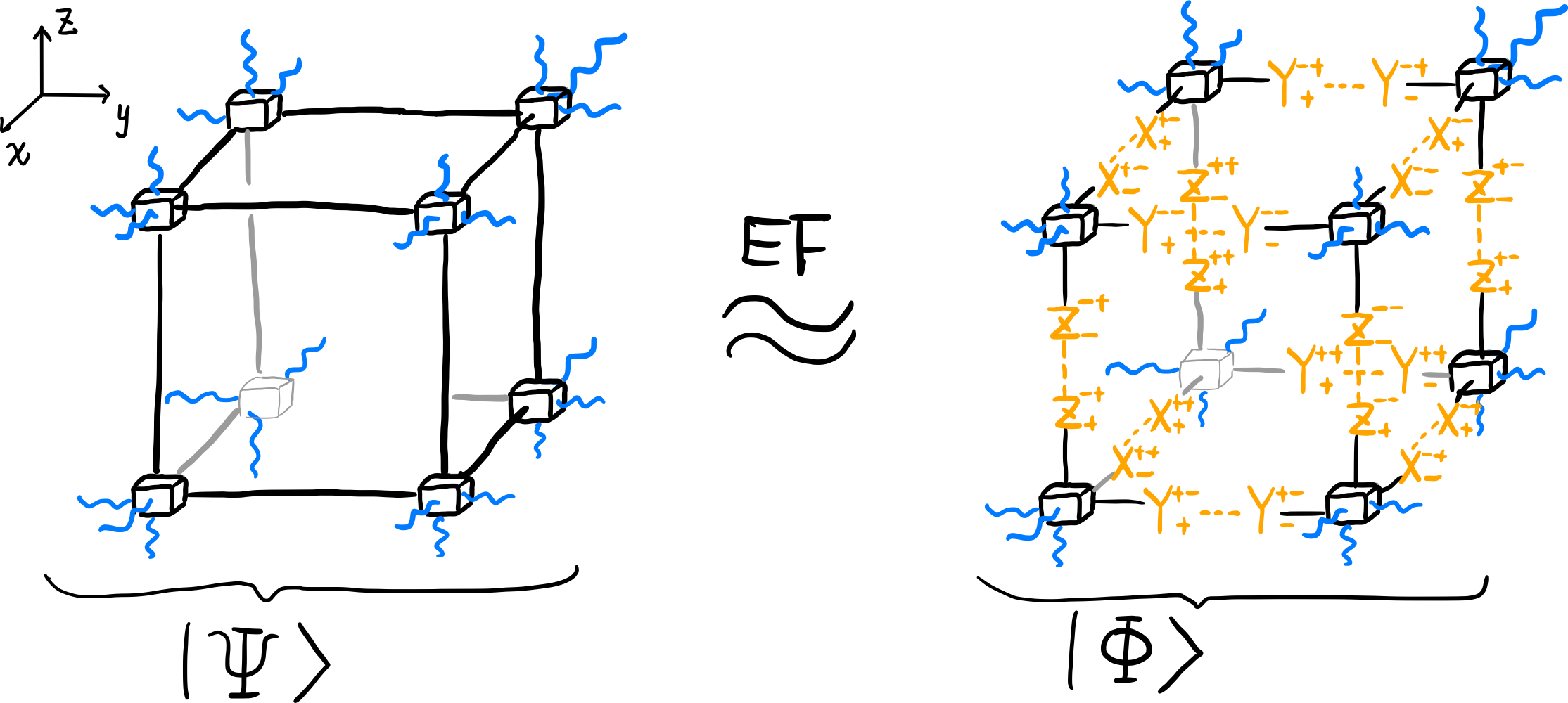}.
\end{align}
In total, there are 12 inner bonds (edges) in this cube; for each bond, a pair of filtering matrices like $X^{++}_{+}, X^{++}_{-}$ is inserted. 
In our notation for filtering matrices, $X$ denotes the direction of the bond is $x$ direction; the superscript denotes a specific bond among the four bonds in the same direction, while the subscript denotes the relative positive within the pair in a given bond. 
The bond dimension of the solid legs of the filtering matrices is $\chi$, while that of the dashed legs is $\chi_s$.
Without considering the lattice symmetry, there are $12 \times 2=24$ filtering matrices to be determined. 
Since the target patch of this 3D EF in Eq.~\eqref{eq:cubeEFapprox} is a $2 \times 2 \times 2$ cube of tensors, we call it a \emph{cube filtering}.

\subsection{Assembly of entanglement filtering and block-tensor map\label{subsec:assembly}}
\begin{figure}[tb]
    \includegraphics[width=0.95\columnwidth,
    valign=c]{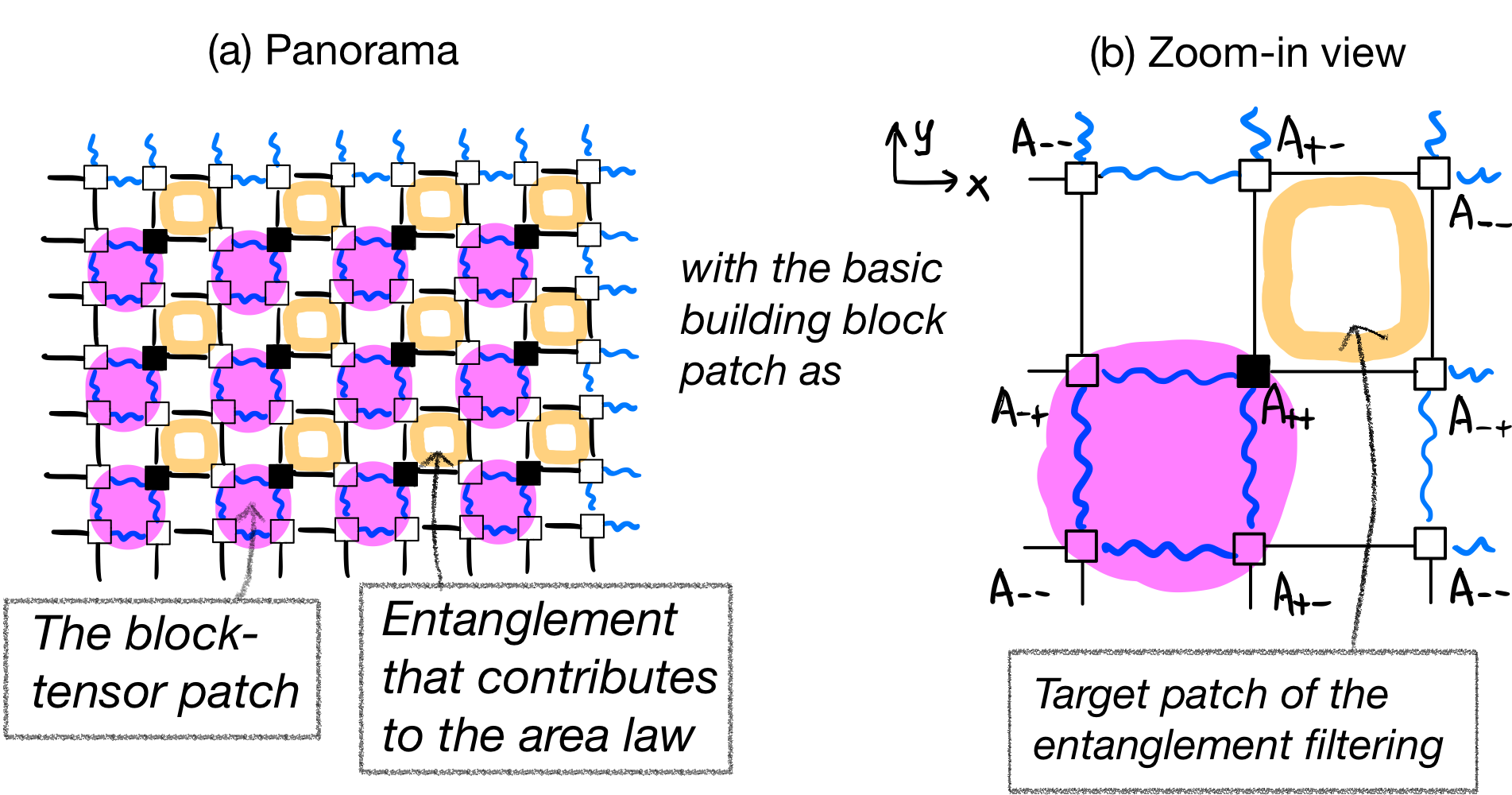}
    \caption{\label{fig:sec2-2dblockEE} Location of redundant entanglement for a simple block-tensor transformation and the principle for incorporating the EF process into a block-tensor map.
        The tensors marked with a black solid dot are the anchor points with their position labeled as $(++)$ in Eqs.~\eqref{eq:tsptrick2D} and~\eqref{eq:2dEFapproxSym}.
    }
\end{figure}

\begin{figure}[tb]
    \includegraphics[width=0.95\columnwidth,
    valign=c]{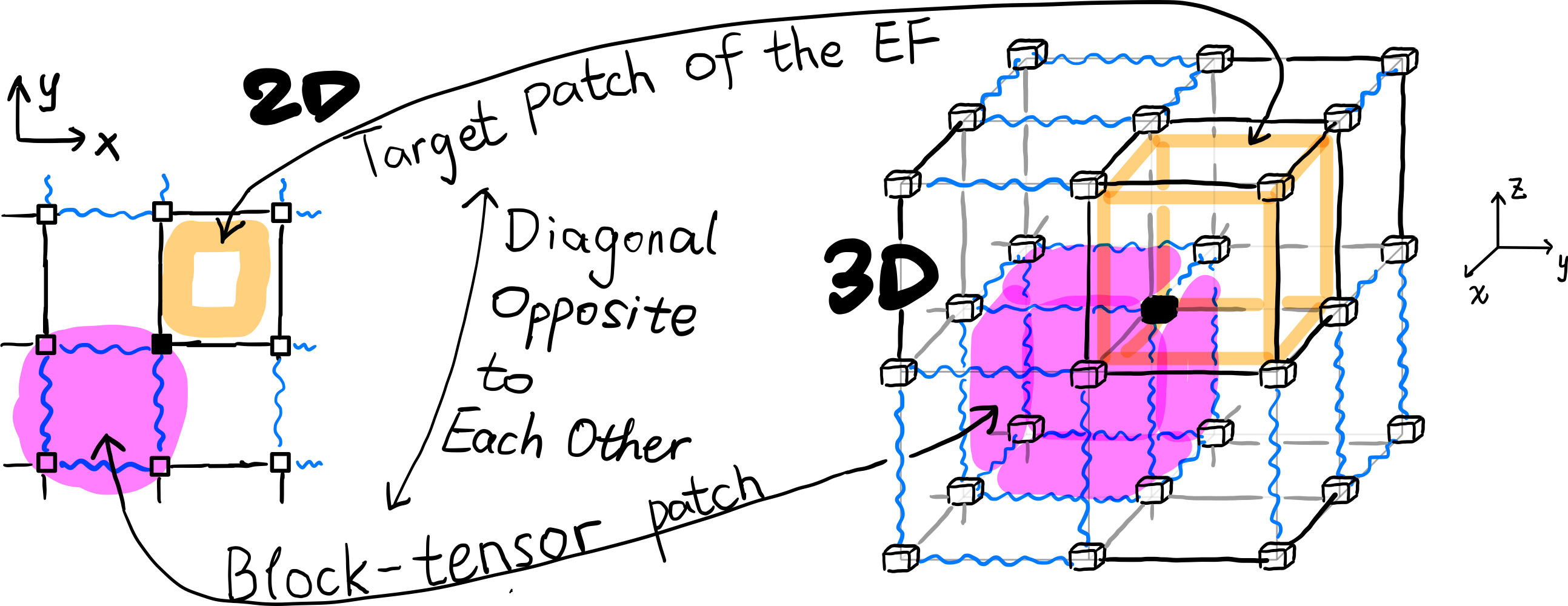}
    \caption{\label{fig:sec2-3dblockEE} The principle for incorporating the EF process into a block-tensor map in 3D.  }
\end{figure}

In this subsection, the principle of how to assemble the EF and the block-tensor map will be demonstrated in the 2D square-lattice tensor network in Eq.~\eqref{eq:tn2Z}. 
The generalization to 3D is straightforward and has already been laid out in our previous paper~\cite{Lyu:Kawashima:2024}.

The EF approximation in Eq.~\eqref{eq:2dEFapprox} should clean up the patch of tensor network where the short-range entanglement cannot be detected and eliminated by a simple block-tensor map; it is such short-range entanglement that is carried over to the next length scale under an RG transformation. 
For the block-tensor transformation in Eq.~\eqref{eq:insertppt}, the location of the redundant entanglement is suggested by the entanglement-entropy area law argument~\cite{Lyu:Kawashima:2023}, which is summarized in~\autoref{fig:sec2-2dblockEE}.
From the panorama of the relationship between the block tensor and the area-law contribution of the entanglement entropy, it becomes clear how to choose the target patch of the EF in Eq.~\eqref{eq:2dEFapprox}: the target patch of the EF locates at the four corners of the block-tensor patch.
The 3D generalization of how to combine the EF with the block-tensor map is summarized in~\autoref{fig:sec2-3dblockEE}.

After determining the eight filtering matrices in Eq.~\eqref{eq:2dEFapprox} according to the techniques explained in~\autoref{subsec:opts} and~\autoref{subsec:inits}, each 4-leg tensor in the block-tensor patch in~\autoref{fig:sec2-2dblockEE} absorbs two filtering matrices.
For example, the tensor $A_{++}$ absorbs the filtering matrices $S_+$ and $W_+$,
\begin{align}
    \label{eq:A2Af}
    \includegraphics[scale=1.0, valign=c]{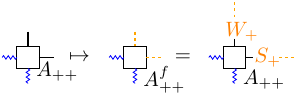}\quad.
\end{align}
The superscript $f$ on $A^f_{++}$ means ``filtered''. 
Afterwards, the block-tensor map is applied to the filtered block-tensor patch, leading to a tensor RG equation
\begin{align}
    \label{eq:efrg2dmap}
    \includegraphics[scale=0.8, valign=c]{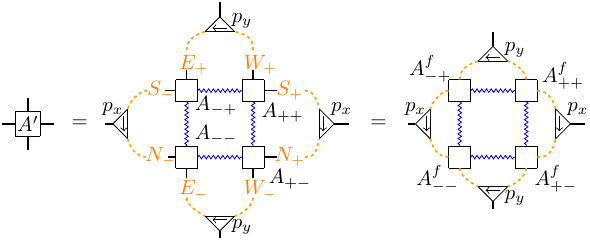}\;.
\end{align}
This is the RG equation of the block-tensor map enhanced by the graph-independent EF.

The computational costs of the block-tensor contraction in Eq.~\eqref{eq:efrg2dmap} can be reduced by performing an additional projective truncation in the inner legs of the block using a pair of a new isometric tensor $p_i$, as has been explained in~\autoref{subsec:hotrg-like}. 
After incorporating the EF, $p_i$ should be chosen independently of the outer isometric tensor $p_x$. 
The output bond dimension $\chi'$ of $p_i$ should be a slightly larger than $\chi$ to make sure that this inner projection truncation always has an error smaller than those of the two outer projective truncations using $p_x$ and $p_y$.

\section{Exploiting the lattice-reflection symmetry\label{sec:lattsym}}
In this section, we will explain how to exploit the lattice-reflection symmetry in TNRG.
Specifically, we will demonstrate
\begin{itemize}
    \item the appropriate definition of lattice-reflection symmetry in TNRG,
    \item how to preserve and impose the lattice-reflection symmetry using a transposition trick, and
    \item the symmetry properties of the isometric tensor in projective truncations and of the filtering matrices in the graph-independent entanglement filtering.
\end{itemize}

\subsection{Definition of the reflection symmetry\label{subsec:defSym}}
We write down the definition of the lattice-reflection symmetry in 1D, 2D and 3D TNRG, respectively.

\subsubsection{An 1D toy example}
The definition of the lattice-reflection symmetry in 1D TNRG is straightforward. 
The tensor-network representation of the partition function is nothing but a chain of copies of a transfer matrix $A$, while the block-tensor map is a matrix multiplication of $A$,
\begin{subequations}
    \label{eq:bkten1d}
\begin{align}
    \label{eq:bkten1dPic}
    \includegraphics[scale=1.0, valign=c]{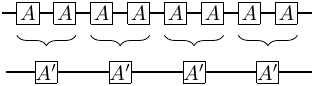}\quad,
\end{align}
and the tensor RG equation is simply
\begin{align}
    \label{eq:bkten1dEq}
    A' = AA.
\end{align}
\end{subequations}
Take the initial tensor of the nearest-neighbor (NN) Ising model as an example.
The initial tensor is 
\begin{align}
    \label{eq:ising1dTen}
    A^{(0)}_{\sigma \sigma'} =
    e^{\beta \sigma \sigma'}
    =
    \includegraphics[scale=1.0, valign=c]{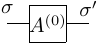}\quad,
\end{align}
where $\beta$ is the inverse temperature and $\sigma,\sigma' \in \{-1, +1\}$ are Ising spins.
The lattice-reflection symmetry, in this tensor-network language, becomes the fact that $A^{(0)}$ is symmetric,
\begin{align}
    \label{eq:1dTenSym}
    \left(A^{(0)}\right)^\intercal = A^{(0)}.
\end{align}
This symmetry is trivially preserved by the block-tensor map in Eq.~\eqref{eq:bkten1d}, since
\begin{align}
    \label{eq:1dTenSymP}
    \left(A^{(1)}\right)^{\intercal}
     &= \left(A^{(0)} A^{(0)}\right)^\intercal
     = \left(A^{(0)}\right)^\intercal \left(A^{(0)}\right)^\intercal \nonumber\\
     &= A^{(0)} A^{(0)} = A^{(1)}
\end{align}
Therefore, the definition of the lattice-reflection symmetry in 1D is that the transfer matrix is symmetric $A^\intercal = A$.

\subsubsection{In 2D}
In 2D TNRG, however, one feature appears that is not present in the above 1D example. 
A proper definition of lattice-reflection symmetry is
\begin{subequations}
    \label{eq:refl2d}
\begin{align}
    \includegraphics[scale=1.0, valign=c]{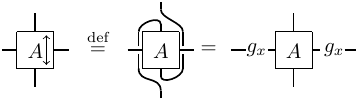}\quad,\\
    \includegraphics[scale=1.0, valign=c]{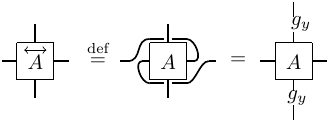}\quad.
\end{align}
\end{subequations}
The new feature is the SWAP-gauge matrices, $g_x$ and $g_y$, which have the following property,
\begin{align}
    \label{eq:2dSWAPmat2id}
    g_x g_x = g_y g_y = 1,
\end{align}
reflecting the $\mathbb{Z}_2$ nature of the lattice-reflection symmetry.

In 2D, it is trickier to see how this symmetry is preserved under an RG transformation.
We will develop a transposition trick in~\autoref{subsec:tsptrick} to make it less tricky.
In~\autoref{subsec:algo2d}, we propose a 2D TNRG algorithm using the transposition trick and prove that the lattice-reflection symmetry of $A$ in Eq.~\eqref{eq:refl2d} is preserved for the coarse-grained tensor $A'$ in Eq.~\eqref{eq:2dEFhotRGeq} with two coarse-grained SWAP-gauge matrices $g_x'$ and $g_y'$.

\subsubsection{In 3D}
The 3D definition is similar to that of the 2D.
Take the reflection across the $z$-plane (we refer to a plane using its normal direction) as an example, the definition of this lattice-reflection symmetry is
\begin{align}
    \label{eq:refl3d}
    \includegraphics[width=0.85\columnwidth, valign=c]{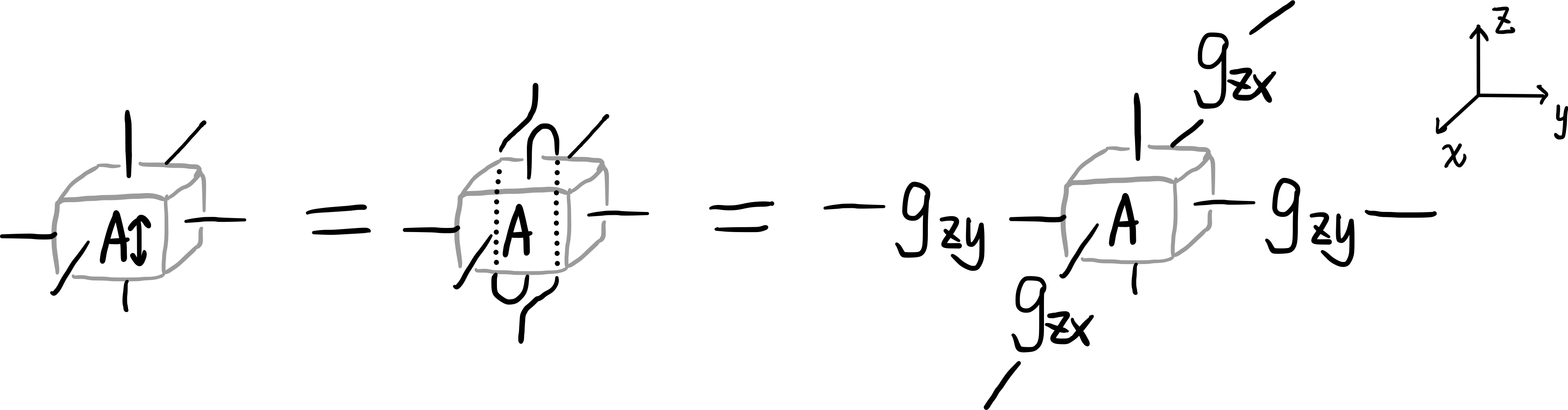}.
\end{align}
There are two SWAP-gauge matrices $g_{zx}, g_{zy}$, where the first index of $g$ denotes the reflection plane, while the second denotes the leg the SWAP-gauge matrix acts on. 
These SWAP-gauge matrices also have the $\mathbb{Z}_2$ property,
\begin{align}
    \label{eq:3dSWAPmat2id}
    g_{zx} g_{zx} = g_{zy} g_{zy} = 1.
\end{align}
The lattice-reflection symmetry for reflection across the $x$- and $y$-plane is defined similarly. 
In total, there are six SWAP-gauge matrices $(g_{zx}, g_{zy}, g_{yz}, g_{yx}, g_{xy}, g_{xz})$, two for each direction.
Like the 2D case, we will postpone the proof that this lattice-reflection symmetry can be preserved under an RG transformation in~\autoref{subsec:3dalgo}.

For the cubic-lattice Ising model with the nearest-neighbor-interaction at inverse temperature $\beta$, the tensor-network representation of its partition function can be constructed according to the procedure described in Ref.~\cite{Hauru:2018}.
The tensor network is also a cubic lattice consisting of copies of six-leg initial tensor $A$ whose components are
\begin{subequations}
\begin{align}
    \label{eq:3dIsingA0}
    A_{i_x i_{x'} i_y i_{y'} i_y i_{y'}}^{(0)}
    =
    \sum_{\sigma}
    W_{\sigma i_x} W_{\sigma i_{x'}}
    W_{\sigma i_y} W_{\sigma i_{y'}}
    W_{\sigma i_z} W_{\sigma i_{z'}},
\end{align}
where
\begin{align}
    \label{eq:3dIsingA0W}
    W = 
\begin{pmatrix}
\sqrt{\cosh{\beta}} & \sqrt{\sinh{\beta}} \\
\sqrt{\cosh{\beta}} & -\sqrt{\sinh{\beta}} \\
\end{pmatrix}
\end{align}
\end{subequations}
and the three legs $i_{x}, i_{y}, i_{z}$ of $A$ point towards the positive $x,y,z$ directions respectively, while the other three legs point towards negative directions.
It is easy to see that this initial tensor satisfies the definition of lattice-reflection symmetry  with all the six SWAP-gauge matrices being the identity matrix.

\subsection{Origin of the SWAP-gauge matrix\label{subsec:originSWAP}}
It was Evenbly~\cite{Evenbly:2017:algo} who first pointed out the necessity of the SWAP-gauge matrices in the definition of symmetry in terms of the tensors in Eqs.~\eqref{eq:refl2d} and~\eqref{eq:refl3d}.
However, it was not clearly understood whence these matrices arise.
In this subsection, a simple example in 2D is used to demonstrate the origin of these SWAP-gauge matrices.
Take the initial tensor of the 2D Ising model with the NN interaction,
\begin{align}
    \label{eq:ising2dTen}
    A^{(0)}_{\sigma_x \sigma_{x'} \sigma_{y} \sigma_y'} 
    &=
    e^{\beta (\sigma_{x} \sigma_{y} + \sigma_{y} \sigma_{x'} + \sigma_{x'} \sigma_{y'} + \sigma_{y'} \sigma_{x})} \nonumber\\
    &=
    \includegraphics[scale=1.0, valign=c]{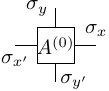}\quad.
\end{align}
The tensor legs represent the Ising spins and the tensor itself encodes a local Boltzmann weight containing four NN interactions.
It is easy to see that this initial tensor satisfies the lattice-reflection symmetry in Eq.~\eqref{eq:refl2d} with trivial SWAP-matrices $g_x = g_y = 1$.

To see how the SWAP-gauge matrix arises, we perform an \emph{exact} coarse graining of two tensors along the $x$ direction, 
\begin{align}
    \label{eq:xblockA0}
    \includegraphics[width=0.85\columnwidth, valign=c]{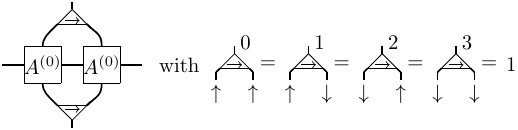},
\end{align}
where all other components of the isometric tensor vanish.
The isometric tensor simply relabel the two indices into one.
We can study the symmetry property of this coarse-grained tensor by swapping its legs in $x$ direction and check how it is related to the original one without the swapping,
%\begin{widetext}
\begin{align}
    \label{eq:swap2A0}
    \includegraphics[scale=1.0, valign=c]{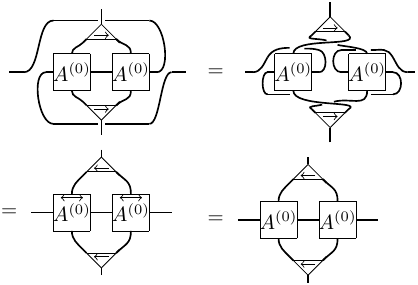}.
\end{align}
%\end{widetext}
Two copies of $A^{(0)}$ are exchanged in the first equality. 
The second equality uses another way to represent the transposition of $x$ legs, as well as the fact that the SWAP operation in Eq.~\eqref{eq:swapOp} changes the arrow direction of the 3-leg tensor. 
In the last equality, the symmetry of the initial tensor under the reflection in Eq.~\eqref{eq:refl2d} (with trivial $g_y = 1$) is used. 
The last expression in Eq.~\eqref{eq:swap2A0} is almost the same as the tensor network in Eq.~\eqref{eq:xblockA0} except that the arrow of the isometric tensor is reversed.

To derive the relationship between the two isometric tensor with opposite arrow, recall that when the output dimensionality in Eq.~\eqref{eq:isom221} $\chi' = \chi^2$, the isometric tensor contains a complete set of orthonormal basis, and the projection operator in Eq.~\eqref{eq:ppt} becomes identity,
\begin{align}
    \label{eq:ppIden}
    \includegraphics[scale=1.0, valign=c]{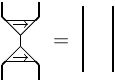}\quad.
\end{align}
This equation is clearly satisfied by the isometric tensor in the exact coarse graining in Eq.~\eqref{eq:xblockA0}.
Apply this identity to the two legs of the isometric tensor with arrow pointing to the left, we have
\begin{align}
    \label{eq:isomChDir}
    \includegraphics[width=0.83\columnwidth, valign=c]{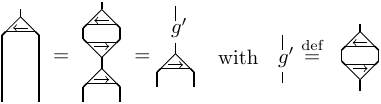}\;,
\end{align}
which can be used to change the direction of the arrow of the isometric tensor in the last diagram in Eq.~\eqref{eq:swap2A0},
\begin{align}
    \label{eq:xblockA0sym}
    \includegraphics[scale=1.0, valign=c]{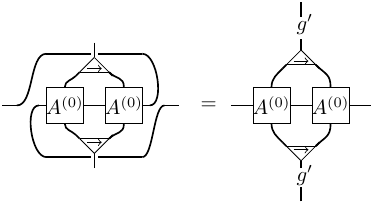}\;.
\end{align}
Therefore, we see that the reflection across the $y$ axis not only transposes two $x$ legs of the tensor, but it also acts on the two $y$ legs as an SWAP operator. 
This is reason why the SWAP-gauge matrix is necessary in the definition of the lattice-reflection symmetry in 2D and 3D in Eqs.~\eqref{eq:refl2d} and~\eqref{eq:refl3d}.
Notice that the SWAP-gauge matrix $g'$ appears in Eqs.~\eqref{eq:isomChDir} and~\eqref{eq:xblockA0sym} because it is associated with the coarse-grained tensor $A'$ after the RG map, instead of the original tensor $A$.

Although in this subsection, we demonstrate the origin of the SWAP-gauge matrix using an exact coarse graining, where the isometric tensor satisfies Eq.~\eqref{eq:ppIden}, the result in Eq.~\eqref{eq:isomChDir} about how an isometric tensor transforms when its arrow changes direction remains the same even when truncation happens ($\chi' < \chi^2$) in projective truncations.
We will prove this claim in~\autoref{subsec:symIT}.

\subsection{Transposition trick\label{subsec:tsptrick}}
After the definition of the lattice-reflection symmetry is written down, we explain, in this subsection, how to preserve and impose this symmetry in TNRG using a transposition trick. 
It is helpful to have the option to impose the symmetry, which can restrain the possible interactions in the tensor RG space. 
The projective truncations and EF will be applied after the transposition trick. 
We will see below in~\autoref{subsec:symFM},~\autoref{subsec:symIT} and~\autoref{sec:algo} that the transposition trick streamlines the implementation of both the projective truncations and the EF process, and reduces the necessity of considering the SWAP-gauge matrices in the symmetry argument.

\subsubsection{The 1D toy example}
The basic intuition of the transposition trick comes from imposing the matrix to be symmetry under multiplication in 1D. 
In the 1D block-tensor RG in Eq.~\eqref{eq:bkten1d}, $A' = A A$, the 1D lattice-reflection symmetry $A^\intercal = A$ is preserved, but not imposed. 
Specifically, if the symmetry condition of the origin tensor $A$ is broken by some numerical errors or artifacts, the lattice-reflection symmetry of the coarse-grained tensor $A'$ is also broken.
To impose it in numerical calculation, one way is transposing half the matrices before the block-tensor map,
\begin{align}
    \label{eq:tsptrick1D}
    \includegraphics[scale=1.0, valign=c]{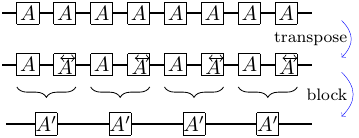},
\end{align}
where the double-arrow notation in the diagram, like Eq.~\eqref{eq:refl2d}, means the transposition of the two legs of the tensor.
After this transposition trick, the tensor RG equation in Eq.~\eqref{eq:bkten1d} becomes
\begin{align}
    \label{eq:bkten1d-sym}
    A' = A A^\intercal.
\end{align}
In this way, the lattice-reflection symmetry is not only preserved, but is also imposed. 
The reason is that the coarse-grained tensor $A'$ in Eq.~\eqref{eq:1dTenSymP} is symmetry by construction regardless of the symmetry of the input tensor $A$.

\subsubsection{In 2D}
In 2D, this transposition trick becomes transposing every other line of tensors. 
Focusing on a $2 \times 2$ block-tensor patch in~\autoref{fig:sec2-2dblockEE}, the process looks like
\begin{align}
    \label{eq:tsptrick2D}
    \includegraphics[width=0.85\columnwidth, valign=c]{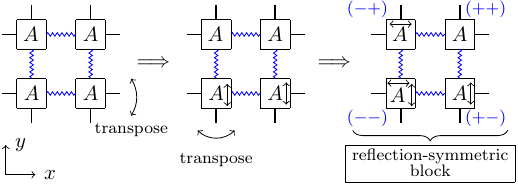}.
\end{align}
In the first step, tensors in the lower row of the $2 \times 2$ block are transposed for their $y$ legs. 
Due to the lattice-reflection symmetry of the tensor in Eq.~\eqref{eq:refl2d}, the partition function remains invariant under this transposition,
\begin{align}
    \label{eq:tsptrick2D-zinv}
    \includegraphics[width=0.83\columnwidth, valign=c]{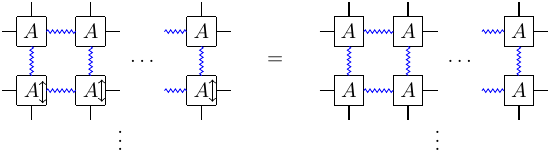}\;,
\end{align}
since the SWAP-gauge matrix $g_x$ in Eq.~\eqref{eq:refl2d} squares to identity, $g_x g_x = 1$. 
For the same reason, in the second step of Eq.~\eqref{eq:tsptrick2D}, the transposition of the $x$ legs of tensors in the left column of the $2 \times 2$ block also preserves the partition function. 
In summary, the partition function is invariant under the transposition trick in Eq.~\eqref{eq:tsptrick2D} if the tensor $A$ has the lattice-reflection symmetry defined in Eq.~\eqref{eq:refl2d}.
After the transposition trick, the last $2 \times 2$ block in Eq.~\eqref{eq:tsptrick2D}, which we will refer to as the \emph{reflection-symmetric block}, is the analogy of the 1D case in Eq.~\eqref{eq:bkten1d-sym}, $A A^\intercal$.

\subsubsection{In 3D}
In 3D, the transposition trick is transposing every other layer of tensors. 
Focusing on a $2 \times 2 \times 2$ block of tensors, the process looks like
\begin{align}
    \label{eq:tsptrick3D}
    \includegraphics[width=0.85\columnwidth, valign=c]{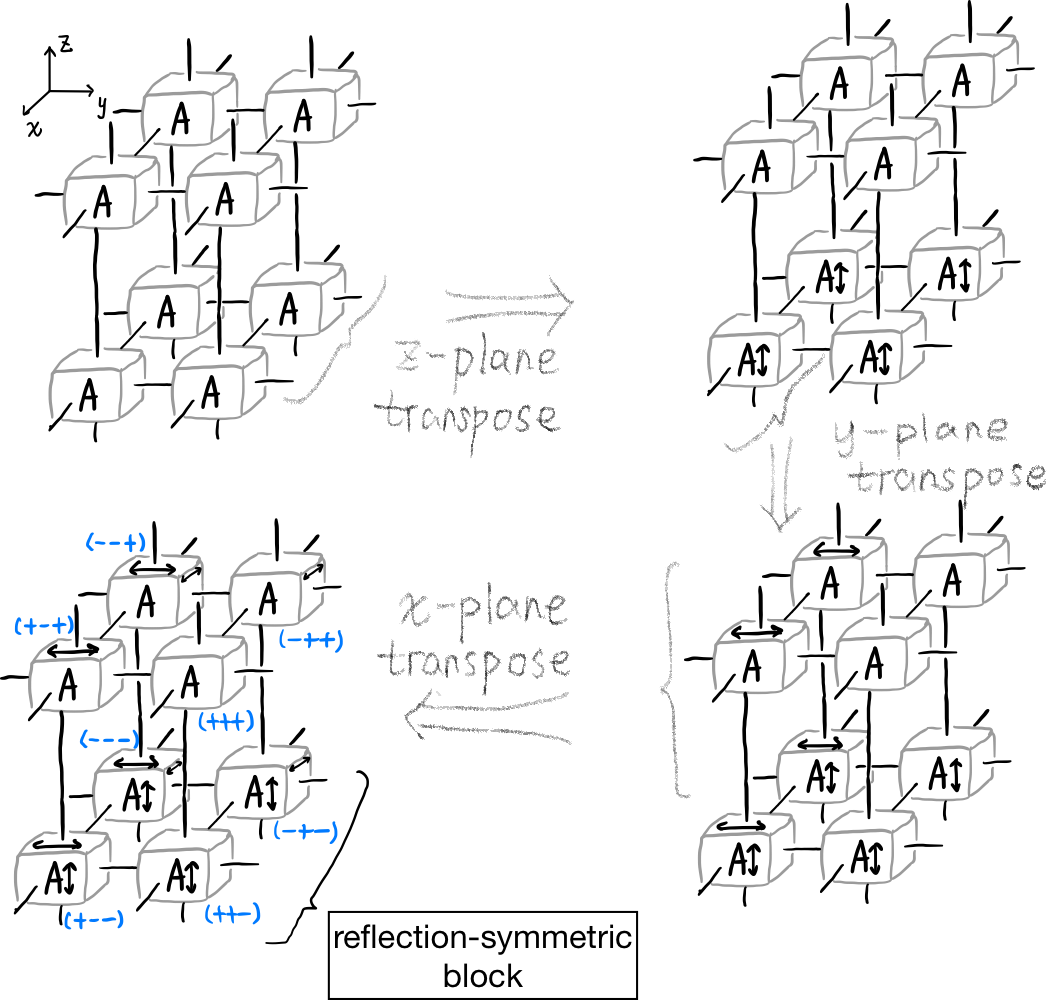}.
\end{align}
Similarly to the 2D case, one can show that the partition function is invariant under this transposition trick.

\subsubsection{The rules of the transposition trick in a reflection-symmetric block}
We can write down the rules of how the tensors in the reflection-symmetric blocks are transposed in Eqs.~\eqref{eq:tsptrick2D} and~\eqref{eq:tsptrick3D}.
Use the 2D case in Eq.~\eqref{eq:tsptrick2D} as an example:
\begin{itemize}
    \item The tensor located at $x+$ and $y+$ corner has no transposition. 
        We label this position $(++)$ and refer to this corner as an \emph{anchor point}.
    \item Other three positions are labeled according to its relative position to the anchor point. 
        For example, $(-+)$ indicates the $x-$ and $y+$ corner of the block.
    \item The legs of the tensor are transposed according to the minus sign of its relative position in the block.
        For example, the tensor located at $(-+)$ corner has its $x$ legs transposed.
\end{itemize}
It is straightforward to generalize these rules to the 3D case in Eq.~\eqref{eq:tsptrick3D}.

\subsection{Symmetry property of the filtering matrices\label{subsec:symFM}}
In this subsection, we study the symmetry properties of the filtering matrices after the transposition trick.
To keep the diagrams clean and easy to understand, we focus on our proof in 2D.

\subsubsection{In 2D}
Without considering the lattice-reflection symmetry, the EF approximation in 2D is shown in Eq.~\eqref{eq:2dEFapprox}.
After the transposition trick, we will show that the number of independent filtering matrices reduces from eight to two, and the EF approximation becomes
\begin{align}
    \label{eq:2dEFapproxSym}
    \includegraphics[width=0.85\columnwidth, valign=c]{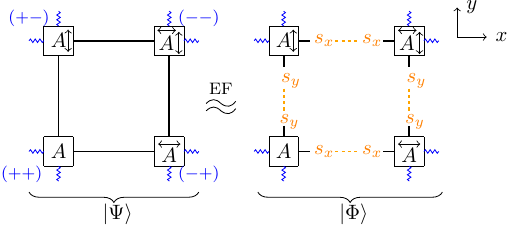}.
\end{align}
Notice that the target patch of the EF in this equation is a different $2 \times 2$ block from the block-tensor patch in Eq.~\eqref{eq:tsptrick2D}, with the EF block locating at the four corner of the block-tensor patch, as has been demonstrated in~\autoref{fig:sec2-2dblockEE}.

To show how the number of independent filtering matrices reduces, our strategy is showing the $\Upsilon_0$ tensor in Eq.~\eqref{eq:UpsilonP0} for initialization of $E_+$ and $E_-$ is the same as that of $W_+$ and $W_-$ (see Eq.~\eqref{eq:diagram2state}).
Recall that the $E_+, E_-$ pair is initialized by treating their product as a low-rank matrix $L_E$ in Eq.~\eqref{eq:LEdef}, whose $\Upsilon_0$ tensor, according to Eq.~\eqref{eq:UpsilonP0}, is
\begin{align}
    \label{eq:UpsilonLE}
    \includegraphics[width=0.85\columnwidth, valign=c]{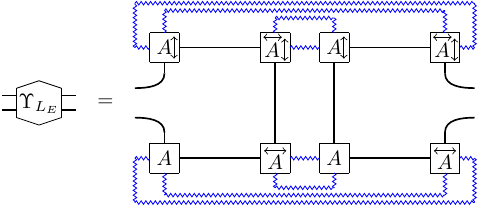}\;.
\end{align}
Meanwhile, the $\Upsilon_0$ tensor for $L_W$ is
\begin{align}
    \label{eq:UpsilonLW}
    \includegraphics[width=0.85\columnwidth, valign=b]{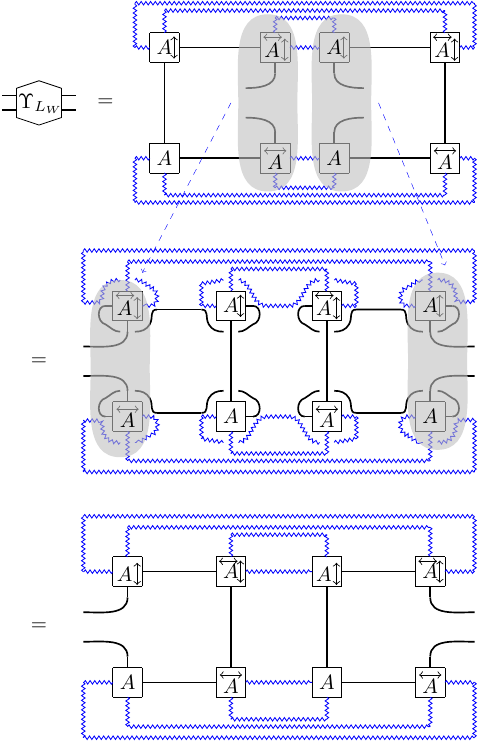}\;.
\end{align}
By comparing the last tensor-network diagram in Eq.~\eqref{eq:UpsilonLW} with Eq.~\eqref{eq:UpsilonLE}, we conclude that
\begin{align}
    \label{eq:UpsilonLWeqLE}
    \Upsilon_{L_W} = \Upsilon_{L_E}.
\end{align}
Therefore, the initialization of the $L_W$, according to the scheme in Appendix~\ref{subsec:inits}, is the same as that of $L_E$ (see Eq.~\eqref{eq:2dEFapprox}),
\begin{align}
    \label{eq:EeqW}
    W_+ = E_+, W_- = E_-.
\end{align}

At this point, we want to comment about one advantage of the transposition trick.
Without the transposition trick, the $\Upsilon_0$ tensors for $L_W$ and $L_E$ are no longer identical.
There will be related to each other by the SWAP-gauge matrix $g_y$ of the tensor $A$.
Therefore, the filtering matrices corresponding to $L_W$ are related to those of $L_E$ by an appropriate multiplication of the SWAP-gauge matrix $g_y$.
We see that the transposition trick streamlines the implementation of the EF process when the lattice-reflection symmetry is exploited---the SWAP-gauge matrix rarely appears in the algorithm and the proof.

By dragging the tensor-network diagram of $\Upsilon_{L_E}$ in Eq.~\eqref{eq:UpsilonLE} around like Eq.~\eqref{eq:UpsilonLW}, it is easy to see the following symmetry of $\Upsilon_{L_E}$,
\begin{align}
    \label{eq:UpsilonLEsym}
    \includegraphics[scale=1.0, valign=c]{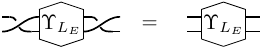}\quad,
\end{align}
whose consequence is the $L_E$ determined according to Eq.~\eqref{eq:LEsol} is symmetric,
\begin{align}
    \label{eq:LEsym}
    \includegraphics[scale=1.0, valign=c]{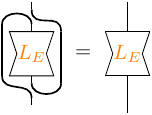}\quad.
\end{align}
Here, we need \emph{an additional assumption that we can choose $L_E$ to be positive semidefinite}, so that $L_E$ can be split using eigenvalue decomposition,
\begin{align}
    \label{eq:LEsplit}
    \includegraphics[scale=0.9, valign=c]{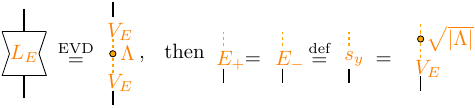}\; .
\end{align}
Therefore, we have shown that only a single filtering matrix $s_y$ is needed for the two vertical bonds in the EF approximation. 
Similarly, one can show the same result for horizontal bonds, whose filtering matrix is $s_x$ in Eq.~\eqref{eq:2dEFapproxSym}.
The assumption that the $L_E$ might be taken as positive semidefinite can be justified in numerical calculation by checking whether the fidelity of the EF approximation in Eq.~\eqref{eq:fidelityDef} is high.

\subsubsection{In 3D}
In 3D, when the EF is applied to the reflection-symmetric $2 \times 2 \times 2$ target patch of the EF, the number of independent filtering matrices reduces from $24$ (see Eq.~\eqref{eq:cubeEFapprox}) to three, one for each direction. 
Using an argument similar to the 2D case, one can show that the EF approximation can be taken to be
\begin{subequations}
    \label{eq:3dEFapproxSym}
\begin{align}
    \label{eq:3dEFapproxSym1}
    \includegraphics[width=0.95\columnwidth, valign=c]{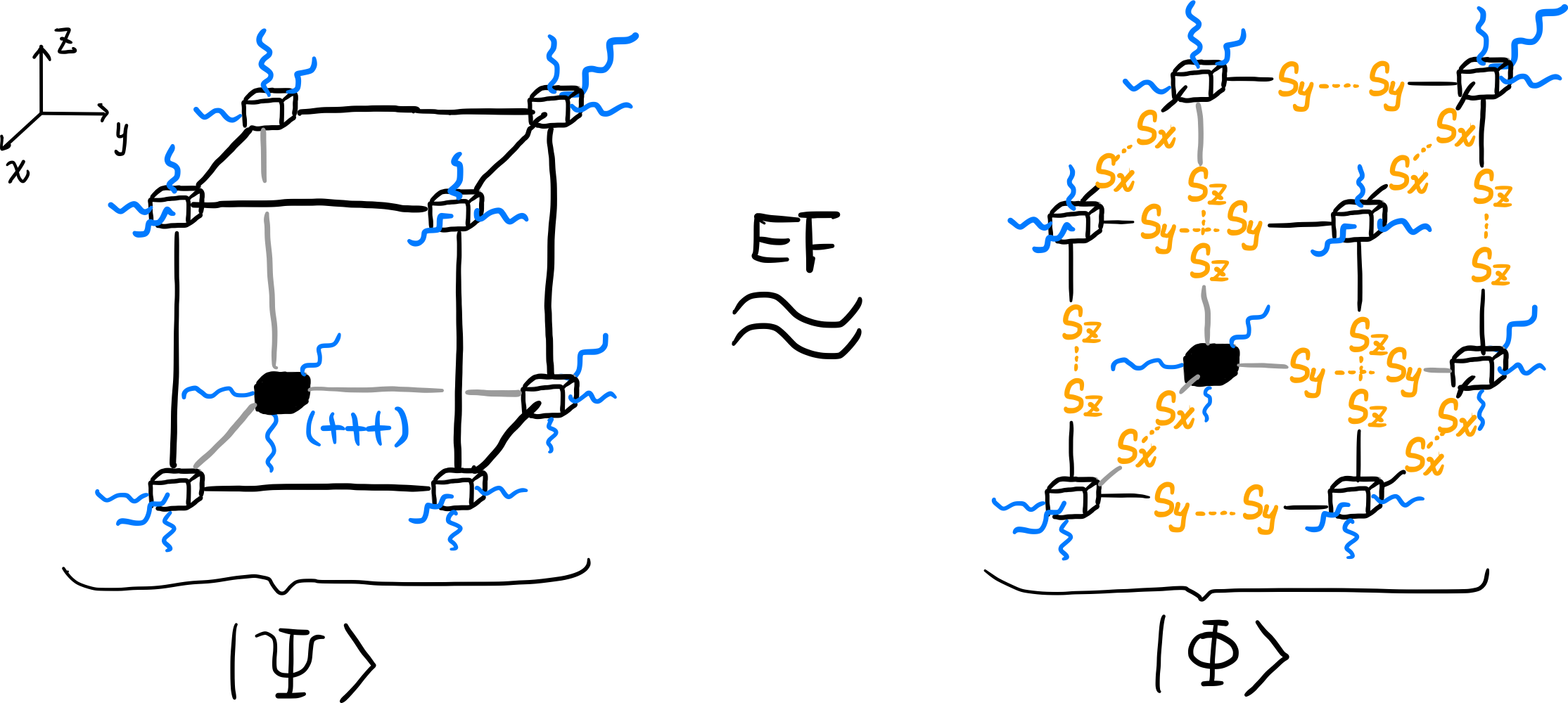},
\end{align}
where the $2 \times 2 \times 2$ cube is the reflection-symmetric block in Eq.~\eqref{eq:tsptrick3D},
\begin{align}
    \label{eq:3dEFapproxSym2}
    \includegraphics[width=0.95\columnwidth, valign=c]{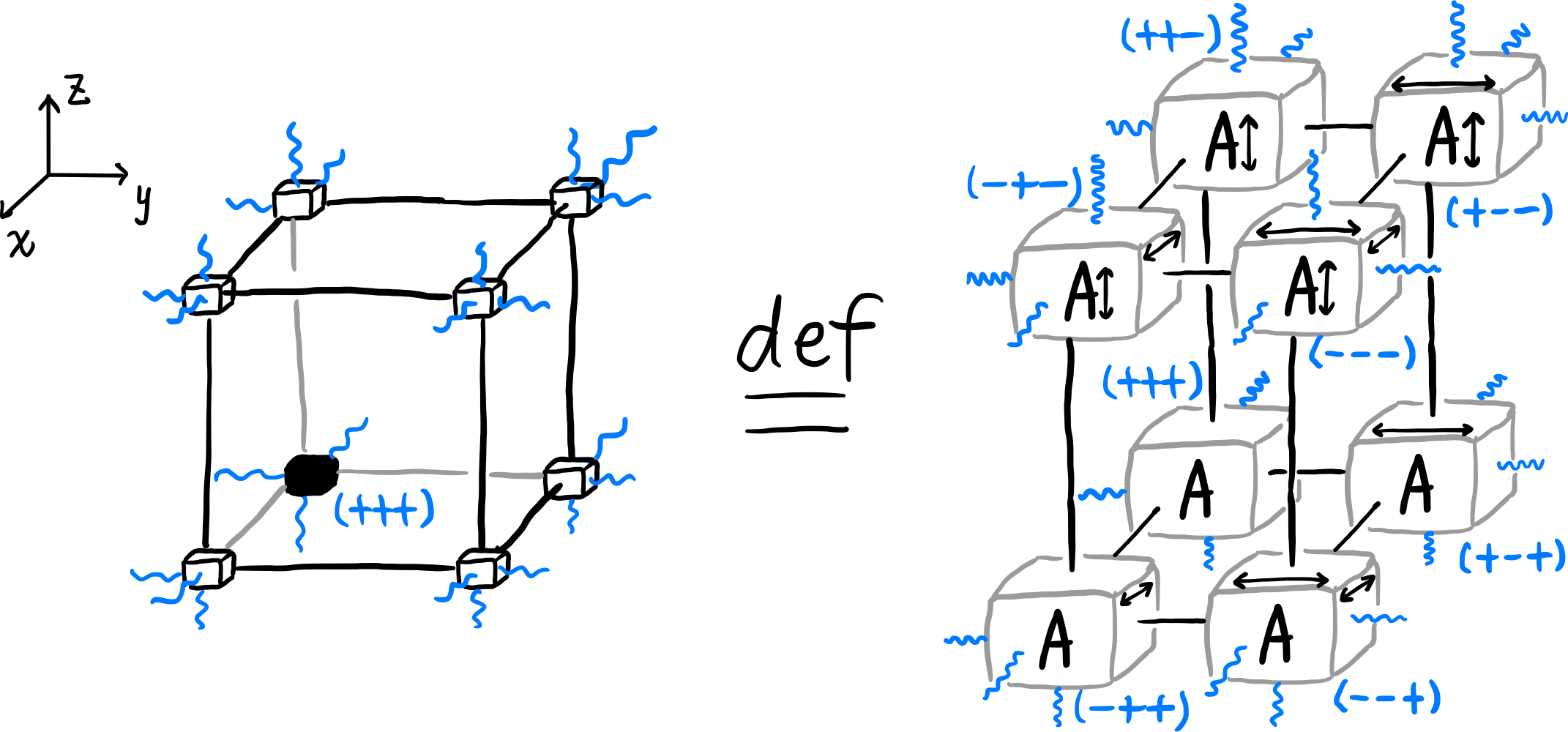}.
\end{align}
\end{subequations}

\subsection{Symmetry property of the isometric tensors\label{subsec:symIT}}
In this subsection, we show the symmetry property of the isometric tensors when the projective truncations are applied to the reflection-symmetric block after the transposition trick. 
The result is that an isometric tensor in projective truncations has the same symmetry property as the exact coarse-graining example we saw before in Eq.~\eqref{eq:isomChDir},
\begin{align}
    \label{eq:isomSym}
    \includegraphics[width=0.85\columnwidth, valign=c]{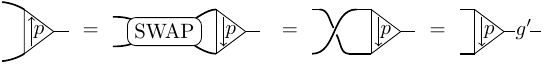},
\end{align}
where the first two equal signs come from the definition of the SWAP operator in Eq.~\eqref{eq:swapOp}, and the $g'$ is the SWAP-gauge matrix similar to that in Eq.~\eqref{eq:isomChDir}.
However, in the projective truncations, we will show that the SWAP-gauge matrix $g'$ is diagonal with diagonal entries $+1$ or $-1$.

\subsubsection{In 2D}
We will provide a proof of this property in 2D, whose generalization to 3D is easy to see.
After the transposition trick and the entanglement filtering, the $2 \times 2$ block-tensor transformation gives the following tensor RG equation,
\begin{subequations}
    \label{eq:2dEFRGeq}
\begin{align}
    \label{eq:2dEFRGeq1}
    \includegraphics[width=0.85\columnwidth, valign=c]{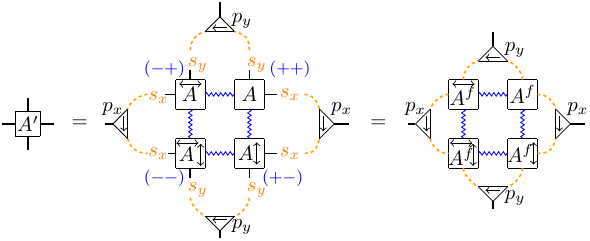},
\end{align}
where the filtered tensor $A^f$ in the last diagram is defined as
\begin{align}
    \label{eq:2dEFRGeq2}
    \includegraphics[scale=1.0, valign=c]{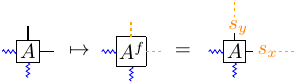}\quad.
\end{align}
\end{subequations}
When there is no EF, the two filtering matrices $s_x$ and $s_y$ are both the identity matrix. 
It suffices to show the proof for $p_x$, since the proof for $p_y$ is the same.
According to the method of projective truncations explained in~\autoref{sec:projtrunc}, the isometric tensor $p_x$ contains the eigenvectors with the first few largest eigenvalues of the following density matrix according to Eqs.~\eqref{eq:projApprox},~\eqref{eq:envAA} and~\eqref{eq:densityM},
\begin{align}
    \label{eq:rhopx}
    \includegraphics[scale=1.0, valign=c]{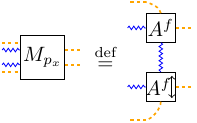},\nonumber\\
    \includegraphics[scale=1.0, valign=c]{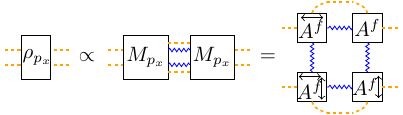},
\end{align}
where we drop the overall normalization factor in Eq.~\eqref{eq:densityM}. 
By dragging the tensor-network diagram around like the proof in Eq.~\eqref{eq:UpsilonLW}, one can show that the density matrix in Eq.~\eqref{eq:rhopx} has the following lattice-reflection symmetry:
\begin{align}
    \label{eq:rhopxSym}
    \includegraphics[width=0.85\columnwidth, valign=c]{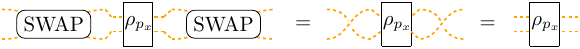}\;.
\end{align}
Since the SWAP matrix commutes with the density matrix $\rho_{p_x}$, the eigenvectors of the density matrix can be made to be the eigenvectors of the SWAP matrix. 
Moreover, eigenvalues of the SWAP matrix can only be $\pm 1$ since it squares to the identity.
This concludes the proof of the reflection symmetry of the isometric tensor in Eq.~\eqref{eq:isomSym}.

An easy way to determine the SWAP-gauge matrix $g'$ corresponding to an isometry $p$ (see Eq.~\eqref{eq:isomSym}) in numerical calculation is\footnote{\label{fn:EVDpSWAP}
    For this equation to be valid, it is necessary to simultaneously diagonalize the SWAP operator and the density matrix $\rho_p$ of the isometric tensor $p$ numerically.
    One concern is when the spectrum of $\rho_p$ contains degenerate eigenvalues.
    To deal with this concern, one numerical trick is diagonalizing $\rho_p + \epsilon \cdot \text{SWAP}$ for a number $\epsilon$ small enough so that the perturbation $\epsilon \cdot \text{SWAP}$ does not change the order of the eigenvalue spectrum of $\rho_p$.
}
\begin{align}
    \label{eq:p2g}
    \includegraphics[width=0.85\columnwidth, valign=c]{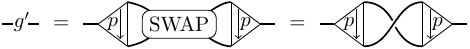}\;.
\end{align}
In 2D, there are two independent isometric tensors for two directions.
Therefore, there are two SWAP-gauge matrices: $g_x'$ and $g_y'$.

We would like to point out the advantage of the transposition trick in the projective truncations.
What happens to the above proof if there is no transposition trick?
If the EF process is turned off, the matrix that commutes with the density matrix $\rho_{p_x}$ in Eq.~\eqref{eq:rhopxSym} would be the SWAP operator multiplied by the SWAP-gauge $g_x$ associated with the tensor $A$.
The SWAP-gauge matrix associated with the coarse-grained tensor $A'$ would be determined by having two additional $g_x$ inserted between two isometric tensors in Eq.~\eqref{eq:p2g}, which was conjectured as the proper way to renormalize the SWAP operator for the HOTRG when people studied how to implement the boundary condition of non-orientable surfaces in TNRG~\cite{Shimizu:2024}.
However, when the EF process is turned on, it becomes necessary to first study some additional symmetry property of the $s_x$ filtering matrix in order to show the symmetry of an isometric tensor in Eq.~\eqref{eq:isomSym}.
Therefore, the transposition trick greatly simplifies the implementation and the proof of the lattice-reflection symmetry for the isometric tensors.

\subsubsection{In 3D}

In 3D, after the transposition trick, the $2 \times 2 \times 2$ block-tensor transformation gives the following tensor RG equation (we refrain from incorporating the EF at this point to make the tensor-network diagrams less clumsy), if we use an HOTRG-like block-tensor map with an arbitrary choice of order of collapses in the HOTRG to be $z \rightarrow y \rightarrow x$,
\begin{subequations}
    \label{eq:3dRGeq}
\begin{align}
    \label{eq:3dRGeq1}
    \includegraphics[width=0.83\columnwidth, valign=c]{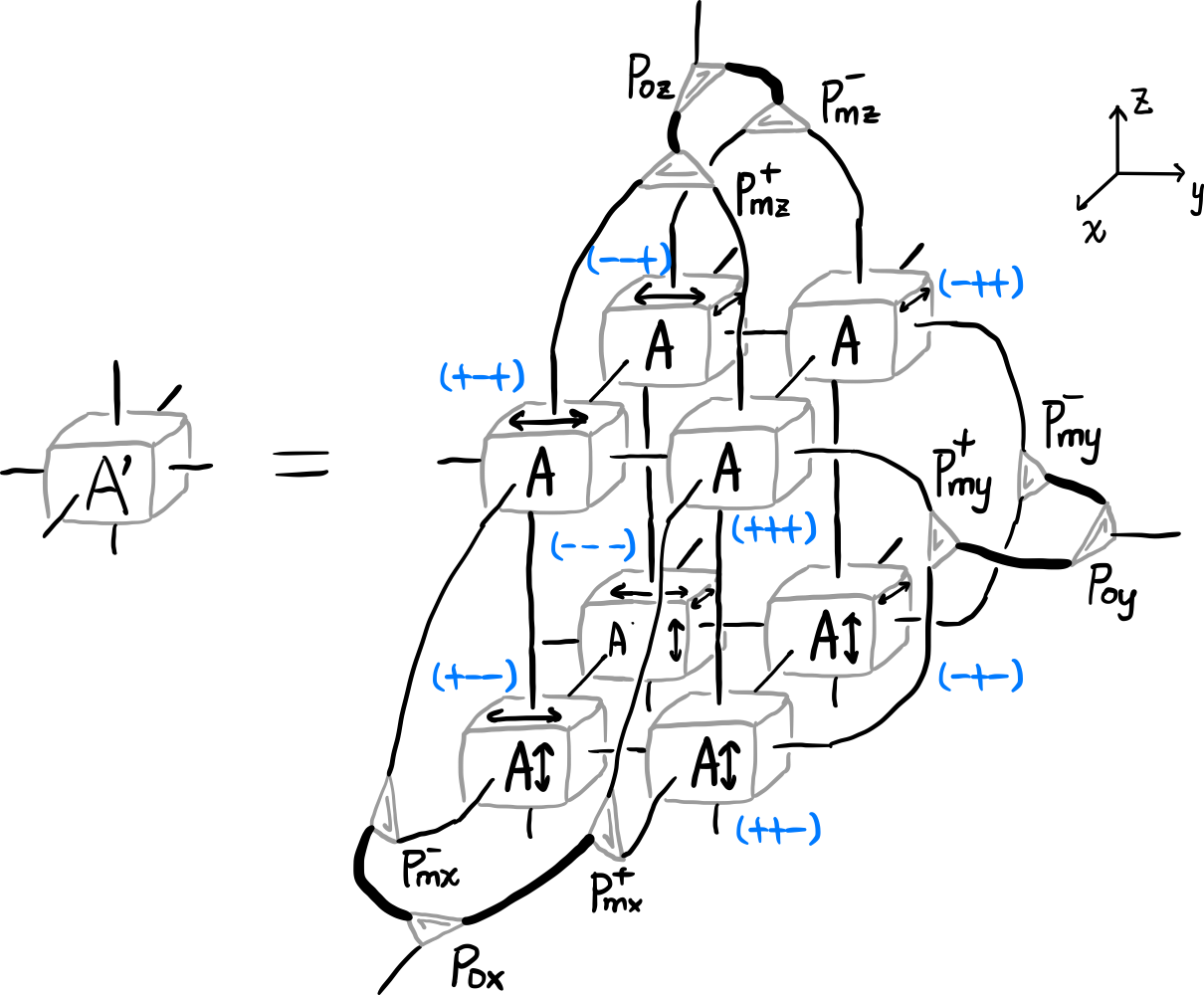},
\end{align}
with
\begin{align}
    \label{eq:3dRGeq2}
    p_{mx}^- = p_{mx}^+ \texteq{def} p_{mx}, \nonumber\\
    p_{my}^- = p_{my}^+ \texteq{def} p_{my},\\
    p_{mz}^- = p_{mz}^+ \texteq{def} p_{mz}.\nonumber
\end{align}
where we refrain from drawing the isometric tensors of inner legs like $p_i$ in Eq.~\eqref{eq:innerppt}.
\end{subequations}
We can impose $p_{mx}^-=p_{mx}^+$ because the following reason. 
It is enough to consider how their density matrices are related to each other. 
Notice that the only difference between two density matrices is that $y$ legs are transposed. 
But all $y$ legs are dummy indices in both density matrices, so two density matrices are identical, leading to $p_{mx}^- = p_{mx}^+$. 
The same is true for the intermediate isometric tensors in the other two directions: $p_{my}^-=p_{my}^+$ and $p_{mz}^-=p_{mz}^+$.
Therefore, we can drop the $+-$ subscript on them.

All the isometric tensors in Eq.~\eqref{eq:3dRGeq} have the lattice-reflection symmetry in Eq.~\eqref{eq:isomSym} due to the same reason that the corresponding density matrices have the lattice-reflection symmetry as $\rho_{p_x}$ in Eq.~\eqref{eq:rhopxSym}. 

The SWAP-gauge matrices of the coarse-grained tensor are determined from these isometric tensors in a slighted different way from the 2D case. 
To be concrete, we focus on the two SWAP-gauge matrices associated with the $x$ leg of the coarse-grained tensor $A'$. 
Just as before, in our notation, $g_{zx}'$ arises due to the reflection across the $z$-plane, while $g_{yx}'$ arises due to the reflection across the $y$-plane. 
They are determined from the isometric tensors according to
\begin{subequations}
    \label{eq:p2g3D}
\begin{align}
    \label{eq:p2gzx}
    &\includegraphics[width=0.83\columnwidth, valign=c]{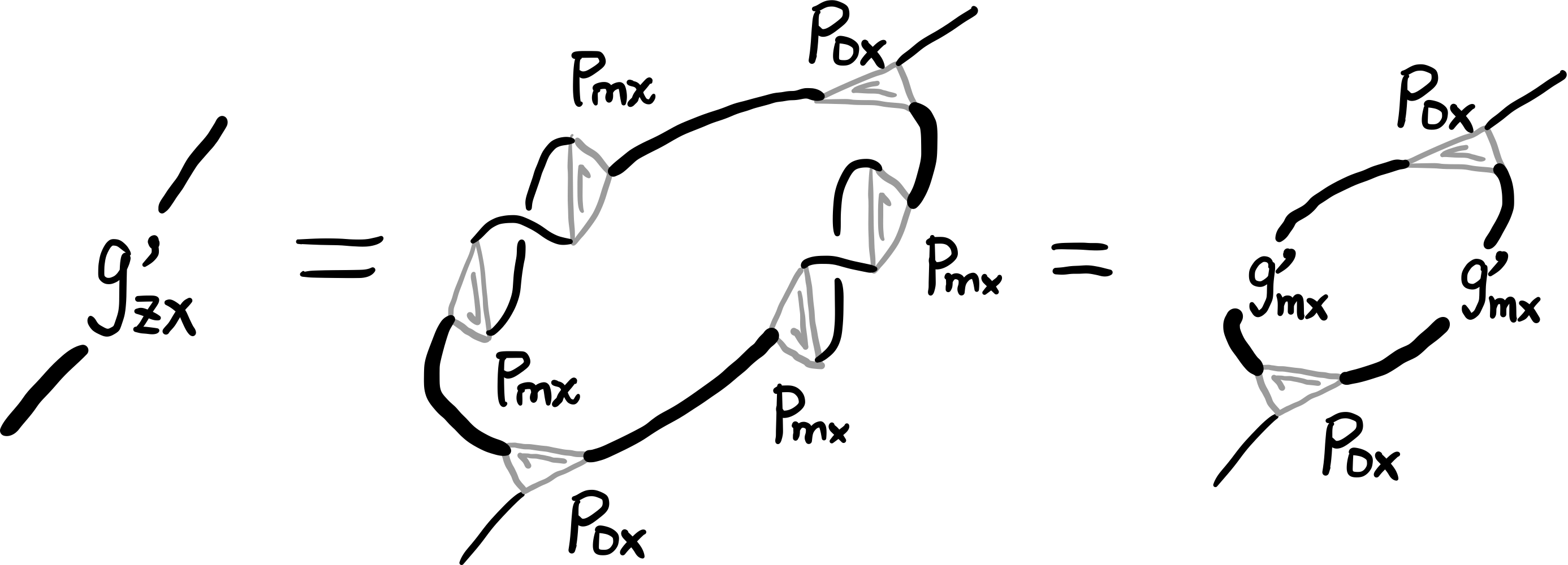} \nonumber\\
    \text{with}\quad
    &\includegraphics[width=0.30\columnwidth, valign=c]{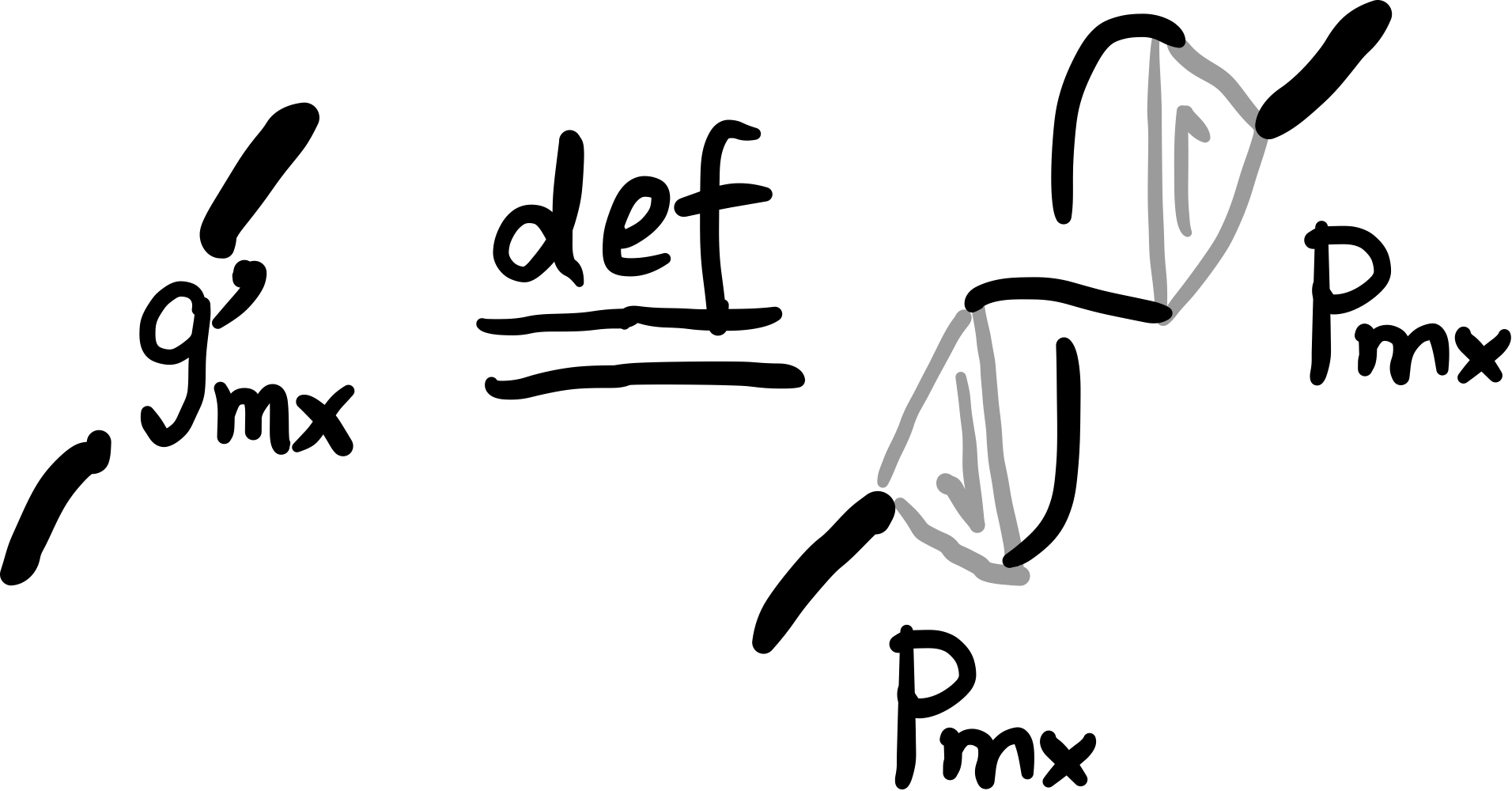}\quad,
\end{align}
and
\begin{align}
    \label{eq:p2gyx}
    \includegraphics[width=0.83\columnwidth, valign=c]{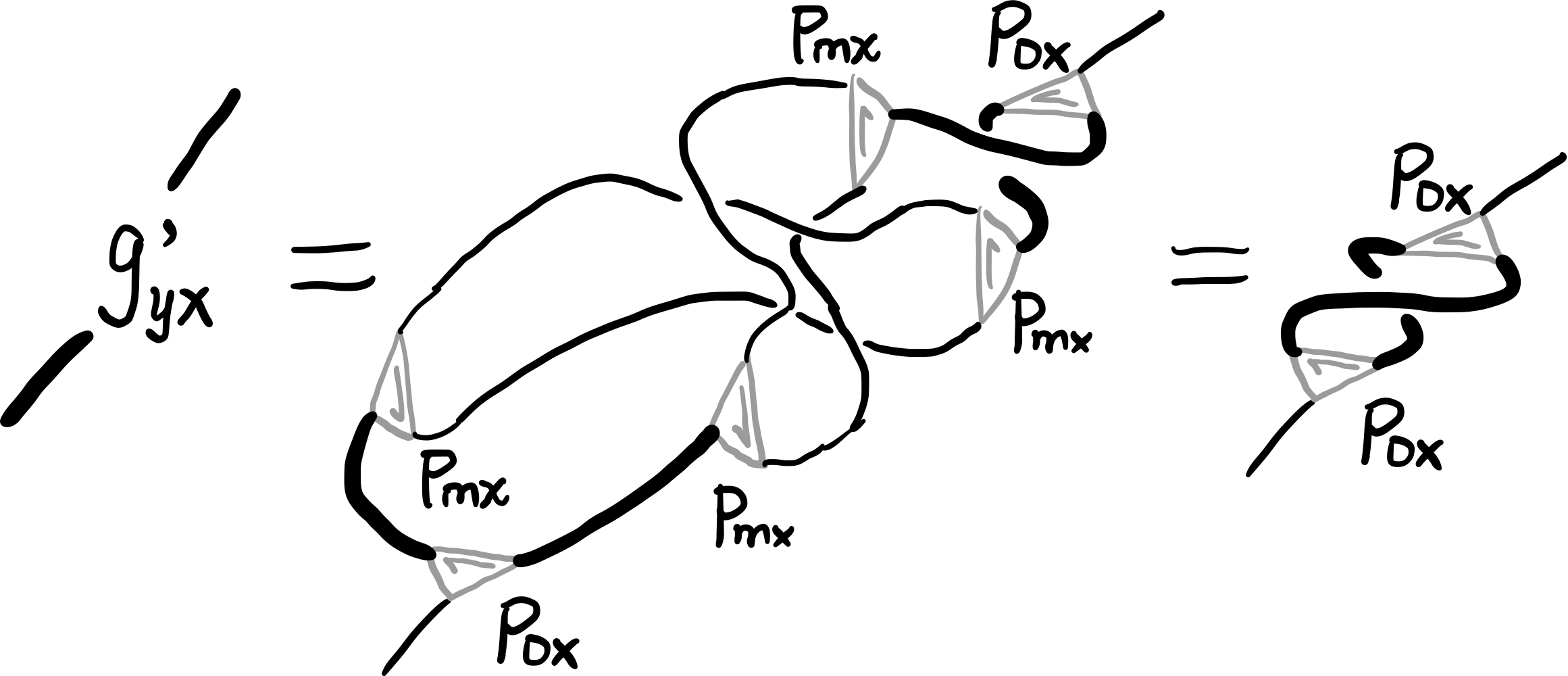},
\end{align}
\end{subequations}
Although the $g_{yx}'$ is determined from $p_{ox}$ in a same way as the 2D case in Eq.~\eqref{eq:p2g}, the $g_{zx}'$ is determined differently. 
The tricky point is how to implement the $z$-plane reflection for $p_{ox}$, since $p_{ox}$ fuses two legs in the $y$ direction, not the $z$ direction. 
The transposition of the two input legs of $p_{ox}$ can only implement the $y$-plane reflection.
The answer is that the SWAP-gauge matrix $g_{mx}'$ associated with the isometric tensor $p_{mx}$ in Eq.~\eqref{eq:p2gzx} implements this $z$-plane reflection. 
The above results can be derived using the same diagrammatic manipulation when we show the origin of the SWAP-gauge matrix in~\autoref{subsec:originSWAP}, as well as the lattice-reflection symmetry of the isometric tensor in Eq.~\eqref{eq:isomSym}.

When the EF is incorporated, the only change is that the tensor A in the $2 \times 2 \times 2$ block becomes the filtered $A^f$,
\begin{align}
    \label{eq:3dEFAf}
    \includegraphics[width=0.83\columnwidth, valign=c]{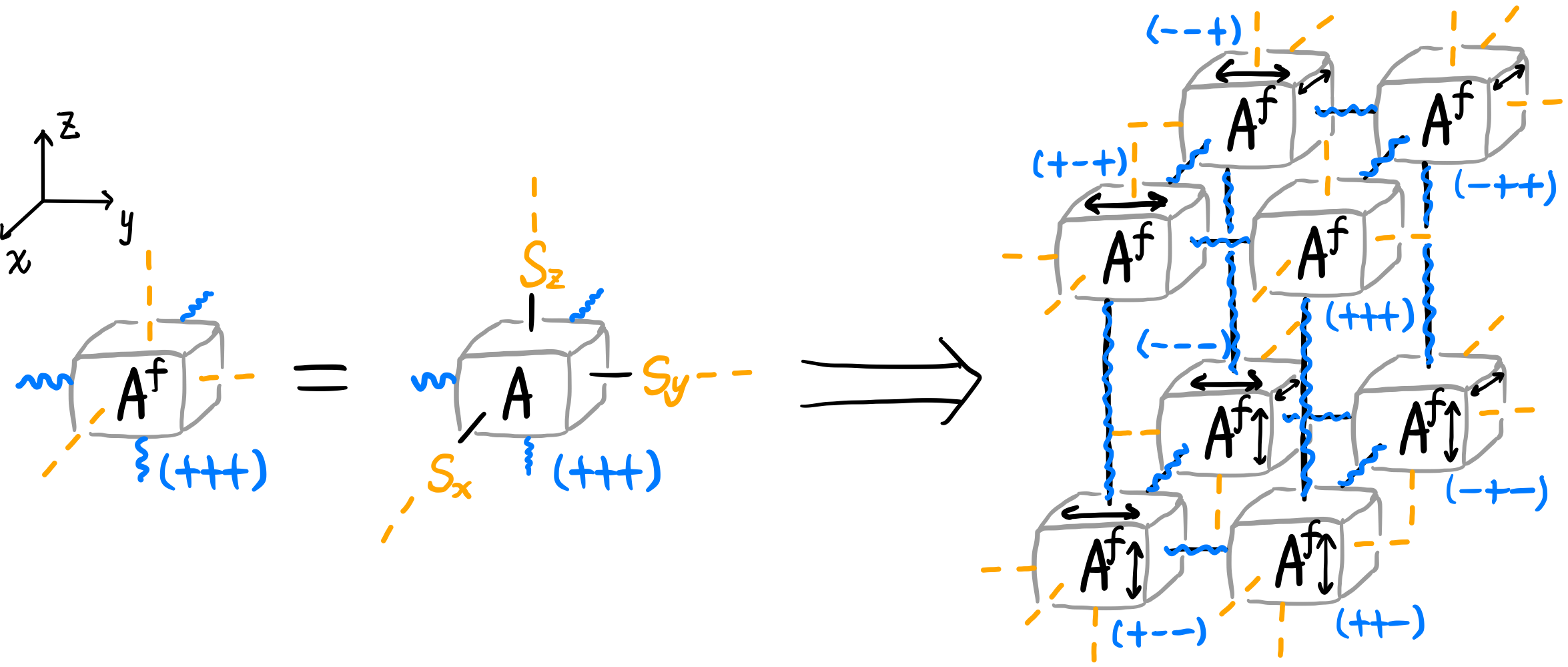}.
\end{align}
The symmetry of the isometric tensors and how to determine the SWAP-gauge matrices of the coarse-grained tensor remain the same as the case without the EF, except the change from $A$ to $A^f$.

\section{Algorithms of EF-enhanced TNRG that preserves lattice-reflection symmetry\label{sec:algo}}
Using the techniques explained and developed from~\autoref{sec:projtrunc} to~\autoref{sec:lattsym}, we write down the algorithms of the EF-enhanced TNRG with lattice-reflection symmetry exploited in both the 2D and 3D.

\subsection{The algorithm in 2D\label{subsec:algo2d}}
The partition function is the full contraction of a square-lattice tensor network consisting of copies of tensor $A$ as in Eq.~\eqref{eq:tn2Z}.
As has been presented in~\autoref{subsec:defSym} and~\autoref{subsec:originSWAP}, the lattice-reflection symmetry is defined in Eq.~\eqref{eq:refl2d}.

\begin{center}
    \underline{\emph{Step 1: Transposition trick}}
    \par
\end{center}
The following transposition trick is performed for the entire tensor network in Eq.~\eqref{eq:tn2Z}, which leaves the partition function invariant due to the symmetry property of $A$ in Eq.~\eqref{eq:refl2d},
\begin{align}
    \label{eq:2dtspZbig}
    Z =
    \includegraphics[scale=0.75, valign=c]{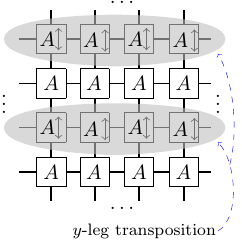}
    =
    \includegraphics[scale=0.75, valign=c]{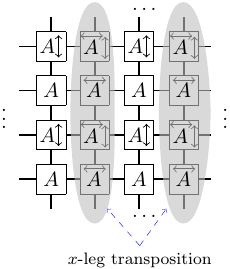}\;.
\end{align}
According to the principle of how to assemble the EF and the block-tensor map explained in~\autoref{subsec:assembly}, we choose the block-tensor patch and the target of the EF in the last tensor network in the above equation as
\begin{align}
    \label{eq:2dbkEF}
    \includegraphics[scale=0.8, valign=c]{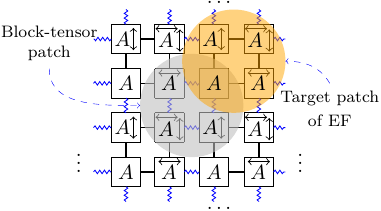}.
\end{align}
Notice that this \emph{Step 1} is done ``on paper'' and nothing needs to be done in numerical calculation. 
This step makes reflection-symmetric both the $2 \times 2$ block-tensor patch and the $2 \times 2$ target patch of the EF, which has been expounded in~\autoref{subsec:symFM} and~\autoref{subsec:symIT}.

\begin{center}
    \underline{\emph{Step 2: Entanglement filtering}}
    \par
\end{center}
The EF is applied to its target patch in Eq.~\eqref{eq:2dbkEF}.
The approximation for the EF is Eq.~\eqref{eq:2dEFapproxSym}.
We use the schemes developed in Appendix~\ref{app:findsmat} to determine the two filtering matrices $s_x$ and $s_y$ in Eq.~\eqref{eq:2dEFapproxSym}.

\underline{\emph{Step 2.1 is the initialization of $s_x$ and $s_y$}}
\par
Take $s_y$ as an example:
\begin{itemize}
    \item \underline{\emph{Step 2.1a:}}
        construct the $\Upsilon_0$ tensor defined in Eq.~\eqref{eq:UpsilonP0} for the initialization of $s_y$ using the target patch of the EF after the transposition trick shown in Eq.~\eqref{eq:2dbkEF}; 
        the resultant $\Upsilon_0$ tensor is shown in Eq.~\eqref{eq:UpsilonLE}. 
    \item \underline{\emph{Step 2.1b:}}
        the low-rank matrix $L_E$ related to $s_y$ is determined according to Eqs.~\eqref{eq:pinvDef} and~\eqref{eq:LEsol}.
    \item \underline{\emph{Step 2.1c:}}
        the low-rank matrix $L_E$ is split according to Eq.~\eqref{eq:LEsplit} to initialize $s_y$.
\end{itemize}
\underline{\emph{Remark:}} 
the other filtering matrix $s_x$ is initialized in the same way, as has been explained in Appendix~\ref{subsec:inits}.

\underline{\emph{Step 2.2 is the optimization of $s_x$ and $s_y$}}
\par
Take $s_y$ as an example.
The fidelity, according to its definition in Eq.~\eqref{eq:fidelityDef}, of the EF approximation in Eq.~\eqref{eq:2dEFapproxSym} is
\begin{subequations}
\begin{align}
   \label{eq:2dEFfide}
   F = 
   \includegraphics[scale=1.0, valign=c]{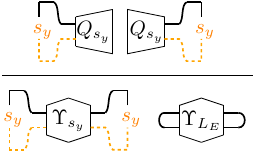}\quad,
\end{align} 
where the tensor $\Upsilon_{L_E}$ is constructed according to Eq.~\eqref{eq:UpsilonLE}, and the other two tensors $\Upsilon_{s_y}$ and $Q_{s_y}$, according to Appendix~\ref{subsec:opts}, are
\begin{align}
   \label{eq:Upsilonsy}
   \includegraphics[width=0.83\columnwidth, valign=c]{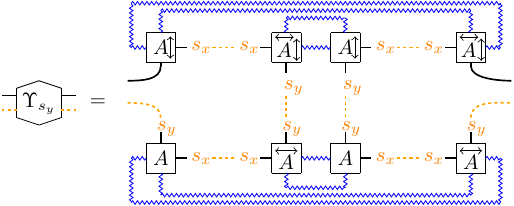}\;,
\end{align} 
\begin{align}
   \label{eq:Qsy}
   \includegraphics[width=0.83\columnwidth, valign=c]{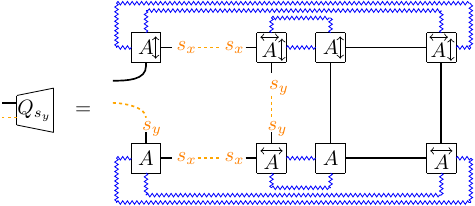}\;.
\end{align} 
\end{subequations}
The $s_y$ is then updated using
\begin{align}
   \label{eq:syUpdate}
   \includegraphics[scale=1.0, valign=c]{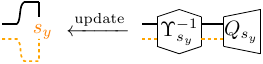}\quad.
\end{align} 
The usual inverse matrix can be used here for $\Upsilon_{s_y}^{-1}$. 
Since both $\Upsilon_{s_y}$ and $Q_{s_y}$ tensors change after the update of $s_y$, one can choose to update $s_y$ iteratively for several times, and then go on updating $s_x$ in the same manner. 
Numerically, the validity of this optimization process comes from checking whether the fidelity in Eq.~\eqref{eq:2dEFfide} increases with iteration. 
When the fidelity converges, one can stop the iteration.

\underline{\emph{Remark (better convergence of the optimization):}} 
Here we explain a trick buried in Evenbly’s codes~\cite{Evenbly:code} for a better convergence of the optimization of the filtering matrices when the lattice-reflection symmetry is exploited. 
Since the $\Upsilon_{s_y}$ and $Q_{s_y}$ tensors depend on the filtering matrix $s_y$, the updated $s_y$ in Eq.~\eqref{eq:syUpdate} cannot guarantee the growth of the fidelity. 
The following numerical trick can make sure that the fidelity does not decrease during the optimization:
\begin{itemize}
    \item Use Eq.~\eqref{eq:syUpdate} to propose a candidate $s_y'$.
    \item Build several convex combinations of this candidate and the old $s_y$,
        \begin{align}
            \label{eq:sytry}
            s_y^{\text{try}} = (1 - p) s_y' + p s_y,
        \end{align}
        where $p$ increases from 0 to 1.
    \item When $p$ increases, once a trial combination in Eq.~\eqref{eq:sytry} results in a growth of the fidelity, $s_y $ is updated to be this combination.
\end{itemize}

\begin{center}
    \underline{\emph{Step 3: Projective truncations}}
    \par
\end{center}
After the two filtering matrices $s_x$ and $s_y$ are determined, they are absorbed into the $2 \times 2$ block-tensor patch in Eq.~\eqref{eq:2dbkEF}, and the following projective truncations are performed,
\begin{align}
   \label{eq:2dEFRGproj}
   \includegraphics[width=0.83\columnwidth, valign=c]{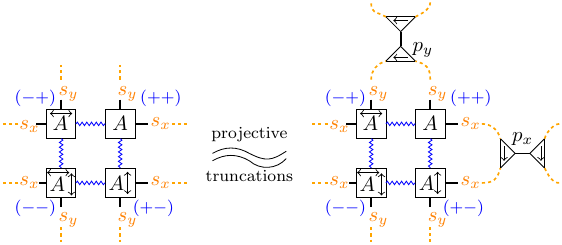}\;.
\end{align} 
The two isometric tensors $p_x$ and $p_y$ in the above equation  are determined according to~\autoref{subsec:proj}.
Again, let us use $p_x$ as an example.
\begin{itemize}
    \item \underline{\emph{Step 3.1:}} 
        act the two filtering matrices $s_x$ and $s_y$ on $A$ according to Eq.~\eqref{eq:2dEFRGeq2} to obtain the filtered tensor $A^f$.
    \item \underline{\emph{Step 3.2:}} 
        construct the density matrix, from which $p_x$ is obtained, according to Eq.~\eqref{eq:rhopx} and use the eigenvalue decomposition for a hermitian matrix to determine its eigenvalues and eigenvectors.
        The isometric tensor $p_x$ consists of the eigenvectors with the first $\chi$ largest eigenvalues.
\end{itemize}
The other isometric tensor $p_y$ is determined in the same way.

\begin{center}
    \underline{\emph{Step 4: Contraction of the tensor RG equation}}
    \par
\end{center}
Now, both the filtering matrices $s_x, s_y$ and the isometric tensors $p_x, p_y$ are determined, the final step is contracting the tensor RG equation in Eq.~\eqref{eq:2dEFRGeq} to obtain the coarse-grained tensor $A'$.
The bottleneck of the computational costs is this step, which are $O(\chi^8)$.

\underline{\emph{Remark (reduce the computational costs):}} 
\par
The ideas in the HOTRG-like block-tensor transformation in~\autoref{subsec:hotrg-like} can be used to reduce the computational costs of the contraction in Eq.~\eqref{eq:2dEFRGeq} by inserting an additional projection operation $p_i p_i^\intercal$ into the inner legs of the $2 \times 2$ block, with the bond dimension of the third leg of $p_i$ to be $\chi_i$.
The resultant tensor RG equation is
\begin{align}
   \label{eq:2dEFhotRGeq}
   \includegraphics[scale=1.0, valign=c]{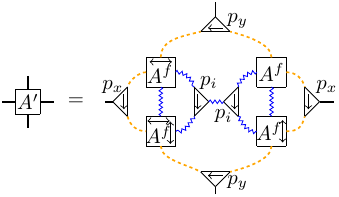}\;,
\end{align} 
where $A^f$ is obtained by acting filtering matrices $s_x, s_y$ on $A$, as is shown in Eq.~\eqref{eq:2dEFRGeq2}.
The environment $M$ in the density matrix $\rho$ (see Eqs.~\eqref{eq:projApprox},~\eqref{eq:envAA} and~\eqref{eq:densityM}) for the $p_i$ can be chosen to be
\begin{align}
   \label{eq:M4pi}
   \includegraphics[scale=1.0, valign=c]{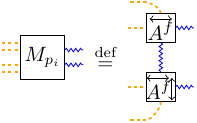}.
\end{align} 
Due to the transposition trick, the isometry $p_i$ also has the lattice-reflection symmetry in Eq.~\eqref{eq:isomSym}; 
therefore the direction of its arrow is immaterial. 
Notice that inserting $p_i p_i^\intercal$ does not change the lattice-reflection symmetry property of the coarse-grained tensor $A'$. 
When the bond dimension $\chi_i$ of $p_i$ is equal to $\chi^2$, the scheme becomes the full block-tensor scheme in Eq.~\eqref{eq:2dEFRGeq}.

The tensor contraction in Eq.~\eqref{eq:2dEFhotRGeq} can be performed as a composition of two collapses in two directions, $A^f \to A^f_{y} \to A'$,
\begin{subequations}
   \label{eq:2dEFhotRGeq2cols}
\begin{align}
   \label{eq:2dEFhotRGycol}
   \includegraphics[scale=1.0, valign=c]{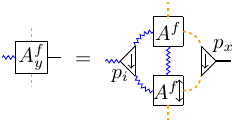}\quad,
\end{align} 
\begin{align}
   \label{eq:2dEFhotRGxcol}
   \includegraphics[scale=1.0, valign=c]{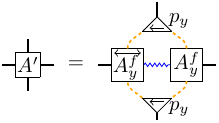}\quad.
\end{align}  
\end{subequations}
The inner bond dimension $\chi_i$ should be large enough to make sure that the error of its projective truncation is smaller than those of $p_x$ and $p_y$.
The choice of $\chi_i = \chi$ makes the computational costs of the contraction in Eq.~\eqref{eq:2dEFhotRGeq2cols} be $O(\chi^7)$, the same as those of the usual HOTRG.
If $\chi_i = \chi^2$, the computational costs go back to those of the full contraction in Eq.~\eqref{eq:2dEFRGeq}, which are $O(\chi^8)$.

\begin{center}
    \underline{\emph{The lattice-reflection symmetry is preserved}}
    \par
\end{center}
By inspecting the tensor RG equation in Eq.~\eqref{eq:2dEFhotRGeq}, along with the symmetry property of the isometric tensors $p_x, p_y$ shown in Eqs.~\eqref{eq:isomSym} and~\eqref{eq:p2g}, it is easy to see that the coarse-grained tensor $A'$ satisfies the same lattice-reflection symmetry as the original tensor $A$ in Eq.~\eqref{eq:refl2d}, with the renormalized SWAP-gauge matrices $g'_x$ and $g'_y$ determined from $p_x$ and $p_y$ according to
Eq.~\eqref{eq:p2g}.
Since the symmetry property of the $A'$ is true regardless of the symmetry of the original tensor $A$ in Eqs.~\eqref{eq:2dEFhotRGeq} and~\eqref{eq:2dEFRGeq2}, the lattice-reflection symmetry is also imposed.

\subsection{The algorithm in 3D\label{subsec:3dalgo}}
In this subsection, we will expound the 3D algorithm first proposed in Ref.~\cite{Lyu:Kawashima:2024}.
The partition function is a full contraction of a cubic-lattice tensor network consisting of copies of tensor $A$.
This tensor network can be generated by repeating the first $2 \times 2 \times 2$ block in Eq.~\eqref{eq:tsptrick3D} in all three directions of the space.
The lattice-reflection symmetry is manifested in the tensor $A$ as Eq.~\eqref{eq:refl3d} for the reflection across the $z$-plane. 
The reflections across the $y$-plane and $x$-plane have similar definitions. 

The 3D algorithm is a straightforward generalization of the 2D algorithm in~\autoref{subsec:algo2d}. 
For this reason, we lay out the big picture of the algorithm and avoid too much detailed explanation. 
The readers are encouraged to understand the 2D algorithm first to develop intuition for understanding the 3D one.

\begin{center}
    \underline{\emph{Step 1: Transposition trick}}
    \par
\end{center}
Since a panorama of the transposition trick like the 2D one in Eq.~\eqref{eq:2dtspZbig} would look clumsy in 3D, we focus on the basic $2 \times 2 \times 2$ building block of the tensor network. 
The transposition trick acts on this block as in Eq.~\eqref{eq:tsptrick3D}, which we choose as the block-tensor patch. 
According to the principle for assembling the EF process and a block-tensor map in~\autoref{fig:sec2-3dblockEE}, the target patch of the EF is shown in Eq.~\eqref{eq:3dEFapproxSym2} after this transposition trick.

\begin{center}
    \underline{\emph{Step 2: Entanglement filtering}}
    \par
\end{center}
The EF is applied to the target patch of the EF shown in~\autoref{fig:sec2-3dblockEE}, whose detailed view is in Eq.~\eqref{eq:3dEFapproxSym2}. 
The approximation of the EF is Eq.~\eqref{eq:3dEFapproxSym1}. 
We use the schemes developed in Appendix~\ref{app:findsmat} to determine the three filtering matrices $s_x$, $s_y$ and $s_z$ in Eq.~\eqref{eq:3dEFapproxSym}. 

\underline{\emph{Step 2.1 is the initialization of $s_x, s_y$ and $s_z$}}
\par
The only difference from the \emph{Step 2.1} of the 2D algorithm in~\autoref{subsec:algo2d} is how to construct the $\Upsilon_0$ tensors for these three filtering matrices in \emph{Step 2.1a}. 
Once a $\Upsilon_0$ tensor is constructed, the same \emph{Step 2.1b} and \emph{Step 2.1c} as the 2D algorithm can be applied.
We will briefly describe how to construct the $\Upsilon_0$ tensor for the initialization of $s_y$ in \emph{Step 2.2} below.

\underline{\emph{Step 2.2 is the optimization of $s_x, s_y$ and $s_z$}}
\par
Take $s_y$ as an example. 
The fidelity $F$ is built according to its definition in Eq.~\eqref{eq:fidelityDef} and the EF approximation in Eq.~\eqref{eq:3dEFapproxSym}. 
For optimization of $s_y$, we rewrite the fidelity in the same form as Eq.~\eqref{eq:2dEFfide}, with different expressions\footnote{
Here, $\Upsilon_{L_E}$ is the $\Upsilon_0$ tensor for the initialization of $s_y$.
}
for $\Upsilon_{s_y}$, $Q_{s_y}$ and $\Upsilon_{L_E}$. 
Among these three tensors, $\Upsilon_{s_y}$ and $Q_{s_y}$ can be read off from the expression of the $F$, while $\Upsilon_{L_E}$ is obtained by setting all $s_x, s_y$ and $s_z$ in $\Upsilon_{s_y}$ to the identity matrix\footnote{
To see why it is so, try comparing them in the 2D algorithm, where $\Upsilon_{s_y}$ is in Eq.~\eqref{eq:Upsilonsy} and $\Upsilon_{L_E}$ is in Eq.~\eqref{eq:UpsilonLE}. 
}.

Once $\Upsilon_{s_y}$ and $Q_{s_y}$ are constructed, the $s_y$ is updated according to Eq.~\eqref{eq:syUpdate}, the same process as the 2D algorithm. 
The same remark and trick apply for a better convergence of the optimization process as the 2D algorithm (see the explanation around Eq.~\eqref{eq:sytry}).

\begin{center}
    \underline{\emph{Step 3: Projective truncations}}
    \par
\end{center}
After the three filtering matrices $s_x, s_y$ and $s_z$ are determined, they are absorbed into the block-tensor patch, which changes from the last diagram in Eq.~\eqref{eq:tsptrick3D} to the block in Eq.~\eqref{eq:3dEFAf}, leading to a map $A \mapsto A^f$.
Then, an HOTRG-like block-tensor transformation (see~\autoref{subsec:hotrg-like}) is applied to the filtered block in Eq.~\eqref{eq:3dEFAf}, whose order of the HOTRG collapses is chosen arbitrarily to be $z \rightarrow y \rightarrow x$.

Take the first collapse in $z$ direction as an example. 
Let us focus on the two tensors located at $(+++)$ and $(++-)$ position in Eq.~\eqref{eq:3dEFAf}. 
The approximation of the projective truncation is
\begin{align}
   \label{eq:3dprojz}
   \includegraphics[width=0.80\columnwidth, valign=c]{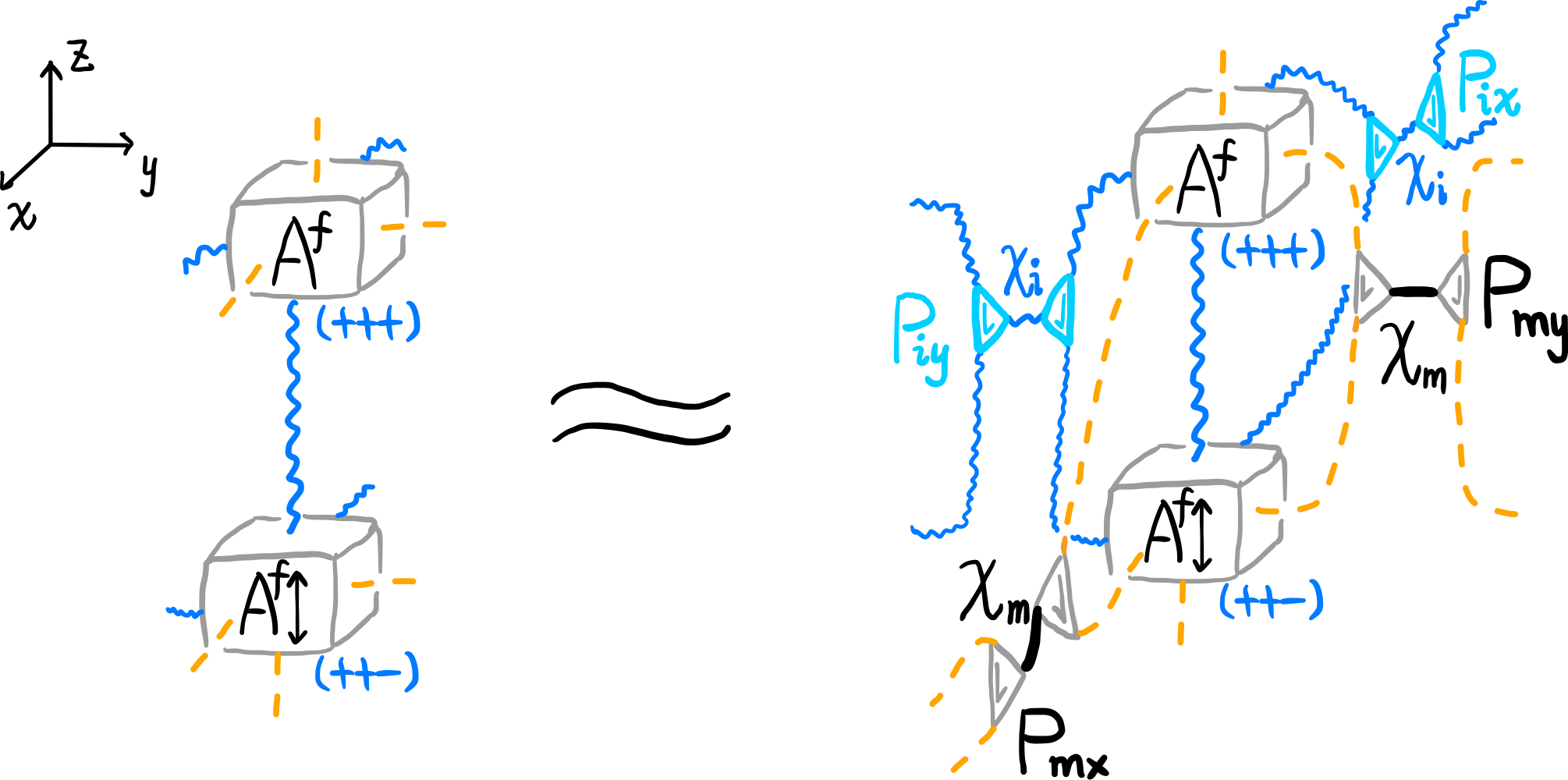}.
\end{align}  
There are four isometric tensors $p_{mx}, p_{my}, p_{ix}, p_{iy}$, where the first two are for fusing outer legs of the block (represented by dashed lines in Eq.~\eqref{eq:3dprojz}) into intermediate legs, while the last two are for inner legs of the block (represented by wavy lines in Eq.~\eqref{eq:3dprojz}). 
The bond dimension of the third leg of $p_{mx}$ and $p_{my}$ is $\chi_m$, while that of $p_{ix}$ and $p_{iy}$ is $\chi_{i}$.
Here, the subscripts $m,i$ denote ``intermediate'' and ``inner''.

The four isometric tensors in Eq.~\eqref{eq:3dprojz} are determined using the scheme explained in~\autoref{subsec:proj}. 
Take $p_{my}$ as an example. 
It is determined by constructing the following density matrix,
\begin{align}
   \label{eq:3drho4pmy}
   \includegraphics[width=0.50\columnwidth, valign=c]{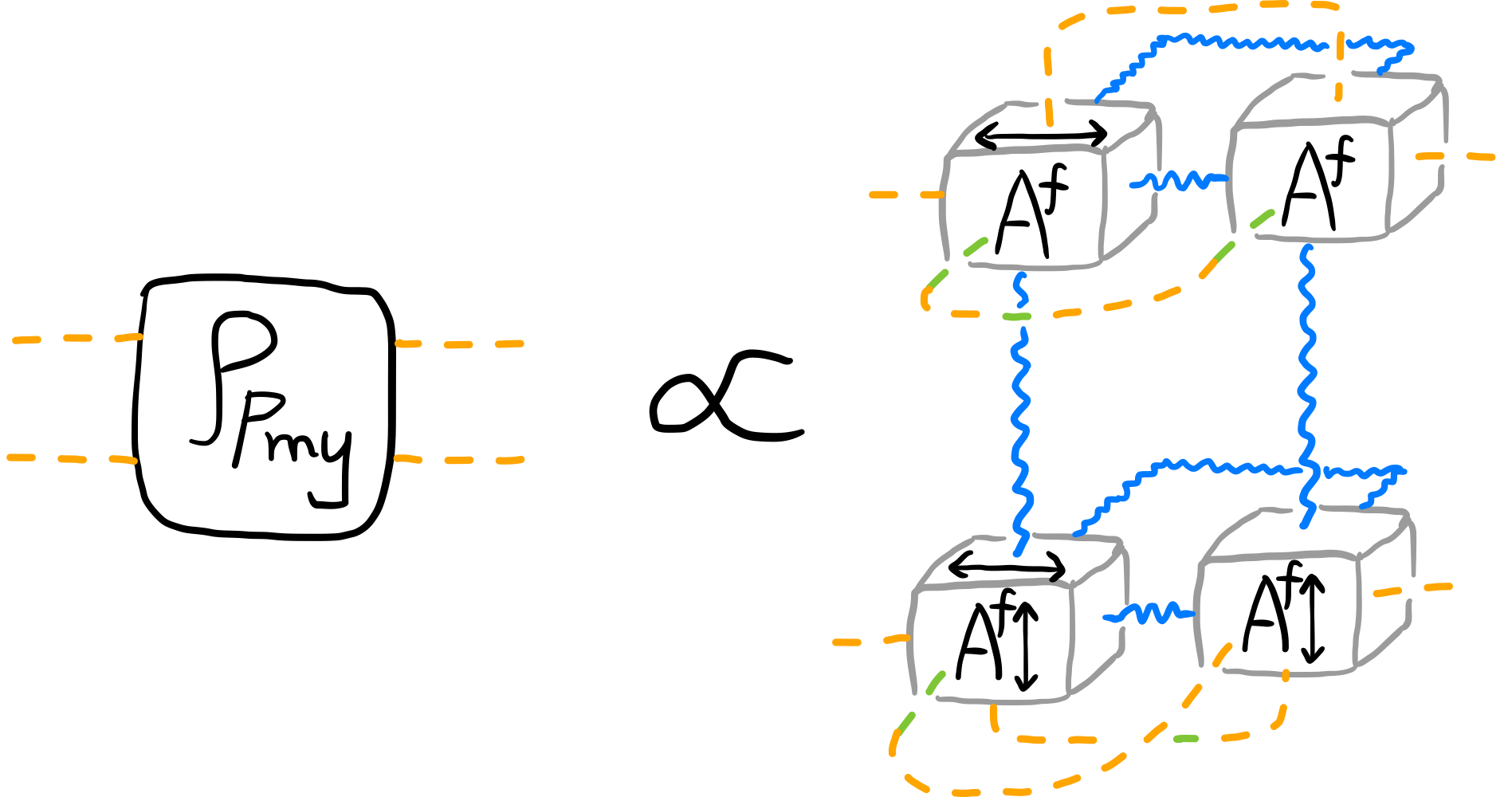}.
\end{align}
The isometric tensor $p_{my}$ consists of the eigenvectors of this density matrix with the first $\chi_m$ largest eigenvalues. 
The other three isometric tensors are determined in the same way.

\begin{center}
    \underline{\emph{Step 4: Contraction of the tensor RG equation}}
    \par
\end{center}

\underline{\emph{Step 4.1 is the $z$ collapse}}
\par
After the four isometric tensors $p_{mx}, p_{my}, p_{ix}, p_{iy}$ are determined, the following contraction of the $z$ collapse is carried out,
\begin{subequations}
\begin{align}
   \label{eq:3dzcol}
   \includegraphics[width=0.80\columnwidth, valign=c]{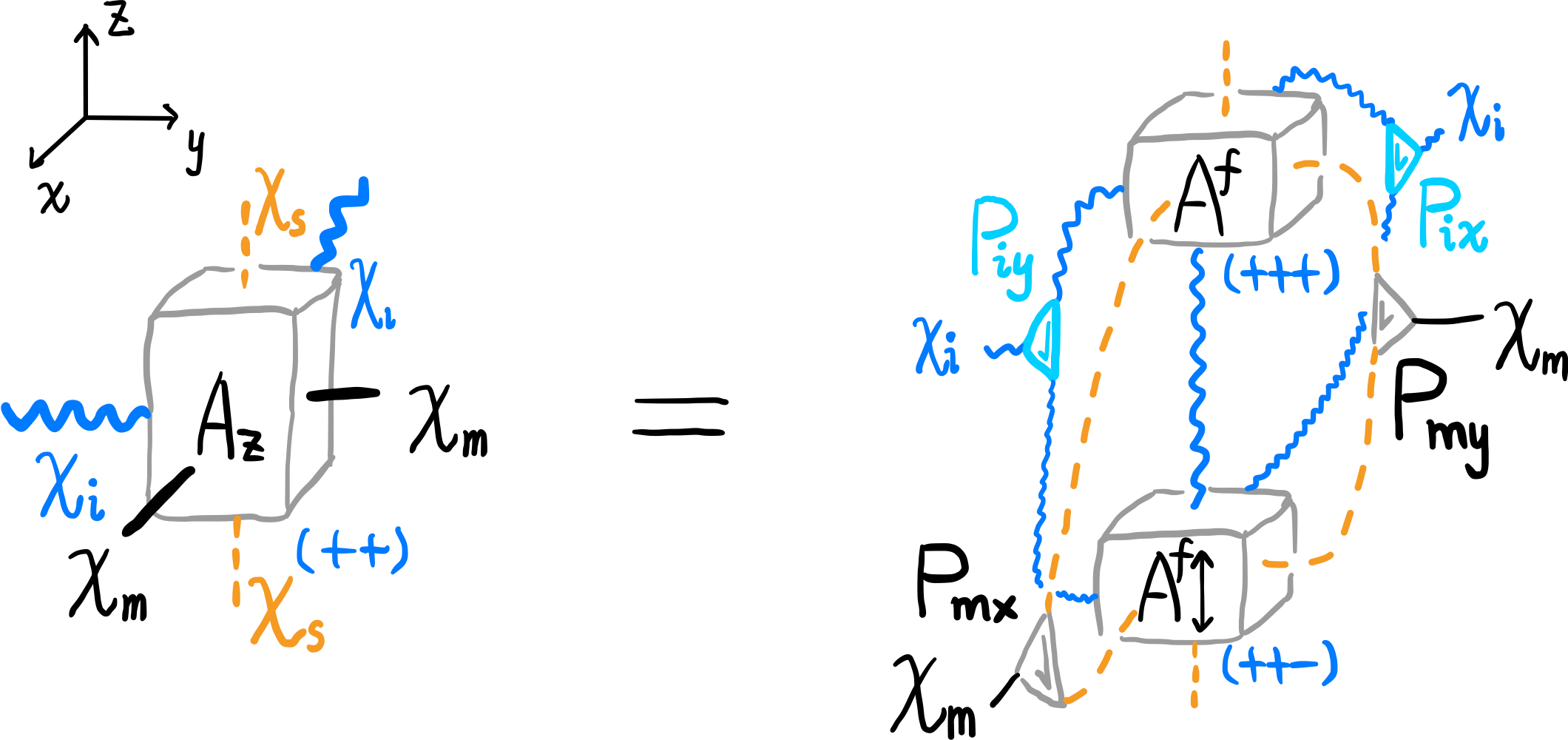}\quad,
\end{align}
which is a map $A^f \mapsto A_z$.
Due to the $z$ collapse, the $2 \times 2 \times 2$ block-tensor patch in Eq.~\eqref{eq:3dEFAf} changes as
\begin{align}
   \label{eq:3dzcolBlock}
   \includegraphics[width=0.85\columnwidth, valign=c]{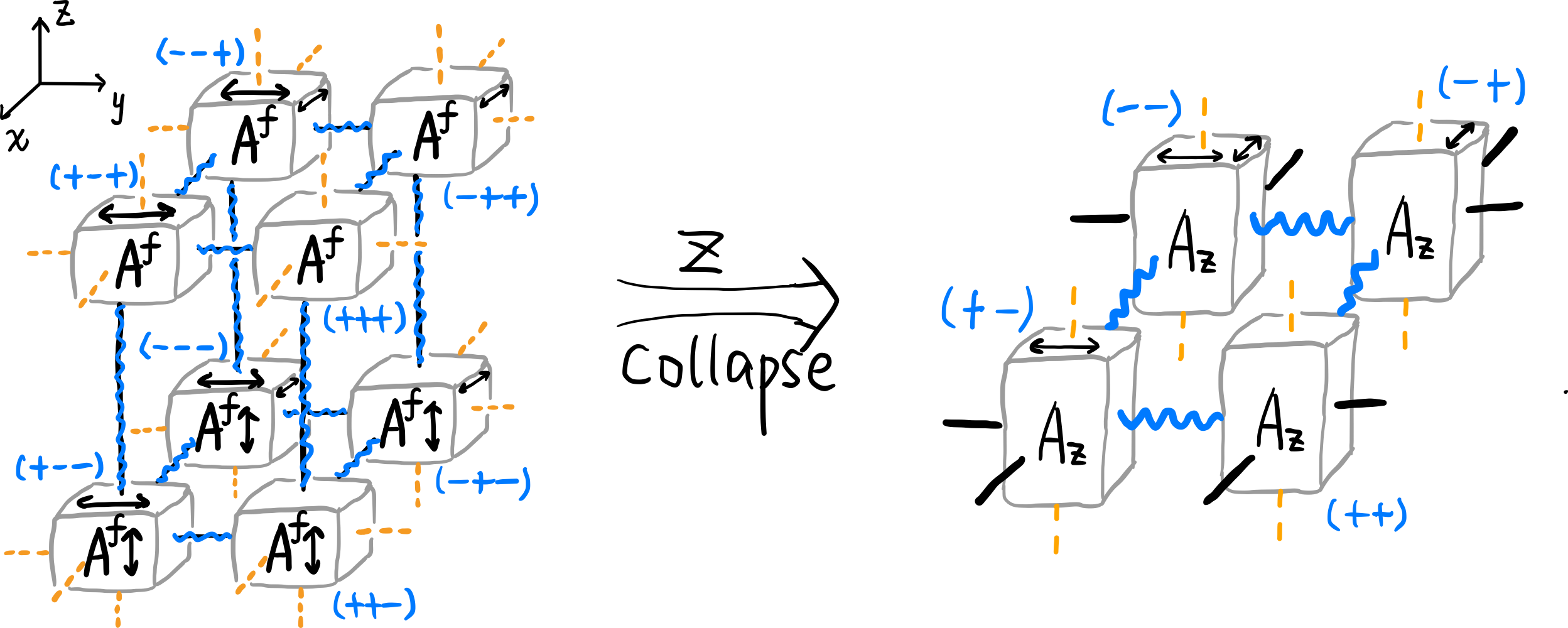}.
\end{align}
\end{subequations}

\underline{\emph{Step 4.2 is the $y$ collapse}}
\par
After the $z$ collapse, the $y$ collapse is performed to $A_z$ in the last diagram of Eq.~\eqref{eq:3dzcolBlock},
\begin{subequations}
\begin{align}
   \label{eq:3dycol}
   \includegraphics[width=0.80\columnwidth, valign=c]{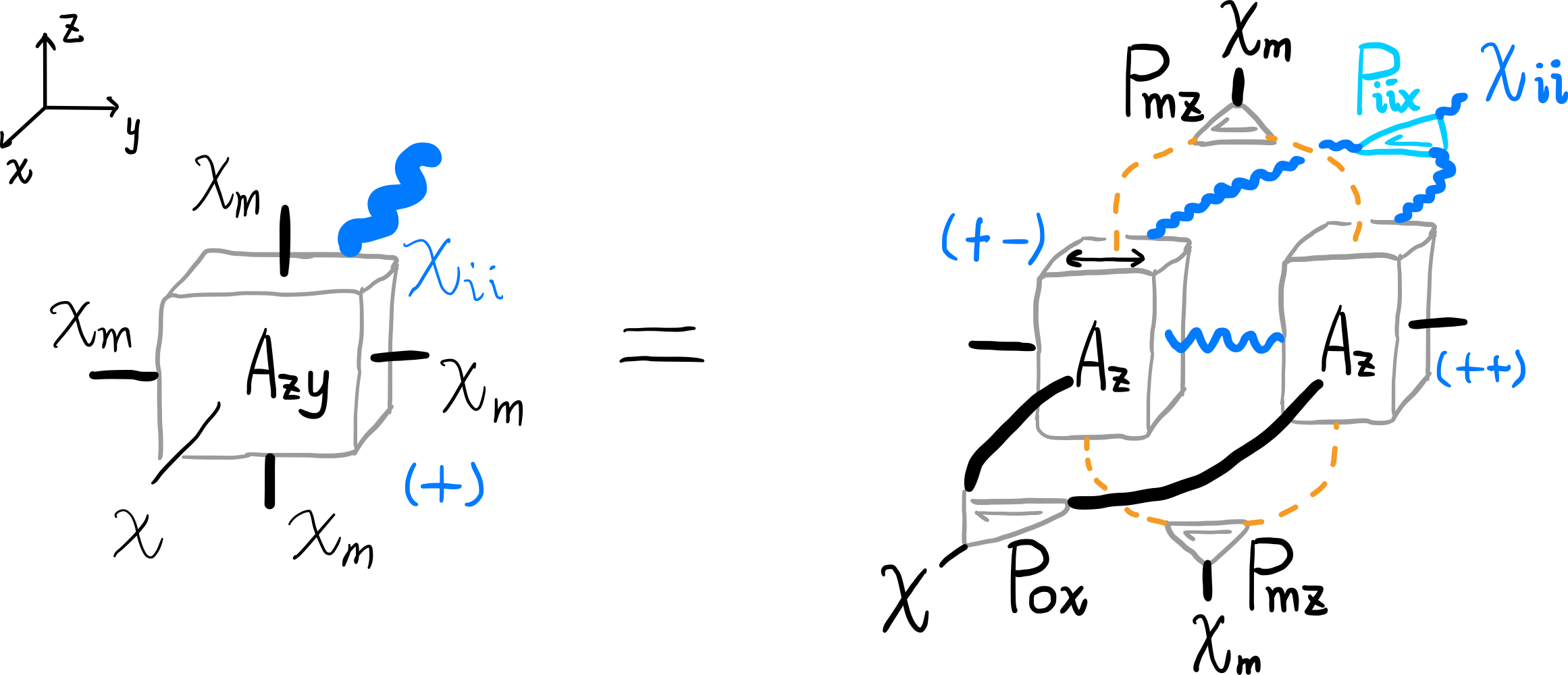}\quad,
\end{align}
which is a map $A_z \mapsto A_{zy}$.
Three isometric tensors $p_{mz}, p_{ox}$ and $p_{iix}$ are involved, and they can be determined in the same way as the \emph{Step 3} of the 3D algorithm by constructing proper density matrices from $A_z$.
Due to the $y$ collapse, the $2 \times 2$ tensor network in the last diagram of Eq.~\eqref{eq:3dzcolBlock} changes as
\begin{align}
   \label{eq:3dycolBlock}
   \includegraphics[width=0.85\columnwidth, valign=c]{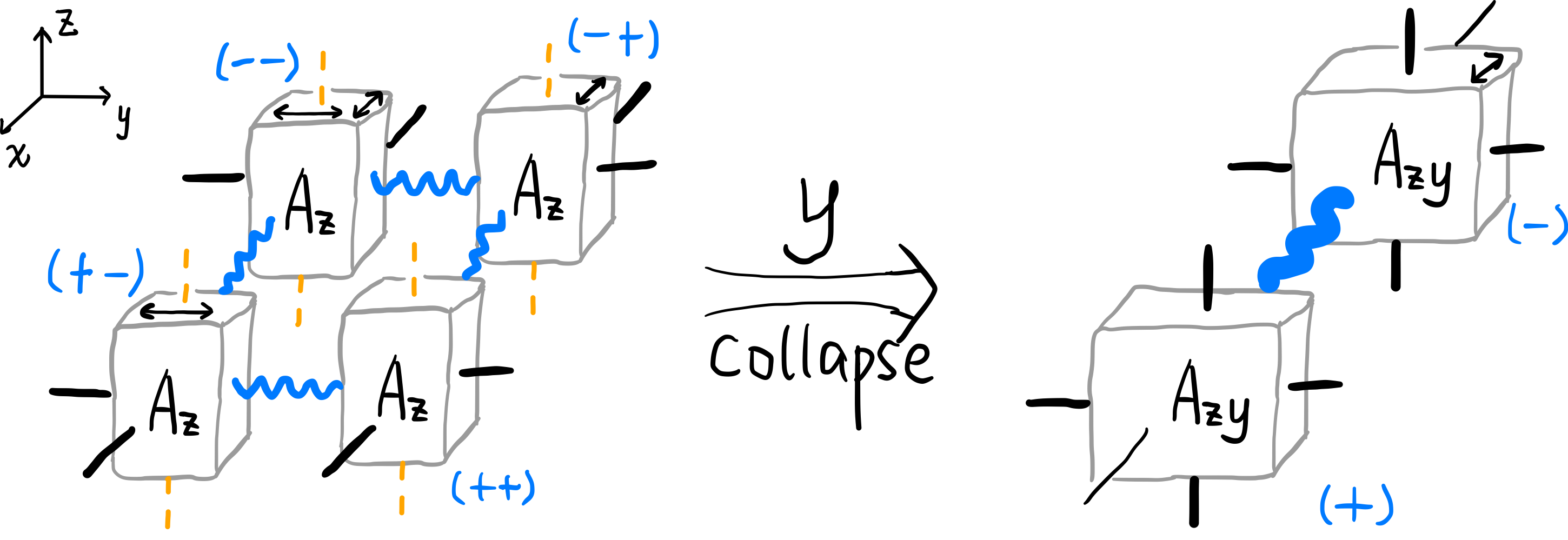}.
\end{align}
\end{subequations}

\underline{\emph{Step 4.3 is the $x$ collapse}}
\par
The last step is the $x$ collapse:
\begin{align}
   \label{eq:3dxcol}
   \includegraphics[width=0.80\columnwidth, valign=c]{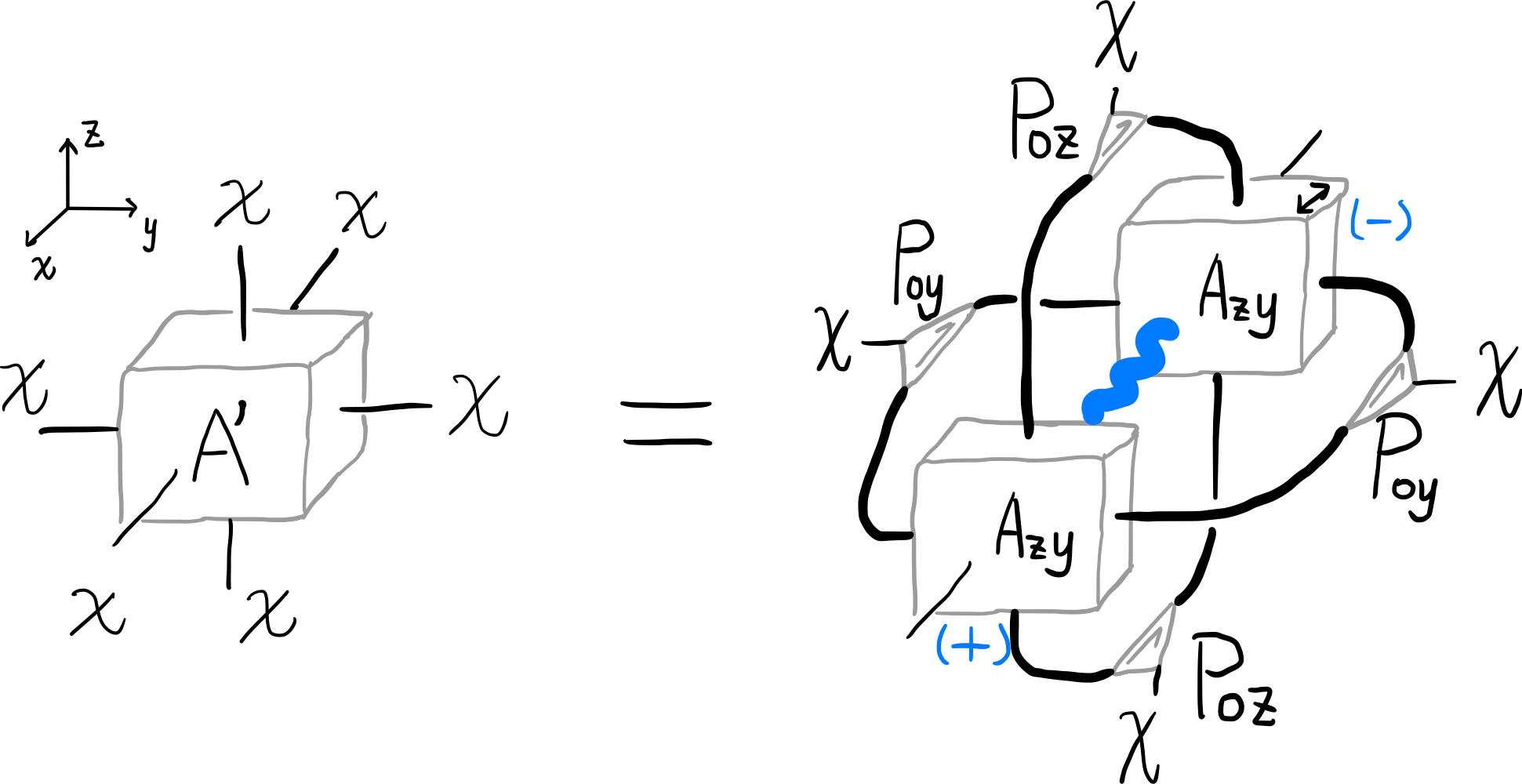}\quad,
\end{align}
which is a map $A_{zy} \mapsto A'$.
Two isometric tensors $p_{oy}$ and $p_{oz}$ are involved, and they can be determined in the same way as the \emph{Step 3} of the 3D algorithm by constructing proper density matrices using $A_{zy}$.

The composition of $A \mapsto A^f$ in Eq.~\eqref{eq:3dEFAf} and the three collapses $A^f \mapsto A_z \mapsto A_{zy} \mapsto A'$ in Eqs.~\eqref{eq:3dzcol},~\eqref{eq:3dycol}, and~\eqref{eq:3dxcol} gives the tensor RG equation $A \mapsto A'$. 
A big picture of the tensor RG equation is shown in Eq.~\eqref{eq:3dRGeq}, with all $A$ tensors replaced by $A^f$ in Eq.~\eqref{eq:3dEFAf} and with all isometric tensors, $p_{ix}, p_{iy}, p_{iix}$, for inner legs not drawn.

\begin{center}
    \underline{\emph{The lattice-reflection symmetry is preserved}}
    \par
\end{center}
By inspecting the tensor RG equation in Eq.~\eqref{eq:3dRGeq} (with $A$ replaced with $A^f$ in Eq.~\eqref{eq:3dEFAf}), along with the symmetry property of the isometric tensors $p_x, p_y$ and $p_z$ shown in Eqs.~\eqref{eq:isomSym} and~\eqref{eq:p2g}, it is easy to see that the coarse-grained tensor $A'$ satisfies the same lattice-reflection symmetry as the original tensor $A$ in Eq.~\eqref{eq:refl3d}.
The renormalized SWAP-gauge matrices are determined from $p_x, p_y$ and $p_z$, several examples of which have been shown in Eq.~\eqref{eq:p2g3D}.
Since the symmetry property of the $A'$ is true regardless of the symmetry of the original tensor $A$ in Eq.~\eqref{eq:3dRGeq}, the lattice-reflection symmetry is also imposed.

\section{Linearization of the RG map in separate lattice-reflection charge sectors\label{sec:linRG}}
According to the Wilsonian RG framework, the critical exponents characterizing a universality class can be extracted from the linearized RG map around the corresponding fixed point of the RG map.
This canonical RG prescription has already been established in the context of TNRG~\cite{Lyu:Xu:Kawashima:2021}, 
which was later reinterpreted as a lattice dilatation operator in Ref.~\cite{Ebel:2025}.
Compared with the more conventional approach of extracting scaling dimensions by building a transfer matrix~\cite{Gu:Wen:2009}, this linearized RG method has advantage in 3D.
It is technically difficult to build a transfer matrix on a sphere from a fixed-point tensor in 3D TNRG, while the linearized RG method works equally well in 2D and 3D~\cite{Lyu:Xu:Kawashima:2021,Lyu:Kawashima:2023}.
With the lattice-reflection exploited, it is expected that the linearized RG map should have various lattice-reflection sectors as its invariant subspace. 
Therefore, a natural question arises about how to linearize the RG map in separate lattice-reflection sectors.

We first explore this question in 1D to develop intuitions that can help one write down the answers in 2D and 3D directly.
Then, we develop a chain rule for the linearization. 
The chain rule facilitates the generalization of the intuition developed in the 1D toy example to higher dimensions.
With the 1D intuition and the chain rule, we will write down the linearized RG map in separate lattice-reflection sectors for the 2D and 3D algorithms in~\autoref{sec:algo}.
Justification of our claim of the linearization in 2D is given in Appendix~\ref{app:2dLinProof}, whose generalization to 3D is straightforward.

\subsection{The 1D toy example\label{subsec:linRG1D}}
We explore, using the 1D toy example in~\autoref{subsec:defSym}, the expectation that the lattice-reflection symmetry can be exploited to construct the linearized RG map in different lattice-reflection symmetry separately.

Without the transposition trick, the tensor RG equation is $A' = A A$ (see Eq.~\eqref{eq:bkten1d}).
Denote a fixed-point tensor of this RG map $A_*$, which is invariant under the RG transformation $A_* = A_* A_*$. 
The linearized RG around this fixed point is obtained by collecting all terms in $(A_* + \delta A) (A_* + \delta A)$ that are first-order in $\delta A$,
\begin{align}
    \label{eq:linearbkten1D}
    \mathcal{R}\bigr|_{A_*}  : \delta A \mapsto
    \delta A' \bigr|_{A_*} = \delta A A_* + A_* \delta A,
\end{align}
Suppose the transfer matrix $A$, as well as $\delta A$, is an $n$-by-$n$ real matrix and denote the vector space of the these matrices by $\MnnF$.
The linearized map above maps $\MnnF$ to itself,
\begin{align}
    \label{eq:linearmapM2M}
    \mathcal{R}\bigr|_{A_*}:
    \MnnF \to \MnnF.
\end{align}
Here, $\delta A \in \MnnF $ can be regarded as an $n^2$-dimensional vector, and the linear map $\mathcal{R}$ an $n^2$-by-$n^2$ matrix.

When the lattice-reflection symmetry is presented, the transfer matrix $A$ is symmetric, $A^\intercal = A$.
Denote the vector space of all $n$-by-$n$ real symmetric matrices as $\Mnn{0}$ and that of all $n$-by-$n$ real antisymmetric matrices $\Mnn{1}$ (if $M \in \Mnn{1}$, then $M^\intercal = - M$); 
the direct sum of the two is $\MnnF$,
\begin{align}
    \label{eq:MeqMpMm}
    \MnnF =
    \Mnn{0} \oplus \Mnn{1}.
\end{align}
We call $c=0,1$ in $\Mnn{c}$ \emph{the charge of the lattice-reflection symmetry}.

\begin{theorem}
    \label{theo:R1Dbreak}
    When $A_*$ is symmetric, the linearization $\mathcal{R}\bigr|_{A_*}$ in Eq.~\eqref{eq:linearbkten1D} has two invariant subspaces: $\Mnn{0}$ and $\Mnn{1}$.
\end{theorem}

\begin{proof}
The theorem can be proven by studying the symmetry property of $\delta A'$ in Eq.~\eqref{eq:linearbkten1D},
\begin{align}
    \label{eq:deltaApT}
    \left(\delta A'\right)^\intercal &=
    \left( \delta A A_* \right)^\intercal + 
    \left(A_* \delta A\right)^\intercal \nonumber\\
    &=
    A_* \left(\delta A\right)^\intercal
    +
    \left(\delta A\right)^\intercal A_*.
\end{align}
If $\delta A \in \Mnn{0}$, then $\left(\delta A'\right)^\intercal = \delta A'$, which means $\delta A' \in \Mnn{0}$.
If $\delta A \in \Mnn{1}$, then $\left(\delta A'\right)^\intercal = -\delta A'$, which means $\delta A' \in \Mnn{1}$.
\end{proof}

\begin{corollary}
    \label{corol:R1DeigvSym}
   When $A_*$ is symmetric, the eigenvectors of $\mathcal{R}\bigr|_{A_*}$ in Eq.~\eqref{eq:linearbkten1D} can be made to be either symmetric or antisymmetric.
\end{corollary}

\begin{proof}
    Suppose that $\mathcal{R}\bigr|_{A_*}$ has an eigenvector $\delta A_{\lambda}$ with eigenvalue $\lambda$,
    \begin{align}
        \label{eq:1dRprfEigv}
        \mathcal{R}\bigr|_{A_*}(\delta A_{\lambda}) = 
        \lambda \cdot \delta A_{\lambda},
    \end{align}
    where $\delta A_{\lambda}$ is neither symmetric nor antisymmetric.
    We can decompose $\delta A_{\lambda}$ into symmetric and antisymmetric parts as
    \begin{align}
        \label{eq:1dRprfDecomp1}
        \delta A_{\lambda} = 
        \delta A_{\lambda}^{(0)} + \delta A_{\lambda}^{(1)}
        \text{ where} \nonumber\\
        \delta A_{\lambda}^{(0)} \in \Mnn{0} \text{ and }
        \delta A_{\lambda}^{(1)} \in \Mnn{1}.
    \end{align}
    Then, the eigenvalue problem in Eq.~\eqref{eq:1dRprfEigv} becomes
    \begin{subequations}
        \begin{align}
            \label{eq:1dRprfEigvDecomp1}
            \mathcal{R}\bigr|_{A_*}(\delta A_{\lambda}^{(0)} + \delta A_{\lambda}^{(1)}) 
            = 
            \lambda (\delta A_{\lambda}^{(0)} + \delta A_{\lambda}^{(1)}).   
        \end{align}
        With some rearrangement of the terms, the above equation becomes
        \begin{align}
            \label{eq:1dRprfEigvDecomp2}
            \mathcal{R}\bigr|_{A_*}(\delta A_{\lambda}^{(0)})
            - \lambda \cdot \delta A_{\lambda}^{(0)}
            = 
            \lambda \cdot \delta A_{\lambda}^{(1)} - 
            \mathcal{R}\bigr|_{A_*}(\delta A_{\lambda}^{(1)}).
        \end{align}
    \end{subequations}
    Due to Theorem~\ref{theo:R1Dbreak}, the left-hand side (LHS) of the above equation is symmetric, while the right-hand side (RHS) is antisymmetric.
    Therefore, both sides of the above equation should vanish, leading to
    \begin{subequations}
    \begin{align}
        \label{eq:1dRprfEigvbreak}
        \mathcal{R}\bigr|_{A_*}(\delta A_{\lambda}^{(0)}) = 
        \lambda \cdot \delta A_{\lambda}^{(0)},\\
        \mathcal{R}\bigr|_{A_*}(\delta A_{\lambda}^{(1)}) = 
        \lambda \cdot \delta A_{\lambda}^{(1)}.
    \end{align}
    \end{subequations}
    This means that if $\mathcal{R}\bigr|_{A_*}$ has an eigenvector with no definite symmetry, we can take its symmetric and antisymmetric parts, and those two parts are eigenvectors of $\mathcal{R}\bigr|_{A_*}$ with the same eigenvalue.
\end{proof}

Since the linearization is usually evaluated at a fixed point, we will often drop $\bigr|_{A_*}$ and simply write $\mathcal{R}$ to simplify the notation.
We will be concerned with nonzero eigenvalues of a linearized RG map.
The following corollary follows immediately from Corollary~\ref{corol:R1DeigvSym}.
\begin{corollary}
Denote the collection of all nonzero eigenvalues of $\mathcal{R}$ in Eq.~\eqref{eq:linearbkten1D} as $\spec{\mathcal{R}}$.
This collection can be divided into two parts:
\begin{align}
    \label{eq:R1DspecDecomp}
    \spec{\mathcal{R}} = 
    \spec^{(0)}{\mathcal{R}} +
    \spec^{(1)}{\mathcal{R}},
\end{align}
where $\spec^{(0)}$ means the collection of eigenvalues whose eigenvectors are symmetric, while $\spec^{(1)}$ for antisymmetric ones.
\end{corollary}

Using the idea of the transposition trick, we discover a way to obtain the spectrum of $\mathcal{R}$ in Eq.~\eqref{eq:linearbkten1D} separately in its symmetric and antisymmetric invariant subspaces.
\begin{theorem}
    \label{theo:Rc1Dbuild}
    When $A_*$ is symmetric, one can build two linear map $\mathcal{R}^{(c)}, c=0,1$:
\begin{align}
    \label{eq:linearbkten1Dc}
    \mathcal{R}^{(c)}: \delta A \mapsto
    \delta A' \bigr|_{A_*}^c = 
    \delta A (A_*)^\intercal + A_* (\delta A)^\intercal \cdot (-1)^c,
\end{align}
whose spectra are related to that of the linearization $\mathcal{R}\bigr|_{A_*}$ in Eq.~\eqref{eq:linearbkten1D}.
All nonzero eigenvalues of $\mathcal{R}^{(c)}$ are those of $\mathcal{R}$ with symmetric eigenvectors if $c=0$, or with antisymmetric eigenvectors if $c=1$:
\begin{align}
    \label{eq:R1D2Rc1Dspec}
    \spec{\mathcal{R}^{(c)}} =
    \spec^{(c)} \mathcal{R}.
\end{align}
\end{theorem}

\begin{proof}
    The first step of the proof is to notice that in the subspace $\Mnn{c}$, $\mathcal{R}^{(c)}$ in Eq.~\eqref{eq:linearbkten1Dc} is the same linear map as $\mathcal{R}$ in Eq.~\eqref{eq:linearbkten1D},
\begin{align}
    \label{eq:R1D2RcinMnnc}
    \forall \delta A \in \Mnn{c},
    \mathcal{R}^{(c)}(\delta A)
    =
    \mathcal{R}(\delta A).
\end{align}
This indicates that $\spec^{(c)}{\mathcal{R}^{(c)}} = \spec^{(c)}{\mathcal{R}}$.
The second step is to notice that the range of the linear map $\mathcal{R}^{(c)}$ is $\Mnn{c}$,
\begin{subequations}
    \label{eq:linearmapM2Mp}
\begin{align}
    \mathcal{R}^{(c)}:
    \MnnF \to \Mnn{c}, \text{ or}\\
    \forall \delta A \in \MnnF,
    \mathcal{R}^{(c)}(\delta A) \in \Mnn{c}.
\end{align} 
\end{subequations}
The indication is that $\spec{\mathcal{R}^{(c)}} = \spec^{(c)}{\mathcal{R}^{(c)}}$.
To see the reason of this indication, suppose $\delta A_{\lambda}$ is an eigenvector of $\mathcal{R}^{(c)}$ with eigenvalue $\lambda \neq 0$, so $\mathcal{R}^{(c)}(\delta A_{\lambda}) = \lambda \cdot \delta A_{\lambda}$.
The LHS of this equation must be in the subspace $\Mnn{c}$ due to Eq.~\eqref{eq:linearmapM2Mp}, which shows that the eigenvector $\delta A_{\lambda}\in \Mnn{c}$.
Equations~\eqref{eq:R1D2RcinMnnc} and~\eqref{eq:linearmapM2Mp} together prove this theorem.
\end{proof}

\begin{remark}
    The domain of $\mathcal{R}^{(c)}$ in Eq.~\eqref{eq:linearbkten1Dc} is $\MnnF$, the same as that of $\mathcal{R}$ in Eq.~\eqref{eq:linearbkten1D}.
\end{remark}

\begin{remark}
    The linear map $\mathcal{R}^{(0)}$ can be understood as the linearization of the RG equation $A' = A A^\intercal$ in Eq.~\eqref{eq:bkten1d-sym} after the transposition trick in Eq.~\eqref{eq:tsptrick1D}.
\end{remark}

In Appendix~\ref{app:2dLinMatrix}, we will provide some intuition about the proofs in this subsection by treating the linear maps $\mathcal{R}$ and $\mathcal{R}^{(c)}$ as matrices.

\begin{center}
    \underline{\emph{Rules for the linearization}}
    \par
\end{center}
By inspecting the tensor RG equation with the transposition trick $A' = A A^\intercal$ in Eq.~\eqref{eq:bkten1d-sym} and the linearized RG map in Eq.~\eqref{eq:linearbkten1Dc}, we can summarize some rules for writing down the linearization once the RG map is given. 
To this end, we define a following bilinear map for two matrices with the same dimension, $M_1$ and $M_2$, to be a matrix multiplication of the two,
\begin{align}
    \label{eq:1dcontrDef}
    \mathcal{C}(M_1, M_2) \texteq{def} M_1 M_2.
\end{align}
The tensor RG equation $A' = A A^\intercal = \mathcal{T}(A)$ in Eq.~\eqref{eq:bkten1d-sym}, seen as a nonlinear map $\mathcal{T}: A \mapsto A'$, can be rewritten using the bilinear map $\mathcal{C}$ as
\begin{align}
    \label{eq:1dC2T}
    A' = \mathcal{T}(A) = A A^\intercal =
    \mathcal{C}(A, A^\intercal).
\end{align}
The linearized RG map $\mathcal{R}^{(c)}$ in Eq.~\eqref{eq:linearbkten1Dc} can be expressed using this bilinear map $\mathcal{C}$ as
\begin{align}
    \label{eq:1dC2R}
    \mathcal{R}^{(c)}:
    &\delta A' \bigr|_{A_*}^c 
    =
    \mathcal{L}\bigr|_{A_*}^c(\delta A) \nonumber\\
    &=
    \mathcal{C}(\delta A, A_*^\intercal).
    +
    \mathcal{C}(A_*, \left(\delta A\right)^\intercal) \cdot (-1)^c.
\end{align}
The notation $\mathcal{L}\bigr|_{A_*}^c$ means a linearization of $\mathcal{T}$ around the tensor $A_*$ in the lattice-reflection sector with charge $c$.
Although this notation looks redundant since we already have $\mathcal{R}^{(c)}$, it will be convenient in higher dimensions.
The rules can be summarized by comparing Eqs.~\eqref{eq:1dC2T} and~\eqref{eq:1dC2R} as follows,
\begin{itemize}
    \item Set all copies of tensor $A$ in the RG equation with the transposition trick to be the fixed-point tensor $A_*$, including all its transposition; specifically, in the 1D toy example, it means $\mathcal{C}(A_*, A_*^\intercal)$.
    \item Write down all possible terms by replacing each argument of the expression in the previous step with $\delta A$; if an argument is transposed in $\mathcal{C}$, $\delta A$ also gets transposition.
        In the 1D toy example, there are two such terms $\mathcal{C}(\delta A, A_*^\intercal)$ and $\mathcal{C}(A_*, \left(\delta A\right)^\intercal)$.
    \item If the $\delta A$ is transposed in a term, this term should be multiplied by a phase factor $(-1)^c$, according to the charge of the sector $c$;
        in the 1D toy example, this means the second term in the previous step becomes $\mathcal{C}(A_*, \left(\delta A\right)^\intercal) \cdot (-1)^c$.
    \item The sum of all these terms gives a linearized RG map in the lattice-reflection sector with charge $c$.
\end{itemize}

\subsection{Chain rule\label{subsec:chainrule}}
For a generalization to higher dimensions of the rules summarized from the 1D toy example in~\autoref{subsec:linRG1D}, we will develop a trick to break the total linearization into linearized RG maps in different directions of the HOTRG-like block-tensor map.
It is nothing but the chain rule of the total derivative of a composite function in calculus, applied in the context of TNRG\@.
For a certain direction, the rules in~\autoref{subsec:linRG1D} can be used to write down the linearized RG map.
For simplicity, we demonstrate this chain rule for the RG equation of \emph{the usual HOTRG in 2D in Eq.~\eqref{eq:hotrg2step}, without the transposition trick}.

The tensor RG equation of the 2D HOTRG in Eq.~\eqref{eq:hotrg2step} is a composition of two collapses in two directions. 
Write the tensor-network contraction of the $y$ collapses $\mathcal{T}_{\text{ntt}}^y$ (the first equation in Eq.~\eqref{eq:hotrg2step}) as a bilinear map $\mathcal{C}^y$,
\begin{align}
    \label{eq:2dC2TyNoSym}
    A_y = \mathcal{T}_{\text{ntt}}^y(A)
    =
    \mathcal{C}^y(A, A; \{p^{\text{ntt}}\}).
\end{align}
The subscript and superscript ``ntt'' stands for ``no transposition trick''.
Just like Eqs.~\eqref{eq:1dcontrDef} and~\eqref{eq:1dC2T}, the notation $T_{\text{ntt}}^y$ means a nonlinear map $A \to A_y$, while $C^y$ is bilinear in its two arguments.
The collection of isometric tensors  $\{p^{\text{ntt}}\}$ in $\mathcal{C}^y$, which represents just a single $p_x$ for the case in Eq.~\eqref{eq:hotrg2step}, parametrizes this bilinear map.
The first argument of $\mathcal{C}^y$ stands for the upper tensor $A$, while the second for the lower one in first equation in Eq.~\eqref{eq:hotrg2step}.
The collapse in $x$ direction $\mathcal{T}_{\text{ntt}}^x$ (the second equation in Eq.~\eqref{eq:hotrg2step}) can be written similarly as
\begin{align}
    \label{eq:2dC2TxNoSym}
    A' = \mathcal{T}_{\text{ntt}}^x(A_y)
    =
    \mathcal{C}^x(A_y, A_y; \{p^{\text{ntt}}\}),
\end{align}
where the first argument of $\mathcal{C}^x$ stands for the right tensor $A_y$, while the second for the left one in the second equation in Eq.~\eqref{eq:hotrg2step}.
The tensor RG equation in Eq.~\eqref{eq:hotrg2step} $\mathcal{T}_{\text{ntt}}: A \mapsto A'$ is a composition of the two collapses,
\begin{align}
    \label{eq:2dTxTyTNoSym}
    \mathcal{T}_{\text{ntt}} =
    \mathcal{T}_{\text{ntt}}^x \circ \mathcal{T}_{\text{ntt}}^y.
\end{align}
The linearization of $\mathcal{T}_{\text{ntt}}$ around $A$ can be obtained by the following chain rule for total derivatives,
\begin{align}
    \label{eq:2dLinRGNoSym}
    \mathcal{R}_{\text{2D}}^{\text{ntt}}\bigr|_{A}
    =
    \mathcal{L}^x_{\text{ntt}}{} \bigr|_{A_y} \circ
    \mathcal{L}^y_{\text{ntt}} \bigr|_{A}\quad.
\end{align}
$\mathcal{L}^y_{\text{ntt}} \bigr|_{A}$ is the linearization of $\mathcal{T}^y_{\text{ntt}}$ around $A$, which can be determined from the $\mathcal{C}^y$ in Eq.~\eqref{eq:2dC2TyNoSym},
\begin{align}
    \label{eq:2dLinYNoSym}
    \delta A_y 
    &= \mathcal{L}^y_{\text{ntt}} \bigr|_{A}(\delta A) \nonumber\\
    &=
    \mathcal{C}^y(\delta A, A; \{p^{\text{ntt}}\})
    +
    \mathcal{C}^y(A, \delta A; \{p^{\text{ntt}}\}),
\end{align}
and $\mathcal{L}^x_{\text{ntt}}\bigr|_{A_y}$ means the linearization of $\mathcal{T}^x_{\text{ntt}}$ around $A_y = \mathcal{T}^y_{\text{ntt}}(A)$, which can be determined from the $\mathcal{C}^x$ in Eq.~\eqref{eq:2dC2TxNoSym},
\begin{align}
    \label{eq:2dLinXNoSym}
    \delta A' 
    &= \mathcal{L}^x_{\text{ntt}} \bigr|_{A_y}(\delta A_y) \nonumber\\
    &=
    \mathcal{C}^x(\delta A_y, A_y; \{p^{\text{ntt}}\})
    +
    \mathcal{C}^x(A_y, \delta A_y; \{p^{\text{ntt}}\}).
\end{align}
The composition of the two Eqs.~\eqref{eq:2dLinYNoSym} and~\eqref{eq:2dLinXNoSym} is the $\mathcal{R}_{\text{2D}}^{\text{ntt}}\bigr|_{A}: \delta A \mapsto \delta A'$. 

An important remark in the linearization scheme here is that \emph{we treat all isometric tensors for the coarse graining as constants during the linearization}. 
This treatment works well for the Ising model in 2D and 3D~\cite{Lyu:Xu:Kawashima:2021, Lyu:Kawashima:2023, Lyu:Kawashima:2024,Guo:Wei:2024}. 
The same treatment is used in the linearization of Entanglement Renormalization~\cite{Vidal:2010-book}, which also indicates its validity considering the high accuracy of the scaling dimensions estimations. 
In the context of TNRG, this numerical practice has been explored in more details in Ref.~\cite{Ebel:2025}, where it is shown that the linearized RG map obtained by freezing the isometric tensors is equivalent to the lattice dilatation operator.

Without exploiting the lattice-reflection symmetry using the transposition trick, the EF process leads to an RG equation in Eq.~\eqref{eq:efrg2dmap}.
The eight filtering matrices, acting on the $4$-leg tensors in the block-tensor patch, make the tensors at different positions in the $2 \times 2$ block different from each other.
While the chain rule still applies, it complicates a lot.
However, as will be demonstrated in the following subsections, when the lattice-reflection symmetry is exploited using the transposition trick, the chain rule for the linearization remains simple even after the EF is incorporated.
This is another advantage of using the transposition trick to exploit the lattice-reflection symmetry in TNRG\@.

\subsection{The linearization in 2D\label{subsec:2dLinRG}}
With the rules summarized from the intuition of the 1D toy example in~\autoref{subsec:linRG1D} and the chain rule of the linearization in 2D in~\autoref{subsec:chainrule}, we can write down the linearization in separate lattice-reflection sectors of the 2D and 3D EF-enhanced TNRG maps in~\autoref{sec:algo}.
In this subsection, we write down the answer in 2D.
Since the chain rule in 2D has only been developed without exploiting the lattice-reflection symmetry and without incorporating the EF, the claimed linearized RG maps are educated guesses at this point.
The proof of our claim will be given in Appendix~\ref{app:2dLinProof}.

First, we generalize the notation of symmetric and antisymmetric matrices to tensors. 
We focus on 4-leg tensors in 2D, and denote the linear space where they live as $\mathfrak{T}_4$. 
As a subset of $\mathfrak{T}_4$, a lattice-reflection sector $\mathfrak{T}_4^{(c_x,c_y)}$, with charge set $(c_x, c_y)$ and SWAP-gauge matrix set $(g_x, g_y)$, is defined as
\begin{align}
    \label{eq:T4cxcyDef}
    \mathfrak{T}_4^{(c_x, c_y)}
    \texteq{def}
    \left\{ \delta A \in \mathfrak{T}_4 \left|
    \includegraphics[width=0.55\columnwidth, valign=c]{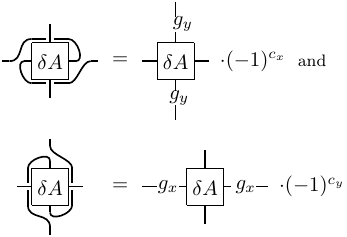}\right.
    \right\},
\end{align}
where $g_x$ and $g_y$ are SWAP-gauge matrices that square to identity $g_x g_x = g_y g_y = 1$ and $c_x, c_y \in \{0, 1\}$. 
This definition reduces to the symmetric and antisymmetric matrices when the bond dimension of two legs of a certain direction is set to be $1$. 
It is easy to see that $\mathfrak{T}_4^{(c_x, c_y)}$ is also a subspace of $\mathfrak{T}_4$. 
According to the definition of the lattice-reflection symmetry in Eq.~\eqref{eq:refl2d}, the tensor $A$ has lattice-reflection charges $c_x=c_y=0$.

The tensor RG map of the 2D algorithm in~\autoref{subsec:algo2d} is a composition of the following three transformations.
The first transformation is acting the filtering matrices $s_x$ and $s_y$ on the 4-leg tensor $A$, as is shown in Eq.~\eqref{eq:2dEFRGeq2}. 
We rewrite this EF transformation as
\begin{subequations}
\begin{align}
    \label{eq:2dA2Af}
    \mathcal{F}^{\square}: A \mapsto A^f =
    \mathcal{F}^{\square}\left( A; \{s\} \right),
\end{align}
where $\{s\}$ denotes  $s_x, s_y$ collectively and the subscript $\square$ means the target patch of the EF is a $2 \times 2$ square tensor network.
The remaining two transformations are two collapses of the HOTRG-like block tensor transformation in Eq.~\eqref{eq:2dEFhotRGeq}. 
The second transformation is the $y$ collapse $\mathcal{T}^y_{\text{2D}}: A^f \mapsto A^f_y$ in Eq.~\eqref{eq:2dEFhotRGycol}, which can be rewritten using a bilinear map $\mathcal{C}^y_{\text{2D}}$ as
\begin{align}
    \label{eq:2dAf2Afy}
    \mathcal{T}^y_{\text{2D}}: A^f \mapsto 
    A^f_y 
    &=
    \mathcal{T}^y_{\text{2D}}(A^f)  \nonumber\\
    &=
    \mathcal{C}^y_{\text{2D}}\left(
        A^f, \left(A^f\right)^{T y}; \{p\}
    \right).
\end{align}
The notation $\left(A^f\right)^{T y}$ means transposing the two $y$ legs of the tensor $A^f$, as is shown in Eq.~\eqref{eq:2dEFhotRGycol}.
The third transformation is $\mathcal{T}^x_{\text{2D}}: A^f_y \mapsto A'$ in Eq.~\eqref{eq:2dEFhotRGxcol}, which can be rewritten using a bilinear map $\mathcal{C}^x_{\text{2D}}$ as
\begin{align}
    \label{eq:2dAfy2Ap}
    \mathcal{T}^x_{\text{2D}}: A^f_y \mapsto A' =
    \mathcal{T}^x_{\text{2D}}(A^f_y) =
    \mathcal{C}^x_{\text{2D}}\left(
        A^f_y, \left(A^f_y\right)^{T x}; \{p\}
    \right).
\end{align}
Similarly, the notation $\left(A^f_y\right)^{T x}$ means transposing the two $x$ legs of the tensor $A^f_y$, as is shown in Eq.~\eqref{eq:2dEFhotRGxcol}.
\end{subequations}
The tensor RG map $\mathcal{T}_{\text{2D}}$ is a composition of the above three transformations,
\begin{align}
    \label{eq:T2Ddecomp}
    \mathcal{T}_{\text{2D}} =
    \mathcal{T}^x_{\text{2D}} \circ
    \mathcal{T}^y_{\text{2D}} \circ
    \mathcal{F}^\square.
\end{align}

We claim that the linearization $\mathcal{R}_{\text{2D}}^{(c_x, c_y)}$ related to the above RG map $\mathcal{T}_{\text{2D}}$ in the lattice-reflection sector $\mathfrak{T}_4^{(c_x, c_y)}$ can be written down using the chain rule in~\autoref{subsec:chainrule} and the rules summarized from the 1D toy example in~\autoref{subsec:linRG1D}.
The linearized RG map $\mathcal{R}_{\text{2D}}^{(c_x, c_y)} \bigr|_A$ around a tensor $A$ is a composition of three linearized transformations,
\begin{align}
    \label{eq:R2Ddecomp}
    \mathcal{R}_{\text{2D}}^{(c_x, c_y)} \bigr|_A =
    \mathcal{L}^x_{\text{2D}} \bigr|^{c_x}_{A^f_y} \circ
    \mathcal{L}^y_{\text{2D}} \bigr|^{c_y}_{A^f} \circ
    \mathcal{L}^\square_{\mathcal{F}} \bigr|_A.
\end{align}
In the above equation, $\mathcal{L}^\square_{\mathcal{F}} \bigr|_A$ is the linearization of $\mathcal{F}^\square$ in Eq.~\eqref{eq:2dA2Af} around tensor $A$.
Since $\mathcal{F}^\square : A \mapsto A^f$ is linear, its linearization is easy to write down,
\begin{subequations}
\begin{align}
    \label{eq:2dlinF}
    \mathcal{L}^\square_{\mathcal{F}} \bigr|_A:
    \delta A \mapsto 
    &\delta A^f =
    \mathcal{F}^\square \left( \delta A; \{s\} \right)
    \text{ or pictorially as} \nonumber\\
    &\includegraphics[scale=1.0, valign=c]{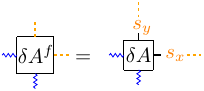}\quad.
\end{align}
The second map $\mathcal{L}^y_{\text{2D}} \bigr |_{A^f}^{c_y}$ is related to the $y$ collapse $\mathcal{T}^y_{\text{2D}}$ in Eq.~\eqref{eq:2dAf2Afy}, which can be determined using the rules in~\autoref{subsec:linRG1D},
\begin{align}
    \label{eq:2dlinLy}
    &\mathcal{L}^y_{\text{2D}} \bigr|_{A^f}^{c_y}:
    \delta A^f \mapsto 
    \delta A^f_y =
    \mathcal{C}^y_{\text{2D}}\left(
        \delta A^f, \left(A^f\right)^{T y}; \{p\}
    \right) \nonumber\\
    &+
    \mathcal{C}^y_{\text{2D}}\left(
        A^f, \left(\delta A^f\right)^{T y}; \{p\} 
    \right) \cdot (-1)^{c_y}
    \text{ or pictorially as} \nonumber\\
    &\includegraphics[width=0.80\columnwidth, valign=c]{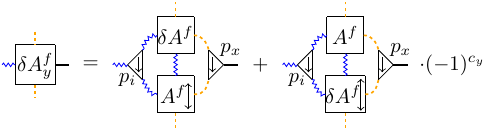}.
\end{align}
The third map $\mathcal{L}^x_{\text{2D}} \bigr |_{A^f_y}^{c_x}$ is related to the $x$ collapse $\mathcal{T}^x_{\text{2D}}$ in Eq.~\eqref{eq:2dAfy2Ap}, which can also be determined using the rules in~\autoref{subsec:linRG1D},
\begin{align}
    \label{eq:2dlinLx}
    &\mathcal{L}^x_{\text{2D}} \bigr|_{A^f_y}^{c_x}:
    \delta A^f_y \mapsto 
    \delta A' =
    \mathcal{C}^x_{\text{2D}}\left(
        \delta A^f_y, \left(A^f_y\right)^{T x}; \{p\}
    \right) \nonumber\\
    &+
    \mathcal{C}^x_{\text{2D}}\left(
        A^f_y, \left(\delta A^f_y \right)^{T x}; \{p\} 
    \right) \cdot (-1)^{c_x}
    \text{ or pictorially as} \nonumber\\
    &\includegraphics[width=0.80\columnwidth, valign=c]{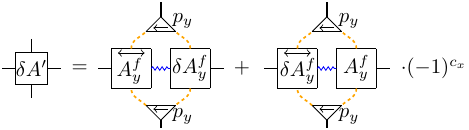}.
\end{align}
\end{subequations}
The proof of this claim from Eq.~\eqref{eq:R2Ddecomp} to Eq.~\eqref{eq:2dlinLx} will be given in Appendix~\ref{app:2dLinProof}. 
% Notice that only when $c_x=c_y=0$, the linear map $\mathcal{R}_{\text{2D}}^{(c_x, c_y)}$ in Eq.~\eqref{eq:R2Ddecomp} is the linearization of the RG map $\mathcal{T}_{\text{2D}}$ in Eq.~\eqref{eq:T2Ddecomp}.
We will also justify this claim numerically using the 2D Ising model in~\autoref{sec:numdemo}.

\subsection{The linearization in 3D\label{subsec:3dLinRG}}
It is straightforward to generalize the linearization of the 2D EF-enhanced RG in previous subsection to 3D.
Denote the linear space where all 6-leg tensors live as $\mathfrak{T}_6$. 
A lattice-reflection sector $\mathfrak{T}_6^{(c_x,c_y,c_z)}$, with charge set $(c_x,c_y, c_z)$ and SWAP-gauge matrix set $(g_{zx}, g_{zy}, g_{yz}, g_{yx}, g_{xy}, g_{xz})$, is a subspace of $\mathfrak{T}_6$ that contains all 6-leg tensors that satisfies a similar condition to the 2D one in Eq.~\eqref{eq:T4cxcyDef}. 
In parallel with the definition of a lattice-reflection symmetric tensor in Eq.~\eqref{eq:refl3d}, take the reflection across the $z$-plane as an example; for any $\delta A \in \mathfrak{T}_6^{(c_x, c_y, c_z)}$, it should satisfy the following condition when its two $z$ legs are transposed,
\begin{align}
    \label{eq:T6cxcyczDefz}
    \includegraphics[width=0.75\columnwidth, valign=c]{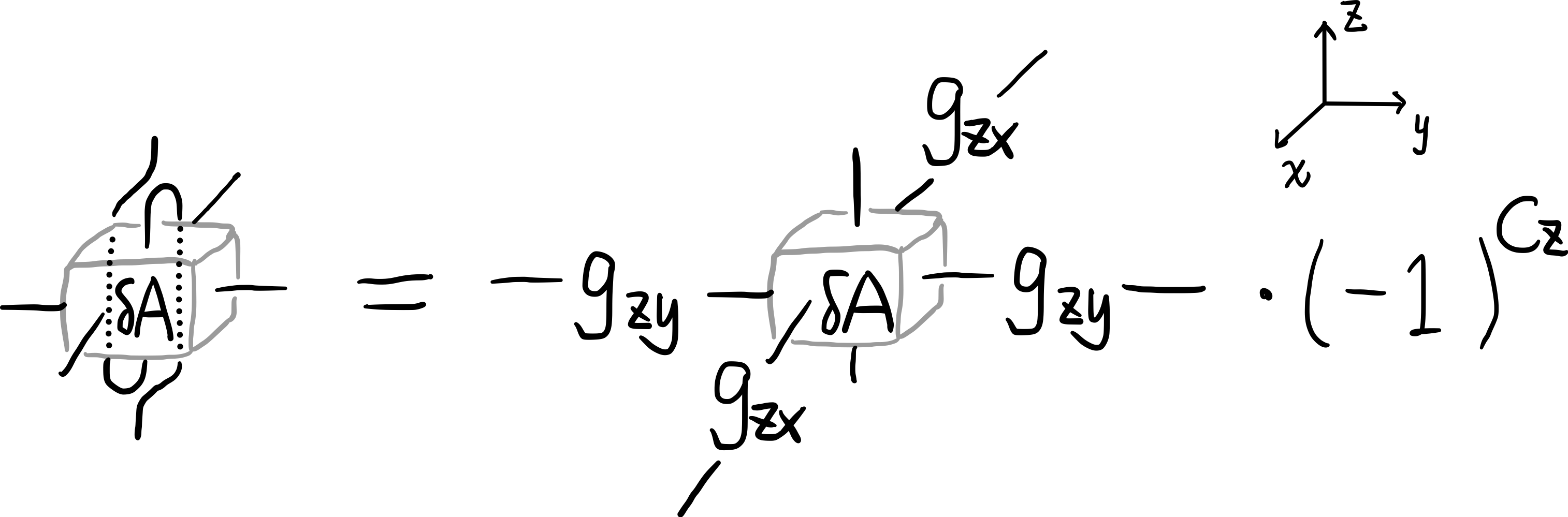}\quad.
\end{align}
Similar conditions hold for $\delta A$ when it reflects across the $y$- and $x$-plane. 
The tensor $A$ in Eq.~\eqref{eq:refl3d} has reflection charges $c_x=c_y=c_z=0$.

The 3D EF-enhanced RG map is a composition of four transformations--- $A \mapsto A^f$ in Eq.~\eqref{eq:3dEFAf} and $A^f \mapsto A_z \mapsto A_{zy} \mapsto A'$ in Eqs.~\eqref{eq:3dzcol},~\eqref{eq:3dycol} and~\eqref{eq:3dxcol},
\begin{align}
    \label{eq:T3Ddecomp}
    \mathcal{T}_{\text{3D}} =
    \mathcal{T}^x_{\text{3D}} \circ
    \mathcal{T}^y_{\text{3D}} \circ
    \mathcal{T}^z_{\text{3D}} \circ
    \mathcal{F}^{\cube}.
\end{align}
Specifically, $\mathcal{F}^{\cube}$ is the EF map in Eq.~\eqref{eq:3dEFAf}, 
\begin{subequations}
\begin{align}
    \label{eq:3DA2Af}
    \mathcal{F}^{\cube}: A \mapsto A^f = \mathcal{F}^{\cube}(A; \{s\})
\end{align}
with the subscript $\cube[.8]$ indicating that the target patch of the EF is a $2 \times 2 \times 2$ cube tensor network.
The other three maps in Eq.~\eqref{eq:T3Ddecomp} are HOTRG-like collapses in $z,y,x$ directions,
\begin{align}
    \label{eq:3DAf2Az}
    \mathcal{T}^{z}_{\text{3D}}: 
    A^f \mapsto 
    A_z 
    &= 
    \mathcal{T}^z_{\text{3D}}(A^f) \quad\text{(see Eq.~\eqref{eq:3dzcol})} \nonumber\\
    &= 
    \mathcal{C}^{z}_{\text{3D}}(A^f, \left(A^f\right)^{T z}; \{p\}),
\end{align}
\begin{align}
    \label{eq:3DAz2Azy}
    \mathcal{T}^{y}_{\text{3D}}: 
    A_z \mapsto 
    A_{zy} 
    &= 
    \mathcal{T}^y_{\text{3D}}(A_z) \quad\text{(see Eq.~\eqref{eq:3dycol})}\nonumber\\
    &=
    \mathcal{C}^{y}_{\text{3D}}(A_z, \left(A_z\right)^{T y}; \{p\}),
\end{align}
\begin{align}
    \label{eq:3DAzy2Ap}
    \mathcal{T}^{x}_{\text{3D}}: 
    A_{zy} \mapsto 
    A' 
    &= 
    \mathcal{T}^x_{\text{3D}}(A_{zy}) \quad\text{(see Eq.~\eqref{eq:3dxcol})}\nonumber\\
    &=
    \mathcal{C}^{x}_{\text{3D}}(A_{zy}, \left(A_{zy}\right)^{T x}; \{p\}).
\end{align}
\end{subequations}

We claim that the linearization $\mathcal{R}_{\text{3D}}^{(c_x, c_y, c_z)} \bigr|_A$ related to the 3D RG map $\mathcal{T}_{\text{3D}}$ around a tensor $A$ is a composition of the four transformations on the right-hand side of Eq.~\eqref{eq:T3Ddecomp},
\begin{align}
    \label{eq:R3Ddecomp}
    \mathcal{R}_{\text{3D}}^{(c_x, c_y, c_z)} \bigr |_A
    =
    \mathcal{L}^x_{\text{3D}} \bigr|_{A_{zy}}^{c_x}\circ
    \mathcal{L}^y_{\text{3D}} \bigr|_{A_z}^{c_y}\circ
    \mathcal{L}^z_{\text{3D}} \bigr|_{A^f}^{c_z}\circ
    \mathcal{L}^{\cube}_{\mathcal{F}} \bigr|_A.
\end{align}
\begin{subequations}
In the above equation, the linearization $\mathcal{L}^{\cube}_{\mathcal{F}} \bigr|_A$ of the EF map $\mathcal{F}^{\cube}$ in Eq.~\eqref{eq:3DA2Af} is easy to obtain since $\mathcal{F}^{\cube}$ itself is linear,
\begin{align}
    \label{eq:3dlinF}
    \mathcal{L}^{\cube}_{\mathcal{F}} \bigr|_A:
    \delta A \mapsto 
    &\delta A^f = 
    \mathcal{F}^{\cube}(\delta A; \{s\})
    \text{ or pictorially as } \nonumber\\
    &\includegraphics[width=0.35\columnwidth, valign=c]{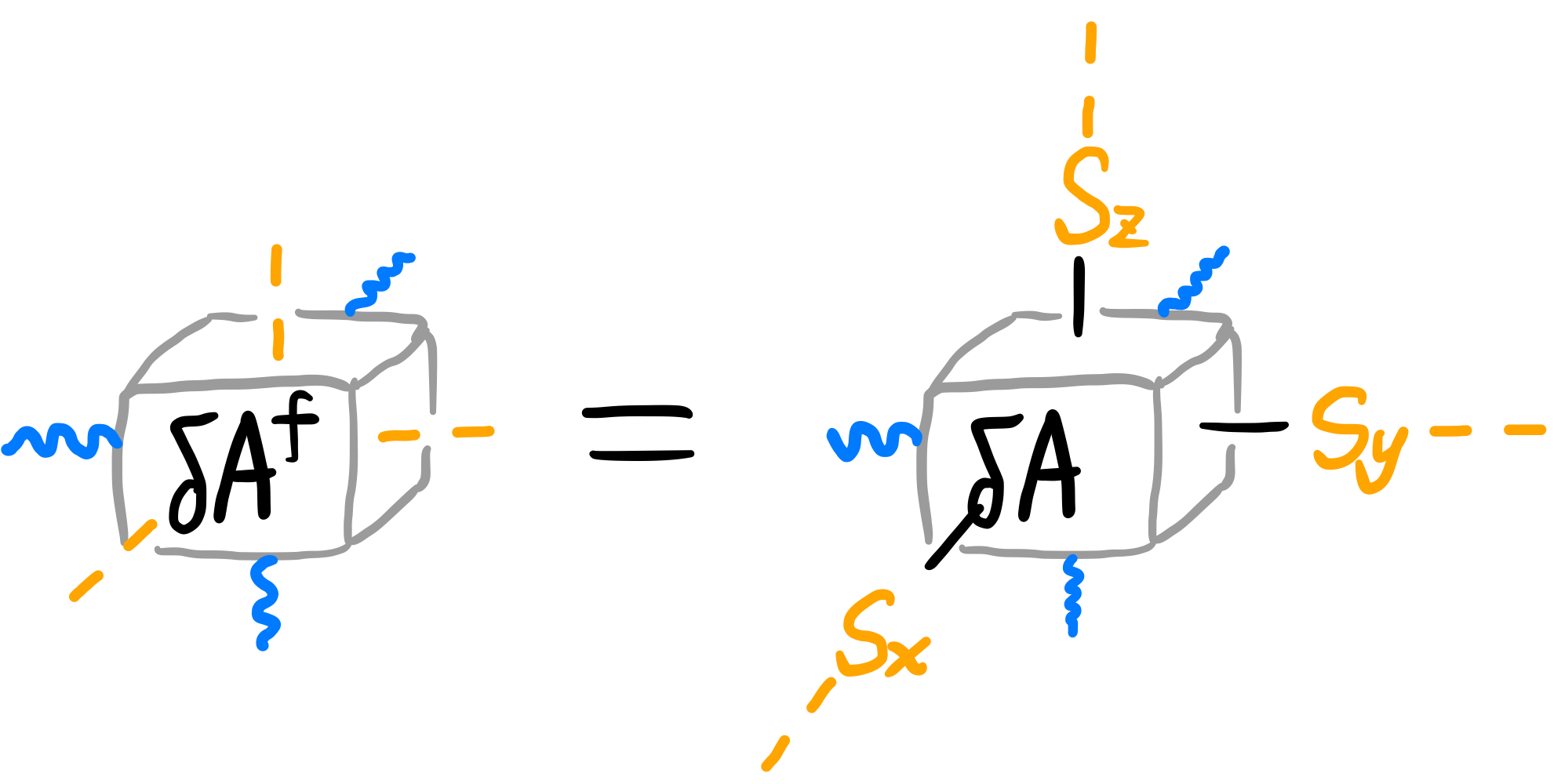}\quad.
\end{align}
The linearized maps related to the HOTRG-like collapses can be written down according to the rules summarized from the 1D toy example in~\autoref{subsec:linRG1D}.
The map $\mathcal{L}^z_{\text{3D}} \bigr|_{A^f}^{c_z}$ related to the $z$ collapse in Eq.~\eqref{eq:3DAf2Az} is
\begin{align}
    \label{eq:3dlinLz}
    &\mathcal{L}^z_{\text{3D}} \bigr|_{A^f}^{c_z}:
    \delta A^f \mapsto 
    \delta A_z = 
    \mathcal{C}^{z}_{\text{3D}}(\delta A^f, \left(A^f\right)^{T z}; \{p\}) \nonumber\\
    &+
    \mathcal{C}^{z}_{\text{3D}}(A^f, \left(\delta A^f\right)^{T z}; \{p\}) \cdot (-1)^{c_z}
    \text{ or pictorially as } \nonumber\\
    &\includegraphics[width=0.83\columnwidth, valign=c]{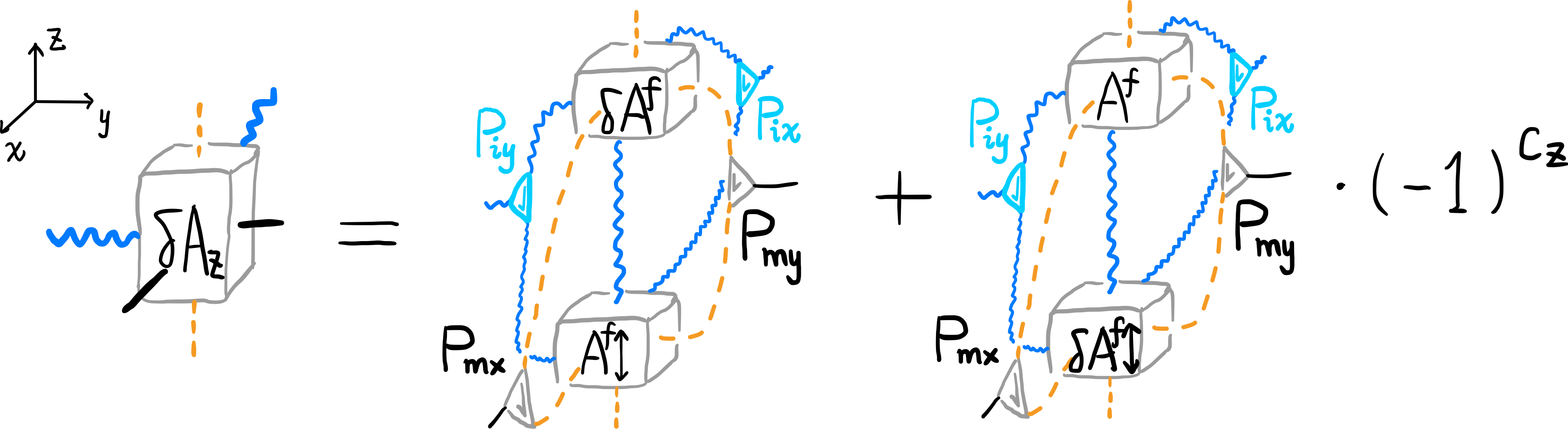}.
\end{align}
The map $\mathcal{L}^y_{\text{3D}} \bigr|_{A_z}^{c_y}$ related to the $y$ collapse in Eq.~\eqref{eq:3DAz2Azy} is
\begin{align}
    \label{eq:3dlinLy}
    &\mathcal{L}^y_{\text{3D}} \bigr|_{A_z}^{c_y}:
    \delta A_z \mapsto 
    \delta A_{zy} = 
    \mathcal{C}^{y}_{\text{3D}}(\delta A_z, \left(A_z\right)^{T y}; \{p\}) \nonumber\\
    &+
    \mathcal{C}^{y}_{\text{3D}}(A_z, \left(\delta A_z\right)^{T y}; \{p\}) \cdot (-1)^{c_y}
    \text{ or pictorially as } \nonumber\\
    &\includegraphics[width=0.83\columnwidth, valign=c]{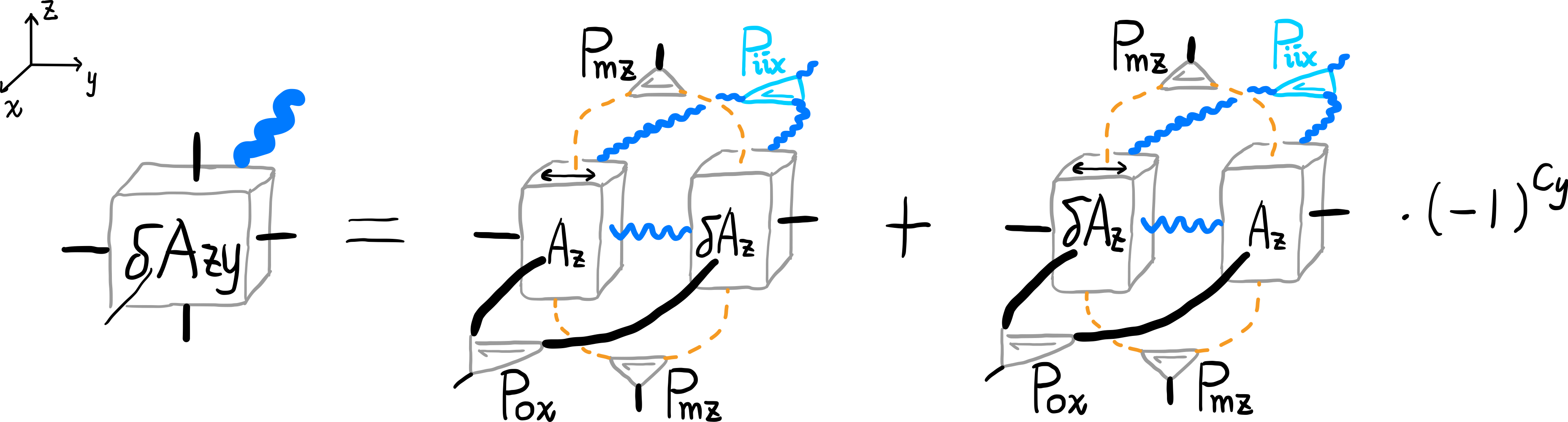}.
\end{align}
The map $\mathcal{L}^x_{\text{3D}} \bigr|_{A_{zy}}^{c_x}$ related to the $x$ collapse in Eq.~\eqref{eq:3DAzy2Ap} is
\begin{align}
    \label{eq:3dlinLx}
    &\mathcal{L}^x_{\text{3D}} \bigr|_{A_{zy}}^{c_x}:
    \delta A_{zy} \mapsto 
    \delta A' = 
    \mathcal{C}^{x}_{\text{3D}}(\delta A_{zy}, \left(A_{zy}\right)^{T x}; \{p\}) \nonumber\\
    &+
    \mathcal{C}^{x}_{\text{3D}}(A_{zy}, \left(\delta A_{zy}\right)^{T x}; \{p\}) \cdot (-1)^{c_x}
    \text{ or pictorially as } \nonumber\\
    &\includegraphics[width=0.83\columnwidth, valign=c]{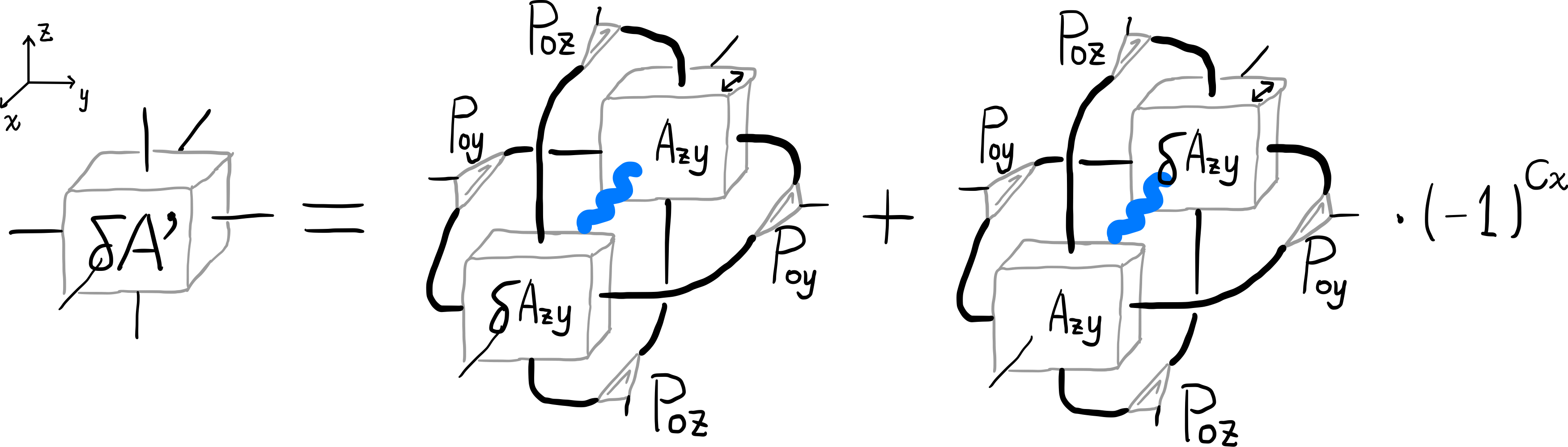}.
\end{align}
\end{subequations}

We omit the proof for the claim in 3D since it follows exactly the same strategy as the proof of the claim in 2D, which will be laid out in Appendix~\ref{app:2dLinProof}. 
The numerical justification of this claim was demonstrated in our previous work~\cite{Lyu:Kawashima:2024} by estimating the scaling dimensions of the 3D Ising model in different lattice-reflection symmetry separately. 
We will provide more numerical evidence in~\autoref{subsec:ising3d}.

\section{Numerical demonstration\label{sec:numdemo}}
In this section, we use the Ising model as a numerical demonstration of the techniques we have developed in the previous sections. 
The emphasis of our numerical demonstration will be extracting scaling dimensions from linearized RG map in separate lattice-reflection sectors, which has not been reported before in TNRG. 
The numerical results in this section can be reproduced using the {\tt python} codes published at Ref.~\cite{Lyu:algo:refl}.

According to Wilsonian RG theory, the scaling dimensions $\{x_i\}$ of a system at criticality can be extracted from the eigenvalues $\{\lambda_i\}$ of the linearized RG map around the fixed point that corresponds to that criticality,
\begin{align}
    \label{eq:rgeig2xGen}
    b^{d - x_i} = \lambda_i,
\end{align}
where $d$ is the spatial dimensionality of the system and $b$ is the rescaling factor of the RG map. 
For the RG maps in 2D and 3D proposed in~\autoref{sec:algo}, the rescaling factor is $b=2$.

In order to use Eq.~\eqref{eq:rgeig2xGen} in the TNRG, three points should be taken care of.
The first point is that the eigenvalues in the linearized RG map that correspond to total derivatives of the field are not universal~\cite{Ebel:2025,Wegner:1976}.
Therefore, in general, a linearized RG map is unable to predict the scaling dimensions of descendant operators in a conformal field theory (CFT).
However, a special way to linearize the RG map by ``freezing'' the isometric tensors and filtering matrices~\cite{Lyu:Xu:Kawashima:2021} is shown to be equivalent to the lattice dilatation operator, and thus is able to predict scaling dimensions of descendant operators~\cite{Ebel:2025}.
The linearized RG map derived in~\autoref{sec:linRG} employs such ``freezing'' method;
hence, its eigenvalue spectrum should be able to predict descendant operators.

The second point is fixing the gauge degrees of freedom. 
A systematical way of gauging fixing has been recently developed in Ref.~\cite{Ebel:2025:Rotation}. 
Here, we use a more preliminary but simple gauge fixing procedure in Ref.~\cite{Lyu:Xu:Kawashima:2021}, which works well in the context of the proposed algorithms. 

The third point is the overall multiplication of the fixed point tensors should be adjusted, since a fixed point tensor $A_*$ generally transforms like $A_* \mapsto c_* A_*$ under the RG map, where $c_*$ is a number.
This overall multiplication constant can be taken care of by rescaling the eigenvalue $\lambda_{\mathds{1}}$ that corresponds to the identity operator $\mathds{1}$ of the theory to be 1. 
For the Ising model, the identity operator has the lowest scaling dimension in spin-flip $\mathbb{Z}_2$ even sector with lattice-reflection charge to be $0$ in all direction. 
Therefore, one can use the following equation to determine the scaling dimensions from the eigenvalues of the linearized RG map,
\begin{align}
    \label{eq:rgeig2x}
    x_i = - \frac{\log\left(\lambda_i / \lambda_{\mathds{1}}\right)}{\log 2},
\end{align}
where $\lambda_{\mathds{1}}$ is the largest eigenvalue of $\mathcal{R}_{\text{2D}}^{(0, 0)}$ in Eq.~\eqref{eq:R2Ddecomp} and $\mathcal{R}_{\text{3D}}^{(0, 0, 0)}$ in Eq.~\eqref{eq:R3Ddecomp} for 2D and 3D respectively, both in the spin-flip $\mathbb{Z}_2$ even sector.

Incorporating the on-site spin-flip $\mathbb{Z}_2$ symmetry is straightforward in the linearized RG map with the technical framework in Ref.~\cite{Singh:Pfeifer:2011} and the \texttt{python} library published in Ref.~\cite{Markus:abelian}. 
Suppose we have a linearized RG map $\mathcal{R}: \delta A \mapsto \delta A'$.
To restrict $\mathcal{R}$ into the spin-flip $\mathbb{Z}_2$ even sector, one restricts the input $\delta A$ into that sector\footnote{
    In the library developed in Ref.~\cite{Markus:abelian}, for an instance \texttt{T} of the type \texttt{TensorZ2}, setting \texttt{T.charge=0} restricts the tensor \texttt{T} into the spin-flip $\mathbb{Z}_2$ sector.
}.

\subsection{The 2D Ising model\label{subsec:num2d}}
\begin{figure*}[tb]
    \subfloat[\label{fig:2Dspinflip0}
    Spin-flip even sector.
    The symbol $\mathds{1}$ denotes the identity operator, $\epsilon$ the energy density operator and $T_{mn}$ the energy-momentum operator.
    ]{
        \includegraphics[width=1.70\columnwidth,
        valign=c]{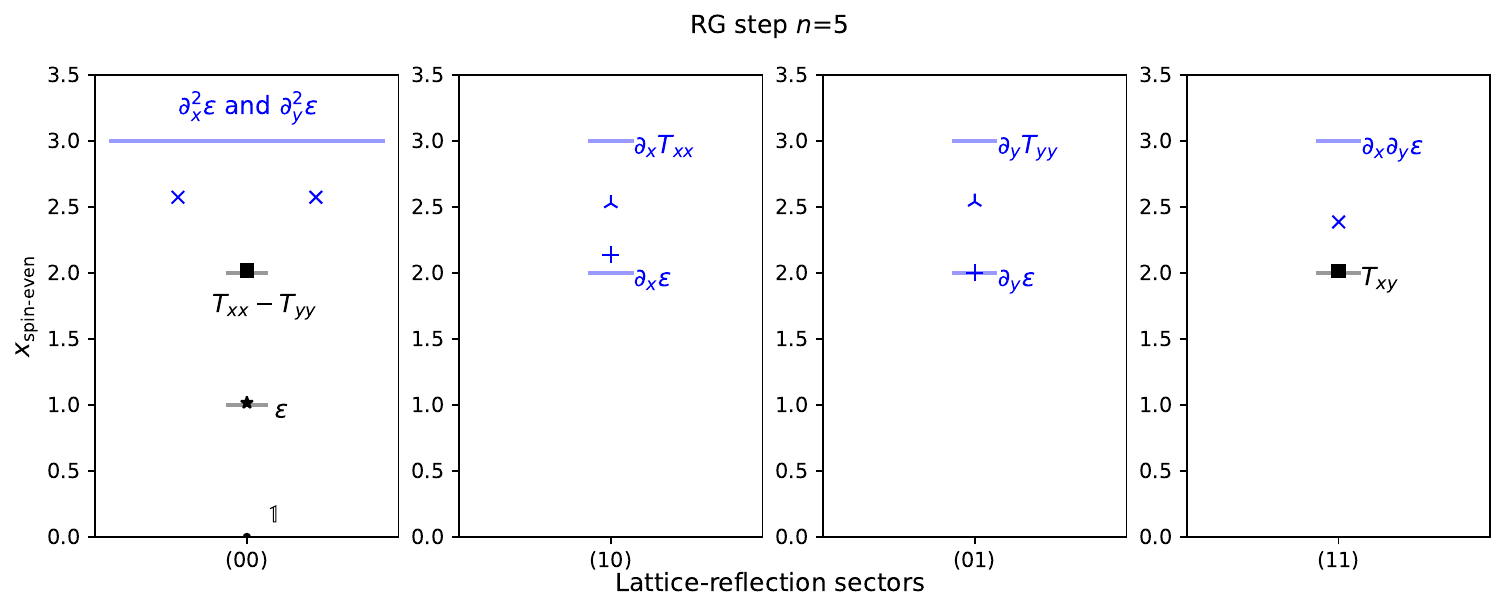}
    }\\
    \subfloat[\label{fig:2Dspinflip1}
    Spin-flip odd sector.
    The symbol $\sigma$ denotes the spin operator.
    ]{
        \includegraphics[width=1.70\columnwidth,
        valign=c]{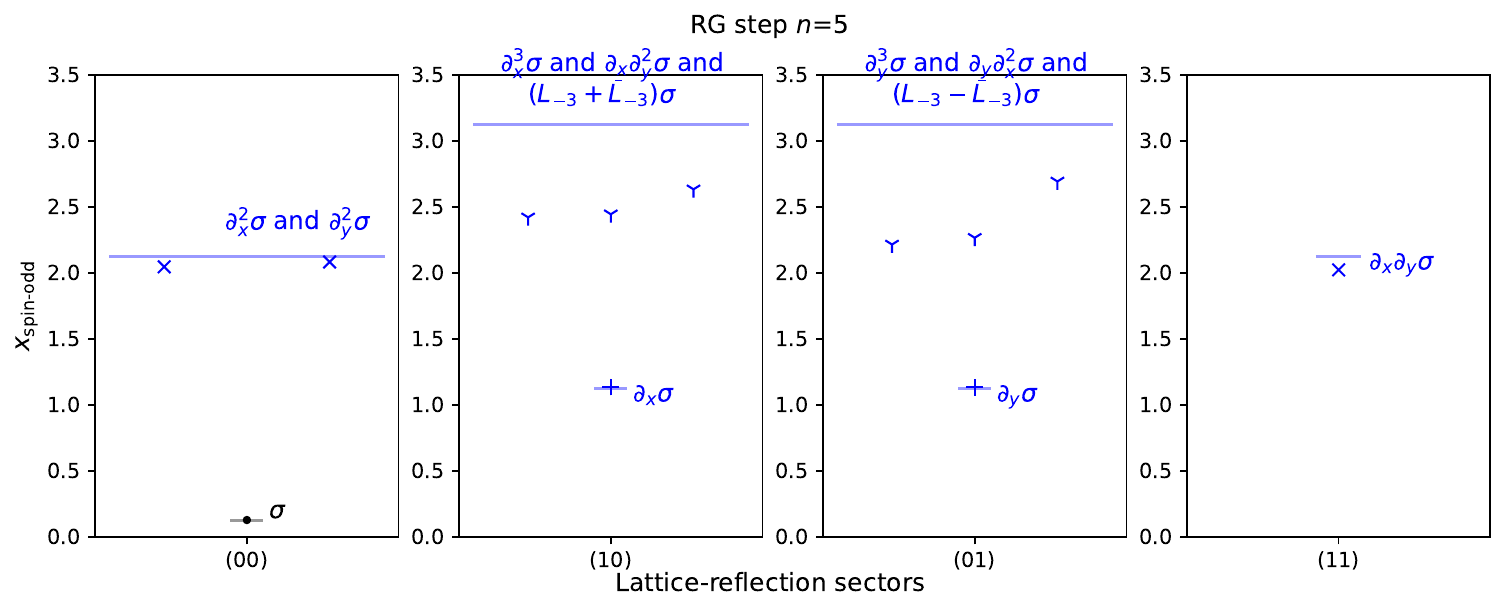}
    }
    \caption{\label{fig:2DscaleD} Scaling dimensions of the 2D Ising model organized by the spin-flip $\mathbb{Z}_2$ and the lattice-reflection symmetry sectors.
        Bond dimensions in the 2D TNRG are $\chi=36, \chi_s = 10$.
    }
\end{figure*}

The scaling dimensions of the 2D Ising university class according to the 2D CFT is usually organized by the conformal spin. 
The 2D CFT~\cite{cftbook:1997:2d} uses complex coordinates $z$ and $\bar{z}$ as the spatial coordinates, instead of the Cartesian coordinates. 
Therefore, some coordinate transformation is needed to figure out how the scaling dimensions distribute in different lattice-reflection sectors, which are parametrized in Cartesian coordinates.
We use the convention in Ref.~\cite{simmons-duffin-cftnotes}.
The 2D infinitesimal conformal maps on the space of functions have the following generators:
\begin{align}
    \label{eq:2dCFTln}
    l_m = z^{m+1} \partial_{z}, \bar{l}_m = \bar{z}^{m+1} \partial_{\bar{z}}, n \in \mathbb{Z}.
\end{align}
In the path-integral language, the conformal charges are given by
\begin{align}
    \label{eq:2dcftCharge}
    Q_V(\Sigma) = 
    \frac{1}{2 \pi i} \oint_{\Sigma} dz v(z)T(z)
    -
    \frac{1}{2 \pi i} \oint_{\Sigma} d\bar{z} \bar{v}(\bar{z}) \bar{T}(\bar{z}),
\end{align}
where $V = v(z) \partial_z + \bar{v}(\bar{z}) \partial_{\bar{z}}$ is a conformal Killing vector and $T, \bar{T}$ are components of the energy-momentum tensor.
Notice that the conformal charges are linear in the Killing vector.
The conformal generators $L_{m}$ and $\bar{L}_{m}$ that obey the Virasoro algebra are charges corresponding to  $l_m$ and $\bar{l}_{m}$ in Eq.~\eqref{eq:2dCFTln}:
\begin{align}
    \label{eq:2dcftviraL}
    L_{m}(\Sigma) = Q_{l_m}(\Sigma),
    \bar{L}_{m}(\Sigma) = Q_{\bar{l}_m}(\Sigma).
\end{align}

We will work out how the conformal generators for $m=-1, -3$ transform under spatial reflection.
Under the $x$ reflection where $x \to -x$ and $y \to y$, the complex coordinates transform like
$z \to - \bar{z}$ and $\bar{z} \to -z$.
For any odd $m$, using Eq.~\eqref{eq:2dCFTln}, generators of conformal maps on the space of functions transform like
$l_m \to - \bar{l}_m$ and  $\bar{l}_m \to - l_m$.
Since $Q_V$ in Eq.~\eqref{eq:2dcftCharge} is linear in $V$, the conformal generators $L_m$ and $\bar{L}_m$ transform in the same way as $l_m$ and $\bar{l}_m$.
Therefore, the linear combination $L_m + \bar{L}_m$ has $x$-reflection charge $c_x=1$, while $L_m - \bar{L}_m$ has $c_x=0$ for $m=-1, -3$.

The $y$ reflection, where $x \to x$ and $y \to -y$, becomes complex conjugation in complex coordinates.
Therefore, the linear combination $L_m + \bar{L}_m$ has $y$-reflection charge $c_y=0$, while $L_m - \bar{L}_m$ has $c_y=1$ for $m=-1, -3$.

Using the spatial reflection charges $c_x$ and $c_y$ for $L_m \pm \bar{L}_m$ for $m=-1,-3$, one can determine $(c_x, c_y)$ of all the scaling dimensions of the 2D Ising that are not larger than 3.125, which are summarized in~\autoref{tab:cftx-reflc-s0} and~\autoref{tab:cftx-reflc-s1}.
Notice that $L_{-1} + \bar{L}_{-1}$ and  $L_{-1} - \bar{L}_{-1}$ correspond to $\partial_x$ and $\partial_y$, respectively.

\begin{table}[tb]
    \caption{
    Scaling dimensions of the 2D Ising in the spin-flip even sector. 
The estimated $x_{\text{RG}}$ is from the linearization at the RG step $n=5$; 
the bond dimensions of the TNRG are $\chi=36, \chi_s=10$.
The symbol $\mathds{1}$ denotes the identity operator, $\epsilon$ the energy-density operator and $T_{mn}$ the energy-momentum operator.
}\label{tab:cftx-reflc-s0}
    \begin{ruledtabular}
        \begin{tabular}{ccccc}
        Operator  & $x_{\text{CFT}}$ & $(c_x, c_y)$   & $x_{\text{RG}}$   \\
        \colrule
        $\mathds{1}$     & 0              & (0, 0)       & 0  \\
        $\epsilon$     & 1              & (0, 0)       & 1.016  \\
        $T_{xx} - T_{yy}$     & 2              & (0, 0)   & 2.021  \\
        $\partial_x \epsilon$     & 2              & (1, 0)   & 2.137  \\
        $\partial_y \epsilon$     & 2              & (0, 1)  & 2.003   \\
        $T_{xy}$     & 2              & (1, 1)       & 2.016  \\
        \colrule
        \multirow{2}{*}{$\partial_x^2 \epsilon, \partial_y^2 \epsilon$ }  & 3  & \multirow{2}{*}{(0, 0)}       & 2.574  \\
             & 3     &        & 2.574  \\
        \colrule
        $\partial_x T_{xx}$     & 3              & (1, 0)  & 2.527  \\
        $\partial_y T_{yy}$     & 3              & (0, 1)  & 2.539   \\
        $\partial_x \partial_y \epsilon$     & 3     & (1, 1)   & 2.387  \\
    \end{tabular}
    \end{ruledtabular}
\end{table} 

\begin{table}[tb]
    \caption{
    Scaling dimensions of the 2D Ising in the spin-flip odd sector.
    For the estimated $x_{\text{RG}}$, the RG step of the linearization and the bond dimensions of the TNRG are the same as those in~\autoref{tab:cftx-reflc-s0}.
    The symbol $\sigma$ denotes the spin operator.
}\label{tab:cftx-reflc-s1}
    \begin{ruledtabular}
    \begin{tabular}{ccccc}
        Operator  & $x_{\text{CFT}}$ & $(c_x, c_y)$   & $x_{\text{RG}}$   \\
        \colrule
        $\sigma$  & 0.125  & (0, 0)       & 0.128  \\
        $\partial_x\sigma$  & 1.125  & (1, 0)   & 1.136  \\
        $\partial_y\sigma$  & 1.125  & (0, 1)   & 1.134  \\
        \colrule
        \multirow{2}{*}{$\partial_x^2 \sigma, \partial_y^2 \sigma$}  & 2.125  & \multirow{2}{*}{(0, 0)}       & 2.047  \\
          & 2.125  &        &  2.084 \\
        \colrule
        $\partial_x \partial_y\sigma$  & 2.125  & (1, 1) & 2.024  \\
        \colrule
        \multirow{3}{*}{\shortstack{$\partial_x^3 \sigma, \partial_x \partial_y^2 \sigma$ and\\ $ (L_{-3} + \bar{L}_{-3}) \sigma$ }}  & 3.125  & \multirow{3}{*}{(1, 0)}       & 2.421  \\
          & 3.125  &        & 2.445  \\
          & 3.125  &        & 2.634  \\
        \colrule
        \multirow{3}{*}{\shortstack{$\partial_y^3 \sigma, \partial_y \partial_x^2 \sigma$ and\\ $ (L_{-3} - \bar{L}_{-3}) \sigma$}}  & 3.125  & \multirow{3}{*}{(0, 1)}       & 2.216  \\
          & 3.125  &        &  2.267 \\
          & 3.125  &        &  2.693 \\
    \end{tabular}
    \end{ruledtabular}
\end{table}

We draw the CFT values of the scaling dimensions using short and long horizontal lines in~\autoref{fig:2DscaleD}. 
A short line indicates no degeneracy in a particular sector, while a long line indicates degeneracy. 
For example, the scaling dimensions of $\partial_x^2 \epsilon$ and $\partial_y^2 \epsilon$ have $2$-fold degeneracy in the lattice-reflection $c_x = c_y=0$ and spin-flip even sector; thus, they cannot be distinguished from each other in the numerical results.

It has already been observed in previous study~\cite{Lyu:Xu:Kawashima:2021,Guo:Wei:2024} that the linearized tensor RG equation in 2D TNRG using the ``freezing'' method, at bond dimensions $\chi=24, 30$, cannot resolve scaling dimensions of the 2D Ising that are larger than $2$. 
In the numerical calculation here, we generate a tensor RG flow at the exact value of the critical temperature $T_c= 2 / \ln( 1 + \sqrt{2} )$ and linearize the RG map at RG step $n=5$. 
The bond dimensions are set to be $\chi = 36, \chi_s=10$ in order to check whether more scaling dimensions larger than 2 can be resolved (see more details about the choice of $\chi$ and $n$ in Appendix~\ref{app:choicechin}).
The scaling operators are identified according to their scaling dimensions in each
symmetry sectors. 
Take lattice-reflection $c_x=c_y=0$ and spin-flip $c=0$ as an example. 
The scaling dimensions are calculated from Eq.~\eqref{eq:rgeig2x} and are organized in ascending order. 
According to~\autoref{tab:cftx-reflc-s0}, the first three are identified as the scaling dimensions of $\mathds{1}, \epsilon$ and $T_{xx} - T_{yy}$, while the fourth and fifth ones are $\partial_x^2 \epsilon, \partial_y^2 \epsilon$ doublet.

The numerical results are plotted as data points with different shapes in~\autoref{fig:2DscaleD}, and they are summarized in~\autoref{tab:cftx-reflc-s0} and~\autoref{tab:cftx-reflc-s1} as $x_{\text{RG}}$. 
The scaling dimensions distribute in different symmetry sectors as the CFT expectation, which provides a numerical justification of linearized RG map in Eqs.~\eqref{eq:R2Ddecomp} to~\eqref{eq:2dlinLx}. 
Meanwhile, however, the numerical results reveal several limitations of the linearization method for calculating scaling dimensions in TNRG developed in Ref.~\cite{Lyu:Xu:Kawashima:2021}. 
First, scaling dimensions larger than $2$ still cannot be resolved even at a bond dimension larger than previous study~\cite{Lyu:Xu:Kawashima:2021, Guo:Wei:2024}. 
As a comparison, the transfer matrix constructed from two copies of the same fixed-point tensor is able to resolve scaling dimensions of the 2D Ising up to 4. 
Secondly, the estimated scaling dimensions of $\partial_x \epsilon$ seems to have larger deviations compared with the values in previous work using smaller bond dimensions~\cite{Lyu:Xu:Kawashima:2021, Guo:Wei:2024}. 
This could be due to the misidentification of the scaling dimensions in previous work where there is no symmetry charge information of the eigenvector of the linearized RG map. 
The estimated scaling dimensions of the second descendants of the spin operator $\sigma$ (they are $\partial_x^2 \sigma, \partial_y^2 \sigma$ and $\partial_x \partial_y \sigma$) could be mistaken as one of the first descendants of the energy density operator $\epsilon$ (they are $\partial_x \epsilon$ and $\partial_y \epsilon$), since they have values closer to $2$.

\subsection{The 3D Ising model\label{subsec:ising3d}}
The efficiency of the 3D EF-enhance TNRG algorithm summarized in~\autoref{subsec:3dalgo} has been demonstrated in a previous study~\cite{Lyu:Kawashima:2024} using the cubic-lattice Ising model. 
The main point of Ref.~\cite{Lyu:Kawashima:2024} is that the 3D scheme is able to produce estimates of scaling dimensions that are stable with respect to the RG step, emphasizing the accuracy of the scaling dimensions of the two primary operator: the spin operator $\sigma$ and the energy-density operator $\epsilon$. 

Here, we dive deeper into the numerical results in Ref.~\cite{Lyu:Kawashima:2024} for bond dimension $\chi=6, \chi_s = \chi_m = 4$ by organizing the estimated scaling dimensions according to the lattice-reflection and spin-flip symmetry sectors. 
Using the first few digits of the bootstrap estimates of the primary operators $\sigma, \epsilon, \epsilon', \sigma_{mn}$, as well as the state-operator correspondence in CFT, one can sort out how these primary operators and the descendants of $\sigma$ and $\epsilon$ with their scaling dimensions less than 4.5 distribute in different lattice-reflection sectors; 
we summarize the CFT prediction in~\autoref{tab:cftx3D-reflc-s0} and~\autoref{tab:cftx3D-reflc-s1}\footnote{
    The lattice-reflection sector $c_x = c_y = c_z = 1$ is ignored, since they contain few scaling dimensions that are smaller than 4.5; in the spin-flip even sector, the lowest level is a triplet $\partial_z T_{xy}$ and its permutations with scaling dimension 4, while in the spin-flip odd sector, the lowest level is a singlet $\partial_x \partial_y \partial_z \sigma$ with scaling dimension $3.518$.
    It does not seem that the numerical results can resolve such high scaling dimensions well.
}. 
The CFT prediction of the scaling dimensions is visualized by plotting them in~\autoref{fig:3DscaleD} using horizontal lines, organized by lattice-reflection and spin-flip sectors. 
Similar to the 2D plot, a short line indicates no degeneracy (a singlet) in a particular sector, while a long line indicates degeneracy (a multiplet). 
For example, in the lattice-reflection $c_x=c_y=c_z=0$ and spin-flip even sector, the third and fourth scaling dimensions are a doublet, corresponding to the energy-momentum operator with its two indices being the same, $T_{xx}$ and $T_{yy}$ (the traceless condition reduces the degrees of freedom to 2).

\begin{table}[tb]
    \caption{
    Scaling dimensions of the 3D Ising in the spin-flip even sector. 
    The subscript $k$ takes values in $x, y, z$.
    The $x_{\text{CFT}}$ of the primary operators are the first few digits of the bootstrap estimates.
    The estimated $x_{\text{RG}}$ is from the linearization at the RG step $n=5$; 
the bond dimensions of the TNRG are $\chi=6, \chi_s=\chi_m=4$.
}\label{tab:cftx3D-reflc-s0}
    \begin{ruledtabular}
    \begin{tabular}{ccccc}
        Operator  & $x_{\text{CFT}}$ & $(c_x, c_y, c_z)$   & $x_{\text{RG}}$   \\
        \colrule
        $\mathds{1}$     & 0       & (0, 0, 0)       & 0  \\
        $\epsilon$      & 1.413         & (0, 0, 0)  & 1.411  \\
        $\partial_x \epsilon$     & 2.413            & (1, 0, 0) & 2.580  \\
        $\partial_y \epsilon$     & 2.413            & (0, 1, 0)  & 2.572  \\
        $\partial_z \epsilon$     & 2.413            & (0, 0, 1) & 2.589  \\
        \colrule
        \multirow{2}{*}{$T_{xx}, T_{yy}$}  & 3  & \multirow{2}{*}{(0, 0, 0)}  & 3.176  \\
             & 3              &        & 3.214  \\
        \colrule
        $T_{xy}$     & 3            & (1, 1, 0)  &  3.233  \\
        $T_{xz}$     & 3            & (1, 0, 1)  & 3.215  \\
        $T_{yz}$     & 3            & (0, 1, 1)  & 3.228   \\
        \colrule
        \multirow{3}{*}{\shortstack{$\partial_x^2 \epsilon, \partial_y^2 \epsilon$ \\ and $\partial_z^2 \epsilon$}}     & 3.413   & \multirow{3}{*}{(0, 0, 0)}   & 3.430  \\
             & 3.413            &        &  3.490 \\
             & 3.413            &        &  3.722 \\
        \colrule
        $\partial_x\partial_y \epsilon$     & 3.413    & (1, 1, 0)     & 4.552  \\
        $\partial_x\partial_z \epsilon$     & 3.413   & (1, 0, 1)   & 4.731   \\
        $\partial_y\partial_z \epsilon$     & 3.413    & (0, 1, 1)   & 5.143   \\
        $\epsilon'$     & 3.830            & (0, 0, 0)   & 4.020  \\
        \colrule
        \multirow{3}{*}{\shortstack{$\partial_x T_{kk}$ and $\partial_k T_{kx}$}}  & 4  & \multirow{3}{*}{(1, 0, 0)}   & 4.052  \\
          & 4  &        & 4.129  \\
          & 4  &        & 4.596 \\
        \colrule
        \multirow{3}{*}{\shortstack{$\partial_y T_{kk}$ and $\partial_k T_{ky}$}}  & 4  & \multirow{3}{*}{(0, 1, 0)}       & 4.146  \\
          & 4  &        & 4.190  \\
          & 4  &        & 5.155  \\
        \colrule
        \multirow{3}{*}{\shortstack{$\partial_z T_{kk}$ and $\partial_k T_{kz}$}}  & 4  & \multirow{3}{*}{(0, 0, 1)}  & 4.102  \\
          & 4  &        & 4.272  \\
          & 4  &        & 5.116  \\
    \end{tabular}
    \end{ruledtabular}
\end{table}

\begin{table}[tb]
    \caption{
    Scaling dimensions of the 3D Ising in the spin-flip odd sector. 
    For the quoted $x_{\text{CFT}}$ of the primary operators and the estimated $x_{\text{RG}}$, see more information in~\autoref{tab:cftx3D-reflc-s0}.
}\label{tab:cftx3D-reflc-s1}
    \begin{ruledtabular}
    \begin{tabular}{ccccc}
        Operator  & $x_{\text{CFT}}$ & $(c_x, c_y, c_z)$   & $x_{\text{RG}}$   \\
        \colrule
        $\sigma$      & 0.518      & (0, 0, 0)   & 0.557  \\
        $\partial_x \sigma$     & 1.518    & (1, 0, 0)   & 1.642  \\
        $\partial_y \sigma$     & 1.518     & (0, 1, 0)  & 1.685  \\
        $\partial_z \sigma$     & 1.518    & (0, 0, 1)  & 1.662   \\
        \colrule
        \multirow{3}{*}{\shortstack{$\partial_x^2 \sigma, \partial_y^2 \sigma$\\ and $\partial_z^2 \sigma$}}  & 2.518  & \multirow{3}{*}{(0, 0, 0)}  & 2.389  \\
             & 2.518              &        & 2.408  \\
             & 2.518              &        & 2.506  \\
        \colrule
        $\partial_x\partial_y \sigma$     & 2.518   & (1, 1, 0)  & 3.631   \\
        $\partial_x\partial_z \sigma$     & 2.518  & (1, 0, 1)  & 3.900   \\
        $\partial_y\partial_z \sigma$     & 2.518    & (0, 1, 1)   & 3.873  \\
        \colrule
        \multirow{3}{*}{$\partial_x\partial_k^2 \sigma$}  & 3.518  & \multirow{3}{*}{(1, 0, 0)}       &  3.884 \\
             & 3.518            &        & 3.928  \\
             & 3.518            &        & 4.031  \\
        \colrule
        \multirow{3}{*}{$\partial_y\partial_k^2 \sigma$}  & 3.518  & \multirow{3}{*}{(0, 1, 0)}       & 3.781  \\
             & 3.518            &        & 4.625  \\
             & 3.518            &        & 4.925  \\
        \colrule
        \multirow{3}{*}{$\partial_z\partial_k^2 \sigma$}  & 3.518  & \multirow{3}{*}{(0, 0, 1)}       & 3.820  \\
             & 3.518            &        & 4.451  \\
             & 3.518            &        & 4.946  \\
        \colrule
         $\partial_x \partial_y \partial_z \sigma$    & 3.518    &  (1, 1, 1)  &   \\
        \colrule
         \multirow{2}{*}{$\sigma_{xx}, \sigma_{yy}$}  & 4.180  & \multirow{2}{*}{(0, 0, 0)}       &  4.305  \\
             & 4.180            &        & 4.402  \\
        \colrule
         $\sigma_{xy}$     & 4.180            & (1, 1, 0)  & 4.316  \\
         $\sigma_{xz}$     & 4.180            & (1, 0, 1)  & 3.940  \\
         $\sigma_{yz}$     & 4.180            & (0, 1, 1)  & 4.326  \\
    \end{tabular}
    \end{ruledtabular}
\end{table}

\begin{figure*}[tb]
    \subfloat[\label{fig:3Dspinflip0}
    Spin-flip even sector.
    The symbol $\mathds{1}$ denotes the identity operator, $\epsilon$ the energy-density operator, $T_{ij}$ the energy-momentum operator and $\epsilon'$ another primary operator in the spin-flip even sector.
    ]{
        \includegraphics[width=1.70\columnwidth,
        valign=c]{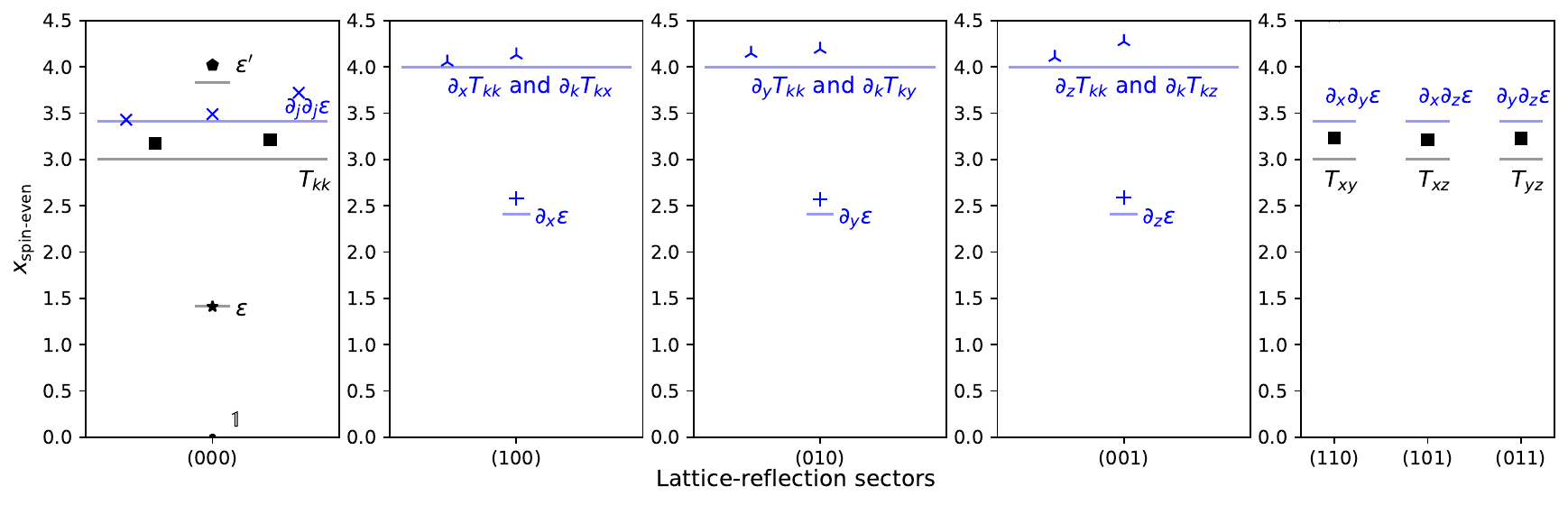}
    }\\
    \subfloat[\label{fig:3Dspinflip1}
    Spin-flip odd sector.
    The symbol $\sigma$ denotes the spin operator and $\sigma_{ij}$ another primary operator in the spin-flip odd sector.
    ]{
        \includegraphics[width=1.70\columnwidth,
        valign=c]{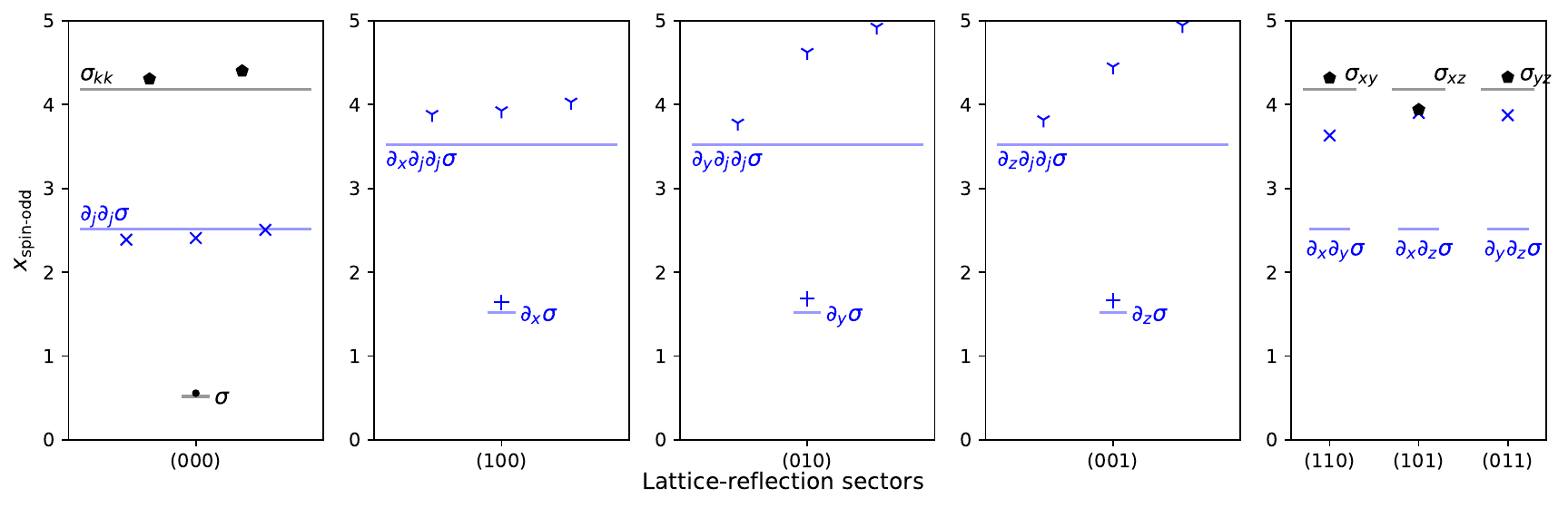}
    }
    \caption{\label{fig:3DscaleD} Scaling dimensions of the 3D Ising model organized by the spin-flip $\mathbb{Z}_2$ and the lattice-reflection symmetry sectors.
        The bond dimensions in the 3D TNRG are $\chi =6, \chi_s = \chi_m=4$.
    }
\end{figure*}

Since the critical temperature of the cubic-lattice Ising model is not known exactly, we use the RG flow of the tensor and a bisection method to estimate $T_c$, the details of which have been expounded in previous study~\cite{Lyu:Xu:Kawashima:2021,Lyu:Kawashima:2024}. 
At bond dimensions $\chi=6, \chi_s=\chi_m = 4$, the estimated critical temperature is $T_c \approx 4.54$, not very close to the known Monte Carlo value 4.5115. 
We generate a tensor RG flow at this estimated $T_c$ and linearize the RG map at the RG step $n=5$ (see more details about the choice of $\chi$ and $n$ in Appendix~\ref{app:choicechin}). 
The scaling operators are identified according to their scaling dimensions in each symmetry sector, the same way as the 2D calculation in~\autoref{subsec:num2d}.

The numerical results are plotted as data points with different shapes in~\autoref{fig:3DscaleD}, and they are summarized in~\autoref{tab:cftx3D-reflc-s0} and~\autoref{tab:cftx3D-reflc-s1} as $x_{\text{RG}}$. 
Except the scaling dimensions of the spin operator $\sigma$ and the energy-density $\epsilon$, most of the other scaling dimensions are far from accurate. 
However, the degeneracy structure of the scaling dimensions $x \leq 3$ agrees with the CFT anticipation in each symmetry sector and among sectors. 
For example, the three-fold degeneracy of the first descendants among the lattice-reflection $(c_x, c_y, c_z) = (1,0,0), (0,1,0)$ and $(0,0,1)$ sectors is quite clear in our numerical estimates for both spin-flip even and odd sectors (they are denoted by $\partial_k \epsilon$ and $\partial_k \sigma$ for $k=x,y,z$). 
Moreover, the five-fold degeneracy of the energy-momentum operator $T_{mn}$ is also clear, as well as how they are distributed into a doublet in $(000)$ sector and three singlets in sectors $(110), (101), (011)$. 
This agreement is a numerical justification of the claimed linearized RG map in Eqs~\eqref{eq:R3Ddecomp} to~\eqref{eq:3dlinLx}.

Just like the 2D numerical demonstration in~\autoref{subsec:num2d}, the current linearization scheme has difficulty resolving higher scaling dimensions of descendant operators. 
Take as an example the second descendants of the energy density, $\partial_i \partial_j \epsilon$, with 6-fold degeneracy and scaling dimension 3.413. 
Only two among the three in the lattice-reflection $(c_x, c_y, c_z) = (0,0,0)$ sector are close to this CFT prediction, with the third one slightly far way. 
The estimates in the lattice-reflection sectors $(110), (101), (011)$ are much worse; 
they do not show degeneracy and are very far away from 3.413, making them unreliable. 
The RG estimates of the second descendants of the spin operator $\partial_i \partial_j \sigma$ show a similar pattern. 
For the first descendants of the energy-momentum operator in the lattice-reflection $(c_x, c_y, c_z) = (1,0,0), (0,1,0)$ and $(0,0,1)$ sectors, the degeneracy structure is wrong: only two out of three in each sector are degenerate. 
For the third descendants of the spin operator, $\partial_i \partial_j \partial_k \sigma$, the approximate 3-fold degeneracy seems to be present in the lattice-reflection $(1,0,0)$ sector, but not in the $(0,1,0)$ and $(0,0,1)$ sectors.

Quite surprisingly, the RG estimates of higher primary fields are not too bad. 
The second-smallest primary operator $\epsilon'$ in the spin-flip even sector has scaling dimension 4.020 estimated from the linearized RG map;
this is about $5\%$ away from the bootstrap estimate 3.830. 
In the spin-flip odd sector, the second-smallest primary operator is a rank-2 traceless symmetric tenor field $\sigma_{ij}$, whose scaling dimension is 4.180 according to the bootstrap estimation. 
The RG estimate that is farthest from this value is that of $\sigma_{xz}$ in the lattice-reflection $(c_x, c_y, c_z) = (1,0,1)$ sector, whose value is 3.940, about $6\%$ away from the bootstrap value.

\section{Summary and discussions}
In this paper, we develop a toolkit for exploiting and imposing lattice-reflection symmetry in the context of coarse graining a tensor-network consisting of real-valued tensors using the TNRG.
The essential contributions of this paper are 
1) clarifying the origin of the definition of the lattice-reflection symmetry for square- and cubic-lattice tensor network,
2) proposing a transposition trick to imposing the lattice-reflection symmetry,
3) working out the implication of the lattice-reflection symmetry in the EF and projective truncations,
4) developing a general technical for proving lattice symmetries by dragging tensors around in a tensor network,
5) designing EF-enhanced TNRG algorithms in both 2D and 3D with the lattice-reflection symmetry exploited and imposed,
and 6) demonstrating how to linearize these two tensor RG equations in separate lattice-reflection symmetry sectors.
We demonstrate the validity of these developments by estimating the scaling dimensions of the square- and cubic-lattice Ising model using the proposed TNRG algorithms in 2D and 3D; 
the scaling dimensions are organized in various lattice-reflection and spin-flip symmetry sectors.

It should be emphasized that \emph{the proposed transposition trick does not introduce any additional assumption or approximation}.
The trick leaves the partition function exactly invariant when the tensor itself satisfies the definition of the lattice-reflection symmetry.
For TNRG schemes without EF, like the TRG~\cite{Levin:Nave:2007} and HOTRG~\cite{HOTRG:2012}, the lattice-reflection symmetry will be preserves up to the machine error in a few RG steps even without the transposition trick.
Using the transposition trick will reproduce the results of these methods if the numerical errors due to machine precision are ignored (recall the 1D example where the transposition trick changes the RG map from $A' = A A$ to $A' = A A^\intercal$, which are the same map when the matrix $A$ is symmetric).

One of the advantages of the transposition trick becomes clear when an EF process (like the Gilt~\cite{Hauru:2018}) is incorporated.
With this trick, it is straightforward to see the implication of the lattice-reflection symmetry in the EF process.
Without exploiting a lattice symmetry in tensor network, a general EF scheme, like the Gilt~\cite{Hauru:2018} or the FET~\cite{Evenbly:2018}, has the tendency to break that lattice symmetry.
The extent to which the lattice symmetry is broken naturally depends on the approximation error in the EF process.
In 3D, the typical EF error for the cubic-lattice Ising model is of order $10^{-3}$ to $10^{-2}$~\cite{Lyu:Kawashima:2024}, while in 2D, the EF error is smaller than $10^{-6}$; hence, exploiting the lattice-reflection symmetry becomes very more relevant in 3D than the 2D TNRG.
Besides preventing the breaking of the symmetry in an 3D EF scheme, exploiting the lattice-reflection symmetry also reduces the number of filtering matrices in the 3D EF dramatically from 24 to 3 (see Eq.~\eqref{eq:cubeEFapprox} and Eq.~\eqref{eq:3dEFapproxSym}); this is very helpful in a practical numerical implementation of a scheme.

There are several 2D TNRG schemes that can preserve the lattice-reflection~\cite{Evenbly:2017:algo} and lattice-rotation symmetry~\cite{Yang:Gu:Wen:2017} by writing down a certain ansatz in their approximations. 
The validity of their ansatz can be justified \emph{a posteriori} if the method works fine with a particular model.
The developments made in the current paper provide a framework for a deeper understanding of the ansatz in those methods.
This deeper understanding can be useful when those lattice-symmetry-preserving methods do not work well on a certain model.
In such a situation, the discussion of this paper might guide the modification or generalization of those 2D lattice-symmetry-preserving schemes.

To conclude, we mention some open problems concerning the lattice symmetry in the TNRG\@.
For the lattice-reflection symmetry, the current paper focuses on tensor networks consisting of real-valued tensors.
It is natural to explore how to generalize the discussion to complex-valued tensors.
Although, it has been proposed in Ref.~\cite{Evenbly:2017:algo} that for complex-valued tensors, the SWAP-gauge matrix $g$ becomes a unitary matrix, our numerical experiments indicate that the SWAP-gauge matrix $g$ trivializes to the identity matrix.

With the general proving techniques developed in this paper, it also seems possible to study the lattice-rotational symmetry in the TNRG\@.
The lattice-rotation symmetry in the 2D TNRG will be studied in a future paper.
In the 3D TNRG, however, the lattice-rotation symmetry is still an open problem; 
one of the difficulties of generalizing the current discussion to the lattice-rotation symmetry in 3D is that the rotation group becomes non-abelian in 3D.
Considering the recent success of the fuzzy-sphere method for studying the (2+1)D quantum criticality~\cite{fuzzysphere:2023}, where the rotational symmetry is fully exploited, it is tempting to conjecture that exploiting the lattice-rotation symmetry might also improve the efficiency of a 3D TNRG map and quality of a critical fixed-point tensor.

\acknowledgments

We thank Glen Evenbly for the tutorial about the TNRG implementation on his website \url{https://www.tensors.net/}, which inspires us to study the lattice symmetry in TNRG. 
We thank Katsuya Akamatsu, Kenji Homma, Feng-Feng Song and Satoshi Morita for useful discussions.
We thank Slava Rychkov, Clement Delcamp for comments and suggestions regarding the overall structure of this paper.
We thank Nikolay Ebel and Rajeev S. Erramilli for clarifying the CFT prediction of the scaling dimensions of the Ising model in 2D and 3D.
X.L. is grateful to the support of the Global Science Graduate Course (GSGC) program of the University of Tokyo.
This work is financially supported by JSPS KAKENHI Grant No. 23K25789.
The computation in this work has been done using the facilities of the Supercomputer Center, the Institute for Solid State Physics, the University of Tokyo.
The manuscript was written and the numerical demonstration in 2D was performed after X.L. moved to Institut des Hautes \'Etudes Scientifiques.

\appendix

\section{Determining the filtering matrices\label{app:findsmat}}
In this appendix, we discuss a scheme for determining the filtering matrices in the EF approximation introduced in~\autoref{subsec:EFformula}.
It suffices to explain in the 2D filtering of the plaquette in Eq.~\eqref{eq:2dEFapprox}; the generalization to the cube filtering in Eq.~\eqref{eq:cubeEFapprox} is straightforward.

\subsection{Optimization of the filtering matrices\label{subsec:opts}}
We use the formalism developed in the full environment truncation (FET)~\cite{Evenbly:2018} to determine the filtering matrices that give a good approximation to the target patch of an entanglement filtering. 

The two tensor-network diagrams in Eq.~\eqref{eq:2dEFapprox} can be seen as two ket vectors, and fidelity $F$ can be used to quantify how good the approximation is. 
To make the pictorial representation look like the familiar Dirac notation, we bend all physical legs to the left,
\begin{align}
    \label{eq:diagram2state}
    \ket{\Psi} =
    \includegraphics[width=0.30\columnwidth, valign=c]{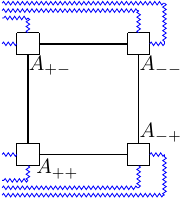},
    \ket{\Phi} = 
    \includegraphics[width=0.30\columnwidth, valign=c]{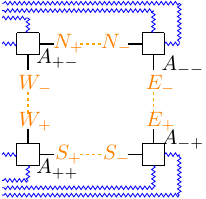},
\end{align}
and the corresponding bra vectors $\bra{\Psi}, \bra{\Phi}$ are obtained by mirror-reflecting these two diagrams. 
The fidelity $F(\Psi, \Phi)$ between the original state $\ket{\Psi}$ and the filtered state $\ket{\Phi}$ is defined to be
\begin{align}
    \label{eq:fidelityDef}
    F(\Psi,\Phi) \texteq{def}
    \frac{\braket{\Phi}{\Psi} \braket{\Psi}{\Phi}}{\braket{\Phi}{\Phi} \braket{\Psi}{\Psi}}.
\end{align}
This fidelity has values between 0 and 1.
We call the fidelity in Eqs.~\eqref{eq:diagram2state} and~\eqref{eq:fidelityDef} the EF fidelity.
The error associated with the EF approximation in Eq.~\eqref{eq:2dEFapprox} is naturally $\epsilon^2_{EF} = 1 -F$.

A good choice of filtering matrices, $N_-, N_+,\ldots$, maximizes the fidelity $F$. 
Since $\braket{\Psi}{\Psi}$ is a constant in this maximization problem, we only need to calculate the other two overlaps:
\begin{subequations}
    \label{eq:twoOverLaps}
\begin{align}
    \braket{\Psi}{\Phi} =
    \includegraphics[width=0.60\columnwidth, valign=c]{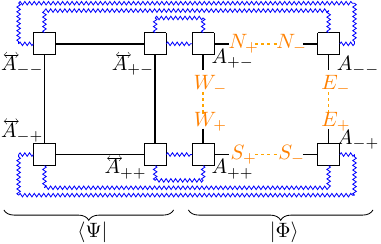}, \\
    \braket{\Phi}{\Phi} =
    \includegraphics[width=0.60\columnwidth, valign=c]{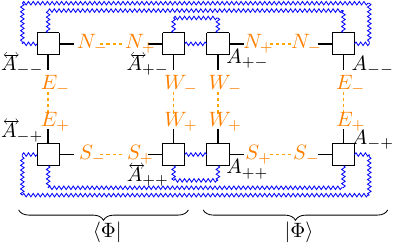},
\end{align}
and $\braket{\Phi}{\Psi}$ is the same as $\braket{\Psi}{\Phi}$ for tensors with real values.
In the above two diagrams, the double-arrow notation on a tensor like $\overleftrightarrow{A}_{++}$ means a transposition of two horizontal legs of the tensor.
\end{subequations}
The idea is to update one filtering matrix at a time. 
For example, when updating the $E_-$ matrix, the fidelity can be seen to have the following form:
\begin{subequations}
\begin{align}
    \label{eq:F4E}
    F \propto
    \includegraphics[scale=0.9, valign=c]{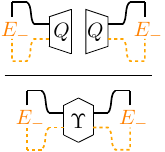}\quad,
\end{align}
where
\begin{align}
    \label{eq:UpsilonP}
    % \includegraphics[width=0.15\columnwidth, valign=c]{sec2-Upsilon.png}
    % \texteq{def}
    \includegraphics[width=0.85\columnwidth, valign=c]{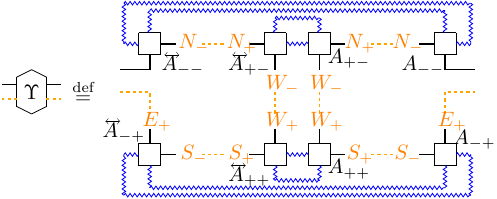},\\
    % \includegraphics[width=0.10\columnwidth, valign=c]{sec2-Q.png}
    % \texteq{def}
    \includegraphics[width=0.85\columnwidth, valign=c]{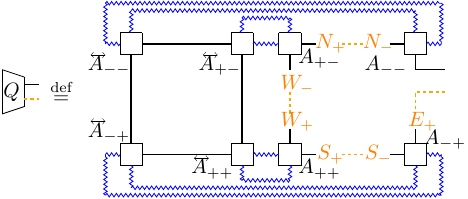}.
\end{align}
\end{subequations}
The optimization problem for $E_-$ is now put into the following form,
\begin{align}
    \label{eq:F4Eopt}
    \max_{x}
    \frac{\bra{x} N \ket{x}}{\bra{x} \Upsilon \ket{x}},
    \text{ with }
    N = \ket{Q} \bra{Q} =
    \includegraphics[scale=0.9, valign=c]{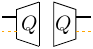}\quad,
\end{align}
and the solution is the largest eigenvector of the following generalized eigenvalue problem~\cite{EVP-Gen:2023}:
\begin{align}
    \label{eq:genEigProb}
    N \ket{x} = \lambda \Upsilon \ket{x}.
\end{align}
For our current problem, $N = \ket{Q} \bra{Q}$ is a rank-one matrix; hence the only eigenvector with a non-vanishing eigenvalue is
\begin{align}
    \label{eq:solx}
    \ket{x} = \Upsilon^{-1} \ket{Q}.
\end{align}
Therefore, the filtering matrix $E_-$ is updated according to
\begin{align}
    \label{eq:updateE}
    \includegraphics[scale=1.0, valign=c]{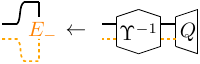}\quad.
\end{align}
After this update, the remaining filtering matrices are updated in the same manner.
This optimization process ends until the fidelity $F$ converges.

\subsection{Initialization of the filtering matrices\label{subsec:inits}}
We propose an efficient initialization scheme for filtering matrices; this scheme is inspired by a previous scheme called Graph independent local truncation (Gilt)~\cite{Hauru:2018}.
In fact, our initialization scheme is equivalent to the Gilt without recursion.
The reformulation we propose here makes the initialization more consistent with the optimization process.

We demonstrate the strategy for the 2D EF in Eq.~\eqref{eq:2dEFapprox}.
Let us show how to initialize the $E_-$ and $E_+$ pair in the EF approximation in Eq.~\eqref{eq:2dEFapprox}.
During their initialization, all other filtering matrices are treated as identity matrix. We combine the unknown $E_-$ and $E_+$ pair as a low-rank matrix $L_E$,
\begin{align}
    \label{eq:LEdef}
    \includegraphics[scale=1.0, valign=c]{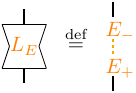}\quad.
\end{align}
The approximation of this initialization is thus
\begin{align}
    \label{eq:initApprox}
    \includegraphics[scale=0.8, valign=c]{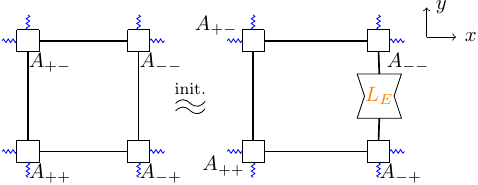}\quad.
\end{align}
Apply the same idea as the optimization procedure. 
The fidelity of this approximation can be written as
\begin{subequations}
\begin{align}
    \label{eq:F4Einit}
    F_0 \propto
    \includegraphics[scale=0.8, valign=c]{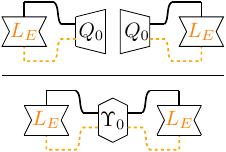}\quad,
\end{align}
where
\begin{align}
    \label{eq:UpsilonP0}
    % \includegraphics[width=0.15\columnwidth, valign=c]{sec2-Upsilon0.png}
    % &\texteq{def}
    &\includegraphics[width=0.80\columnwidth, valign=c]{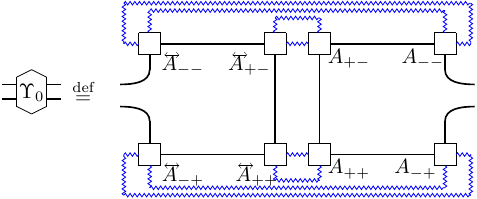}\quad,\\
    % \includegraphics[width=0.10\columnwidth, valign=c]{sec2-Q0.png}
    % &\texteq{def}
    &\includegraphics[width=0.80\columnwidth, valign=c]{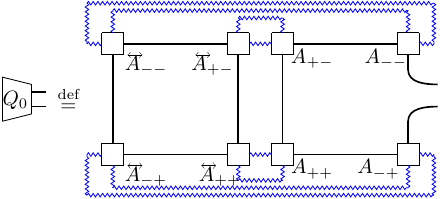}\nonumber\\
    &=
    \includegraphics[scale=0.9, valign=c]{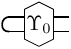}\quad.
\end{align}
\end{subequations}
However, the proposed initial $L_E$ matrix would be the trivial identity matrix if we use the update rule in Eq.~\eqref{eq:solx} since
\begin{align}
    \label{eq:trivialLE}
    \includegraphics[scale=1.0, valign=c]{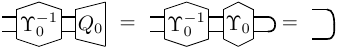}\quad.
\end{align}
This solution gives fidelity $F=1$, but fails to be a low-rank matrix. 
Inspired by the Gilt, we propose the following trick to avoid the full-rank trivial solution. 
The basic observation is that for a CDL tensor, the matrix $\Upsilon_0$ has a lower rank $\chi$ out of its possible full rank  $\chi^2$, and hence the inverse does not exist. 
Instead, we will use its Moore-Penrose inverse $\Upsilon_0^{+}$ that can be constructed using eigenvalue decomposition (EVD) for a real symmetric matrix like $\Upsilon_0$:
\begin{align}
    \label{eq:pinvDef}
    &\Upsilon_0 \texteq{EVD}
    U_0 \Lambda_0 U_0^\intercal, \text{ with }
    \Lambda_0 
    = \diag(\lambda_1, \lambda_2, \ldots, \lambda_{\chi^2}).\nonumber\\
    &\Upsilon_0^{+} \texteq{def}
    U_0 \Lambda_0^{+} U_0^\intercal, \text{ with} \nonumber\\
    &\Lambda_0^{+} \texteq{def}
    \diag(\lambda_1^{-1}, \lambda_2^{-1}, \ldots, \lambda_k^{-1}, 0, 0, \dots, 0),
\end{align}
where in $\Lambda_0^{+}$, the first $k$ eigenvalues become the inverse of the original ones, while the remaining eigenvalues are set to be zero.
The matrix $L_E$ is then obtained by replacing the inverse matrix in Eq.~\eqref{eq:solx} by the Moore-Penrose inverse in Eq.~\eqref{eq:pinvDef},
\begin{align}
    \label{eq:LEsol}
    \includegraphics[scale=1.0, valign=c]{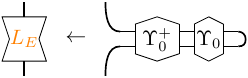}\quad.
\end{align}
The two filtering matrices $E_{+}$ and $E_{-}$ are initialized through a truncated SVD of $L_E$ according to the definition of $L_E$ in Eq.~\eqref{eq:LEdef}.
Other filtering matrices are initialized in the same way.
For a CDL tensor, by choosing $k = \chi$ in the Moore-Penrose inverse and $\chi_s = \sqrt{\chi}$ in the truncated SVD of $L_E$ (the dashed leg of a filtering matrix is filtered and has bond dimension $\chi_s$), the above initialization strategy gives filtering matrices that have fidelity $F=1$ for the EF approximation in Eqs.~\eqref{eq:diagram2state} and~\eqref{eq:fidelityDef}. 
For a given $\chi_s$, a rule of thumb for choosing the number of eigenvalues to be inverted in the Moore-Penrose inverse in Eq.~\eqref{eq:pinvDef} is $k = \chi_s^2$.

\section{Linearization in the 1D toy example as matrices\label{app:2dLinMatrix}}
If we choose the basis of the vector space $\MnnF$ such that the first $N_0 = n(n+1)/2$ basis vectors span the symmetric subspace $\Mnn{0}$ and the second $N_1 = n(n-1)/2$ basis vector span the antisymmetric subspace $\Mnn{1}$, the linear map $\mathcal{R}$ in Eq.~\eqref{eq:linearbkten1D}, due to Theorem~\ref{theo:R1Dbreak}, has the following block-diagonal matrix form:
\begin{align}
    \label{eq:R1Dblockform}
    \mathcal{R} =
\begin{pmatrix}
    \mathcal{R}^{(0)}_s & 0 \\
    0 & \mathcal{R}^{(1)}_s \\
\end{pmatrix},
\end{align}
where $\mathcal{R}^{(0)}_s$ is an $N_0$-by-$N_0$ matrix and $\mathcal{R}^{(1)}_s$ an $N_1$-by-$N_1$ matrix, and the subscript ``s'' means the two matrices has smaller dimension than $\mathcal{R}$.
Therefore, $\mathcal{R}$ is an $N$-by-$N$ matrix with $N = N_0 + N_1 = n^2$.
In this choice for the basis of $\MnnF$, due to Eqs.~\eqref{eq:R1D2RcinMnnc} and~\eqref{eq:linearmapM2Mp}, the constructed linear map $\mathcal{R}^{(c)}, c=0,1$ in Eq.~\eqref{eq:linearbkten1Dc} are $N$-by-$N$ matrices and have the following expression:
\begin{align}
    \label{eq:R1Dcblockform}
    \mathcal{R}^{(0)} =
\begin{pmatrix}
    \mathcal{R}^{(0)}_s & K^{(0)} \\
    0 & 0 \\
\end{pmatrix},
    \mathcal{R}^{(1)} =
\begin{pmatrix}
    0 & 0 \\
    K^{(1)} & \mathcal{R}^{(1)}_s \\
\end{pmatrix},
\end{align}
where $K^{(0)}$ and $K^{(1)}$ are two unknown matrices, and they are nonzero.
However, the presence of non-vanishing $K^{(c)}$ in $\mathcal{R}^{(c)}$ does not change the fact that the nonzero eigenvalues of $\mathcal{R}^{(c)}$ are the same as those of $\mathcal{R}^{(c)}_s$.
Take $\mathcal{R}^{(0)}$ as an example.
Its eigenvalue problem is
\begin{align}
    \label{eq:R1D0eigprob}
    \mathcal{R}^{(0)}
\begin{pmatrix}
    \vec{v}^{(0)}  \\
    \vec{v}^{(1)}  \\
\end{pmatrix} =
\begin{pmatrix}
    \mathcal{R}^{(0)}_s \vec{v}^{(0)} + K^{(0)}  \vec{v}^{(1)}\\
    0  \\
\end{pmatrix} =
\lambda \begin{pmatrix}
    \vec{v}^{(0)}  \\
    \vec{v}^{(1)}  \\
\end{pmatrix},
\end{align}
with $\lambda \neq 0$.
This indicates $\vec{v}^{(1)} = 0$ and the following eigenvalue problem for $\mathcal{R}^{(0)}_s$,
\begin{align}
    \label{eq:R1D0seigprob}
    \mathcal{R}^{(0)}_s \vec{v}^{(0)} = \lambda \vec{v}^{(0)}.
\end{align}
This means the nonzero eigenvalue spectrum of $\mathcal{R}^{(0)}$ is the same as that of $\mathcal{R}^{(0)}_s$.

\section{Proof of the claimed linearization in 2D\label{app:2dLinProof}}
In this subsection, we give a proof that the linearization $\mathcal{R}_{\text{2D}}^{(c_x, c_y)}$ in 2D, which we wrote down from Eq.~\eqref{eq:R2Ddecomp} to Eq.~\eqref{eq:2dlinLx}, is the expected linear map.
Concretely, we mean that the spectrum of $\mathcal{R}_{\text{2D}}^{(c_x, c_y)}$ contains all nonzero eigenvalues of $\mathcal{R}_{\text{2D}}^{\text{ntt}}$ whose eigenvectors live in the subspace $\mathfrak{T}_4^{(c_x, c_y)}$ defined in Eq.~\eqref{eq:T4cxcyDef}, where $\mathcal{R}_{\text{2D}}^{\text{ntt}}$ is the linearization of the corresponding tensor RG equation without the transposition trick.
The key result of this appendix is summarized in Theorem~\ref{theo:keyres}.

The proof in 2D goes in parallel with that of the 1D toy example in~\autoref{subsec:linRG1D}; 
In 1D, we showed two things in the proof of Theorem~\ref{theo:Rc1Dbuild}:
\begin{enumerate}
    \item $\mathcal{R}^{(c)}$ in Eq.~\eqref{eq:linearbkten1Dc} is the same linear map as the linearization $\mathcal{R}$ in Eq.~\eqref{eq:linearbkten1D} for any input $\delta A \in \Mnn{c}$.
    \item $\mathcal{R}^{(c)}$ in Eq.~\eqref{eq:linearbkten1Dc} maps any element in $\mathcal{M}_{nn}(\mathbb{R})$ into an element in $\mathcal{M}_{nn}^{(c)}(\mathbb{R})$.
        To put it differently, the range of the linear map $\mathcal{R}^{(c)}$ is $\Mnn{c}$.
\end{enumerate}

In 2D, however, due to the presence of the filtering matrices $s_x, s_y$ and the isometric tensors $p_x, p_y, p_i$ in the tensor RG equation in Eq.~\eqref{eq:2dEFhotRGeq}, the proof is more involved than that in the 1D toy example.
To simplify the proof, we turn off the EF by setting filtering matrices $s_x, s_y$ to identity and omit the projective truncation for inner legs of the $2 \times 2$ block. 
Later, we will sketch how the proof is modified after turning on the EF and applying projective truncation for inner legs.
We will conduct the 2D proof in three steps by showing that
\begin{enumerate}
    \item Once the isometric tensors $p_x, p_y$ in the tensor RG equation $\mathcal{T}_{\text{2D}}$ in Eq.~\eqref{eq:2dEFhotRGeq} with the transposition trick are determined, the isometric tensors $p_x^{\text{ntt}}, p_y^{\text{ntt}}$ in the tensor RG equation $\mathcal{T}_{\text{ntt}}$ in Eqs.~\eqref{eq:2dC2TyNoSym} to~\eqref{eq:2dTxTyTNoSym} without the transposition trick can be obtained by acting the SWAP-gauge matrices of the input tensor on $p_x$ and $p_y$.
    \item The constructed linear map $\mathcal{R}^{(c_x, c_y)}_{\text{2D}}$ maps any element $\mathfrak{T}_4$ into an element in $\Tsp{4}{(c_x, c_y)}$.
    To put it differently, the range of the linear map $\mathcal{R}^{(c_x, c_y)}_{\text{2D}}$ is $\Tsp{4}{(c_x, c_y)}$.
    \item The constructed linear map $\mathcal{R}^{(c_x, c_y)}_{\text{2D}}$ in Eq.~\eqref{eq:R2Ddecomp} is the same linear map as the corresponding linearization $\mathcal{R}^{\text{ntt}}_{\text{2D}}$ in Eq.~\eqref{eq:2dLinRGNoSym} without the transposition trick\footnote{
    Notice that Eq.~\eqref{eq:2dLinRGNoSym} is the linearization of the usual HOTRG. 
    In this proof, there will be no isometric tensor for the inner legs of the $2 \times 2$ block; this will be more clear later when we explicitly write down $\mathcal{T}_{\text{ntt}}$ and $\mathcal{R}^{\text{ntt}}_{\text{2D}}$. 
}
for any input $\delta A \in \Tsp{4}{(c_x, c_y)} $.
\end{enumerate}

\begin{remark}
    Although the constructed linearization $\mathcal{R}^{(c_x, c_y)}_{\text{2D}}$ in Eq.~\eqref{eq:R2Ddecomp} can be evaluated at any tensor $A$, in the following proof, we will always take $A = A_*$ to be the fixed-point tensor of the corresponding tensor RG map after a proper gauge fixing procedure.
    This means that the SWAP-gauge matrices $g_x, g_y$ are the same before and after the RG map; thus, the subspace $\Tsp{4}{(c_x, c_y)}$ also remains the same\footnote{
        Recall that in the definition of this subspace in Eq.~\eqref{eq:T4cxcyDef}, one needs to specify the SWAP-gauge matrix set $(g_x, g_y)$.
    }.
\end{remark}

% Step 1 of the 2D proof
    The tensor RG equation $\mathcal{T}_{\text{2D}}$ with the transposition trick in Eq.~\eqref{eq:2dEFhotRGeq}, after the EF is turned off and without the projective truncation for inner legs, becomes
\begin{align}
    \label{eq:app-2dprf-RGeq}
    \includegraphics[scale=1.0, valign=c]{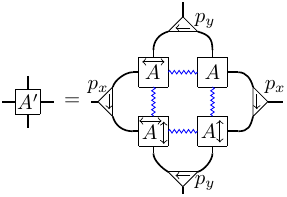}
    = \mathcal{T}_{\text{2D}}(A).
\end{align}
Without the transposition trick and the inner projective truncation, the tensor RG map $\mathcal{T}_{\text{ntt}}$ in Eqs.~\eqref{eq:2dC2TyNoSym} to~\eqref{eq:2dTxTyTNoSym} becomes Eq.~\eqref{eq:bkten}, where the isometric tensors are different from those in Eq.~\eqref{eq:app-2dprf-RGeq},
\begin{align}
    \label{eq:app-2dprf-RGeqntt}
    \includegraphics[scale=1.0, valign=c]{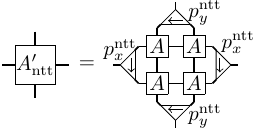}
    = \mathcal{T}_{\text{ntt}}(A).
\end{align}

\begin{theorem}
    \label{theo:2Dp2pntt}
    For the same input tensor $A$ that has the lattice-reflection symmetry defined in Eq.~\eqref{eq:refl2d} with SWAP-gauge matrices $g_x$ and $g_y$, the isometric tensors with and without the transposition trick can be related to each other according to
\begin{align}
    \label{eq:app-2dprf-p2pntt}
    \includegraphics[scale=1.0, valign=c]{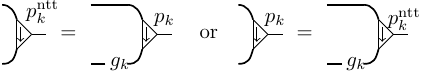},
\end{align}
for $k = x, y$.
\end{theorem}

\begin{proof}
    It suffices to prove the case $k=x$.
    According to the projective truncation reviewed in~\autoref{sec:projtrunc}, the isometric tensor $p_x$ in Eq.~\eqref{eq:app-2dprf-RGeq} is a collection of eigenvectors of the following density matrix
\begin{align}
    \label{eq:app-2dprf-rhopx}
    \includegraphics[scale=1.0, valign=c]{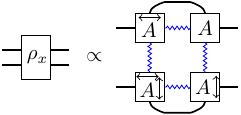}\quad,
\end{align}
while $p_x^{\text{ntt}}$ in Eq.~\eqref{eq:app-2dprf-RGeqntt} is a collection of eigenvectors of
\begin{align}
    \label{eq:app-2dprf-rhopxntt}
    \includegraphics[scale=1.0, valign=c]{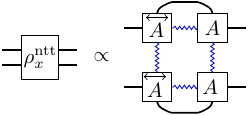}\quad.
\end{align}
Using the lattice-reflection symmetry of the input tensor $A$ in Eq.~\eqref{eq:refl2d}, one can derive the relationship between $\rho_x$ and $\rho_x^{\text{ntt}}$:
\begin{align}
    \label{eq:app-2dprf-rho2rhontt}
    \includegraphics[scale=1.0, valign=c]{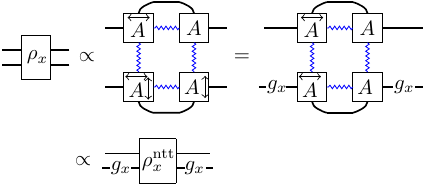}\quad,
\end{align}
where we have used the property of a SWAP-gauge matrix, $g_x g_x = 1$.
Equation~\eqref{eq:app-2dprf-rho2rhontt} indicates that once the isometry $p_x^{\text{ntt}}$ has been determined using the eigenvalue decomposition (EVD) of $\rho_x^{\text{ntt}}$ in Eq.~\eqref{eq:app-2dprf-rhopxntt}, the isometry $p_x$ can be determined by acting $g_x$ on $p_x^{\text{ntt}}$ according to the second equation in Eq.~\eqref{eq:app-2dprf-p2pntt}.
The other equation in Eq.~\eqref{eq:app-2dprf-p2pntt} follows immediately since $g_x g_x = 1$.
\end{proof}

\begin{remark}
    For the same input tensor $A$, if the isometric tensors are related according to Eq.~\eqref{eq:app-2dprf-p2pntt}, the two RG equations in Eq.~\eqref{eq:app-2dprf-RGeq} and~\eqref{eq:app-2dprf-RGeqntt} produce the same coarse-grained tensor $\mathcal{T}_{\text{2D}}(A) = \mathcal{T}_{\text{ntt}}(A)$.
\end{remark}

\begin{remark}
    In numerical calculations, Theorem~\ref{theo:2Dp2pntt} indicates a way to roll back from a scheme with transposition trick to the corresponding scheme without it.
    Specifically, after generating a tensor RG flow of $A, p_x, p_y, g_x, g_y$ using the RG equation in Eq.~\eqref{eq:app-2dprf-RGeq} near a critical fixed point, one can immediately obtain the corresponding $p_k^{\text{ntt}}, k=x,y$ for the RG equation without the transposition trick in Eq.~\eqref{eq:app-2dprf-RGeqntt}.
    With the RG flows of $A, p_x^{\text{ntt}}, p_y^{\text{ntt}}$, one can build the linearization $\mathcal{R}_{\text{2D}}^{\text{ntt}}$ of the tensor RG equation $\mathcal{T}_{\text{ntt}}$ in Eq.~\eqref{eq:app-2dprf-RGeqntt} without the transposition trick according to Ref.~\cite{Lyu:Xu:Kawashima:2021}.
\end{remark}

% Step 2 of the 2D proof
Next, we write down the explicit expression of the linear map $\mathcal{R}_{\text{2D}}^{(c_x, c_y)}$ constructed in Eqs.~\eqref{eq:R2Ddecomp} to~\eqref{eq:2dlinLx} when the EF is turned off and without the projective truncation for inner legs:
\begin{widetext}
\begin{align}
    \label{eq:app-2dprf-linRG}
    \includegraphics[scale=0.9, valign=c]{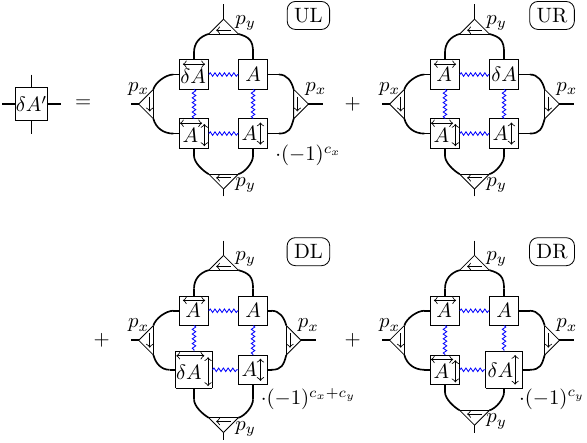}
    =
    \mathcal{R}_{\text{2D}}^{(c_x, c_y)}(\delta A).
\end{align}
\end{widetext}
We label each of the four terms in Eq.~\eqref{eq:app-2dprf-linRG} using two letters. 
The first letter, U or D, indicates ``up'' or ``down'', while the second letter, R or L, indicates ``right'' or ``left''. 
It is easy to see that the two terms, UR and DR, come from the first term in Eq.~\eqref{eq:2dlinLx} when Eqs.~\eqref{eq:2dlinF} and~\eqref{eq:2dlinLy} are plugged into it, while the other two terms, UL and DL, come from the second term in Eq.~\eqref{eq:2dlinLx}.

\begin{remark}
    When $c_x=c_y=0$, the constructed linear map $\mathcal{R}_{\text{2D}}^{(c_x, c_y)}$ is the linearization of the tensor RG map $\mathcal{T}_{\text{2D}}$ with the transposition trick in Eq.~\eqref{eq:app-2dprf-RGeq}.
    The linearization of $\mathcal{T}_{\text{2D}}$ can only have eigenvectors that live in the space $\Tsp{4}{(0, 0)}$ because the lattice-reflection symmetry is imposed in this RG map and the range of $\mathcal{T}_{\text{2D}}$ is $\Tsp{4}{(0,0)}$.
\end{remark}

\begin{remark}
    At this point, other cases where $c_x \neq 0$ or $c_y \neq 0$ of $\mathcal{R}_{\text{2D}}^{(c_x, c_y)}$ can be treated as linearization auxiliary to the $c_x=c_y=0$ case.
    As will be proven below, an auxiliary map $\mathcal{R}_{\text{2D}}^{(c_x, c_y)}$ contains eigenvectors in the subspace $\Tsp{4}{(c_x, c_y)}$.
\end{remark}

\begin{remark}
    The domain of definition of $\mathcal{R}_{\text{2D}}^{(c_x, c_y)}$ is the linear space of all $4$-leg tensors $\mathfrak{T}_4$.
\end{remark}

\begin{theorem}
    \label{theo:2DRrange}
    When evaluated at a fixed-point tensor $A=A_*$ that has reflection symmetry as is defined in Eq.~\eqref{eq:refl2d}, the range of the linear map $\mathcal{R}_{\text{2D}}^{(c_x, c_y)}$ is $\Tsp{4}{(c_x, c_y)}$ defined in Eq.~\eqref{eq:T4cxcyDef} with the SWAP-gauge matrices for $A_*$:
    \begin{subequations}
    \begin{align}
        \label{eq:app-2dprf-range2DR}
        \mathcal{R}_{\text{2D}}^{(c_x, c_y)}:
        \mathfrak{T}_4 \to \Tsp{4}{(c_x, c_y)},
        \text{ or}\\
        \forall \delta A \in \mathfrak{T}_4,
        \mathcal{R}_{\text{2D}}^{(c_x, c_y)}(\delta A) \in \Tsp{4}{(c_x, c_y)}.
    \end{align}
    \end{subequations}
\end{theorem}

\begin{proof}
    Denote the SWAP-gauge matrices of $A=A_*$ as $g_x, g_y$, which indicates the following property of $p_x, p_y$ according to our discussion in~\autoref{subsec:symIT} and footnote~\ref{fn:EVDpSWAP},
\begin{align}
    \label{eq:app-2dprf-psym}
    \includegraphics[scale=0.8, valign=c]{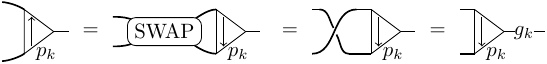},
\end{align}
for $k=x,y$.
It suffices to show the symmetry property of $\delta A'$ in Eq.~\eqref{eq:app-2dprf-linRG} for transposition of its two legs in $x$ direction, since the proof in another direction is exactly the same.
When the two $x$ legs of $\delta A$ is transposed, the sum of the two terms ``UL'' and ``UR'' picks up an overall factor $(-1)^{c_x}$, with the SWAP-gauge $g_y$ acting on their $y$ legs:
\begin{widetext}
\begin{align}
    \label{eq:app-2dprf-linRGrange}
    \includegraphics[scale=1.0, valign=c]{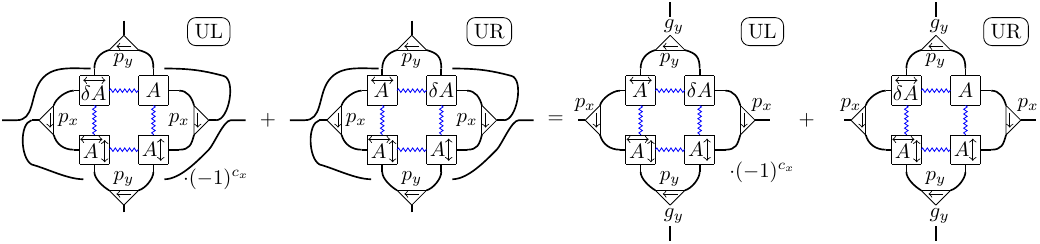},
\end{align}
\end{widetext}
where we have used the symmetry property of $p_y$ in Eq.~\eqref{eq:app-2dprf-psym}.
The same thing happens for the sum of the other two terms ``DL'' and ``DR''.
Therefore, we have proven that
\begin{align}
    \label{eq:app-2dprf-dAsym}
    \includegraphics[scale=1.0, valign=c]{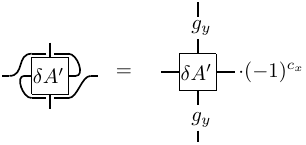}.
\end{align}
\end{proof}

\begin{corollary}
    \label{corol:specRrange}
    Like the 1D toy example, we denote the collection of all nonzero eigenvalues whose eigenvectors live in $\Tsp{4}{(c_x, c_y)}$ as $\spec^{(c_x,c_y)}$.
    Theorem~\ref{theo:2DRrange} indicates that
    \begin{align}
        \label{eq:app-2dprf-specRrange}
        \spec{\mathcal{R}_{\text{2D}}^{(c_x, c_y)}}
        =
        \spec^{(c_x, c_y)}{\mathcal{R}_{\text{2D}}^{(c_x, c_y)}}.
    \end{align}
\end{corollary}

Finally, let us write down the linearization without the transposition trick.
According to Ref.~\cite{Lyu:Xu:Kawashima:2021}, the linearization $\mathcal{R}_{\text{2D}}^{\text{ntt}}$ of the tensor RG map $\mathcal{T}_{\text{ntt}}$ in Eq.~\eqref{eq:app-2dprf-RGeqntt} is
\begin{widetext}
\begin{align}
    \label{eq:app-2dprf-linRGntt}
    \includegraphics[scale=1.0, valign=c]{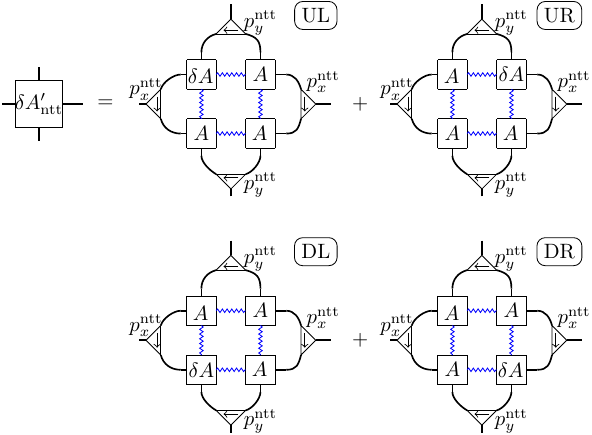}
    =
    \mathcal{R}_{\text{2D}}^{\text{ntt}}(\delta A).
\end{align}
\end{widetext}

\begin{theorem}
    \label{theo:2DR-Rntt}
    When evaluated at the same tensor $A$ that has the lattice-reflection symmetry defined in Eq.~\eqref{eq:refl2d} with SWAP-gauge matrices $g_x$ and $g_y$, the constructed linear map $\mathcal{R}_{\text{2D}}^{(c_x, c_y)}$ in Eq.~\eqref{eq:app-2dprf-linRG} is the same as the linearization $\mathcal{R}_{\text{2D}}^{\text{ntt}}$ in Eq.~\eqref{eq:app-2dprf-linRGntt} in the subspace $\Tsp{4}{(c_x, c_y)}$ defined in Eq.~\eqref{eq:T4cxcyDef}:
    \begin{align}
        \label{eq:app-linMapinT4cxcy}
        \forall \delta A \in \Tsp{4}{(c_x, c_y)},
        \mathcal{R}_{\text{2D}}^{(c_x, c_y)}(\delta A)
        =
        \mathcal{R}_{\text{2D}}^{\text{ntt}}(\delta A).
    \end{align}
\end{theorem}

\begin{proof}
    Using the symmetry of the tensor $A$ in Eq.~\eqref{eq:refl2d} and $\delta A$ in Eq.~\eqref{eq:T4cxcyDef}, along with the relationship between $p_k^{\text{ntt}}$ and $p_k$ for $k=x,y$ in Theorem~\ref{theo:2Dp2pntt}, it is clear that each of the four terms in Eq.~\eqref{eq:app-2dprf-linRG} is the same as the corresponding term in Eq.~\eqref{eq:app-2dprf-linRGntt}.
    Take as an example the term called ``DR'' in Eq.~\eqref{eq:app-2dprf-linRG}:
    \begin{widetext}
    \begin{align}
        \label{eq:app-2dprf-R2Rntt}
        \includegraphics[scale=1.0, valign=c]{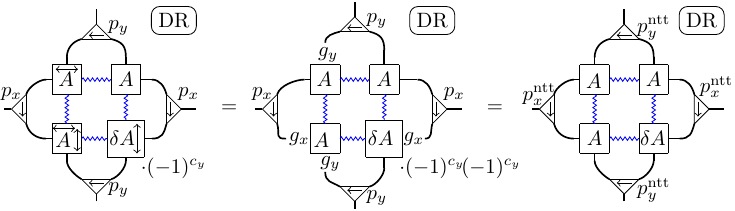},
    \end{align}
    \end{widetext}
    where in the first equality the symmetry properties of $A$ and $\delta A$ are used, and in the second equality, Theorem~\ref{theo:2Dp2pntt} is used.
\end{proof}

\begin{corollary}
    The subspace $\Tsp{4}{(c_x,c_y)}$ is an invariant subspace of $\mathcal{R}_{\text{2D}}^{\text{ntt}}$ since the range of $\mathcal{R}_{\text{2D}}^{(c_x, c_y)}$ is $\Tsp{4}{(c_x,c_y)}$ according to Theorem~\ref{theo:2DRrange}.
\end{corollary}

\begin{corollary}
    \label{corol:specInT4cxcy}
    In the invariant subspace $\Tsp{4}{(c_x,c_y)}$, the eigenvalue spectrum of  $\mathcal{R}_{\text{2D}}^{(c_x, c_y)}$ is the same as that of $\mathcal{R}_{\text{2D}}^{\text{ntt}}$:
    \begin{align}
        \label{eq:app-2dprf-specInT4cxcy}
        \spec^{(c_x,c_y)}{\mathcal{R}_{\text{2D}}^{(c_x, c_y)}}
        =
        \spec^{(c_x,c_y)}{\mathcal{R}_{\text{2D}}^{\text{ntt}}}.
    \end{align}
\end{corollary}

\begin{theorem}
    \label{theo:keyres}
    Corollary~\ref{corol:specRrange} and Corollary~\ref{corol:specInT4cxcy} together indicate the relationship between $\mathcal{R}_{\text{2D}}^{(c_x, c_y)}$ in Eq.~\eqref{eq:app-2dprf-linRG} and $\mathcal{R}_{\text{2D}}^{\text{ntt}}$ in Eq.~\eqref{eq:app-2dprf-linRGntt}:
    \begin{align}
        \label{eq:app-2dprf-result}
        \spec{\mathcal{R}_{\text{2D}}^{(c_x, c_y)}}
        =
        \spec^{(c_x,c_y)}{\mathcal{R}_{\text{2D}}^{\text{ntt}}}.
    \end{align}
\end{theorem}

To conclude this appendix, let us sketch what happens to the proof in this appendix when the EF is turned on and a pair of isometry $p_i$ is inserted into the inner legs of the $2 \times 2$ block.
The essential change is for the tensor RG equations in Eq.~\eqref{eq:app-2dprf-RGeq} and~\eqref{eq:app-2dprf-RGeqntt}, as well as Theorem~\ref{theo:2Dp2pntt}.
Now, the same set of $p_x, p_y$ can be used for both RG equations, but the inner isometry $p_i$ and the filtering matrices $s_x, s_y$ are different.
$p_i$ and $p_i^{\text{ntt}}$ are related to each other in the same way as Eq.~\eqref{eq:app-2dprf-p2pntt} using the SWAP-gauge matrix $g_x$.
For the filtering matrices, without the transposition trick, four filtering matrices $s_{x0}^{\text{ntt}}, s_{x1}^{\text{ntt}}, s_{y0}^{\text{ntt}}, s_{y1}^{\text{ntt}}$ will be needed.
The tensor RG equation without the transposition trick will look like the first equality of Eq.~\eqref{eq:2dEFRGeq1} without the tensor transposition, where $s_{k0}^{\text{ntt}}, s_{k1}^{\text{ntt}}$ act on the first and second incoming leg\footnote{
    Recall that the order of the incoming leg of a 2-to-1 isometric tensor is indicated by the arrow in its tensor-network diagram.
}
of $p_k$ for $k=x,y$.
These four filtering matrices are related to the two filtering matrices $s_x, s_y$ according to $s_{k0}^{\text{ntt}} = s_{k}$ and $s_{k1}^{\text{ntt}} = g_k s_{k}$ for $k=x,y$, where the SWAP-gauge matrix $g_k$ acts on the black-solid leg of $s_k$ (see Eq.~\eqref{eq:2dEFRGeq1}).

\section{Choice of the bond dimension and the RG step for linearizing the RG map\label{app:choicechin}}
In this appendix, we explain the choice of the bond dimension $\chi$ and the RG step $n$ at which the RG map is linearized in the numerical demonstration in~\autoref{sec:numdemo}.

Understood in the Wilsonian RG framework, the bond dimension $\chi$ is directly related to the number of coupling constant retained in an RG map.
For an apt RG transformation, increasing $\chi$ is expected to improve the accuracy of the estimates of the scaling dimensions.
However, there is no clear understanding of the this truncation effect when the TNRG is put into Wilsonian RG framework, since there is still no a priori criterion to tell whether an TNRG is apt unambiguously.
There is numerical evidence~\cite{Evenbly:Vidal:2015} showing that a 2D TNRG enhanced by EF can be an apt RG map.

In the 2D demonstration, we observe a general trend of more accurate estimates of scaling dimension when $\chi$ increases.
We display the $\chi=36$ calculation since it is larger than previous studies~\cite{Lyu:Xu:Kawashima:2021,Guo:Wei:2024} and we want to check whether more scaling dimensions larger than 2 can be resolved.
In the 3D demonstration, $\chi=6$ is chosen for the convenience of the demonstration; the calculation is fast so it is easy to reproduce using the codes published in Ref.~\cite{Lyu:algo:refl}. 
A more detailed study of the $\chi$ dependency has been discussed in Ref.~\cite{Lyu:Kawashima:2024}.

As for the choice of RG step $n$ at which the RG map is linearized, it should be chosen in such a way that the tensor has flowed near the critical fixed point; more concretely, the difference of the two tensors in adjacent RG steps should be small enough.
In both 2D and 3D demonstration, $n=5$ is a reasonable choice according to this criterion.

% The below uses bibtex to construct the bibliography.

% \bibliographystyle{utphys}    % This is a style file for the bibliography. This one does a reasonably good job.
% \bibliographystyle{apsrev4-2}

\bibliography{references}     % This is a reference to "references.bib".

\end{document}